\documentclass[runningheads]{llncs}
\usepackage[T1]{fontenc}
\usepackage{graphicx}

\usepackage{url}
\usepackage[hidelinks,breaklinks=true]{hyperref}
\usepackage{breakcites}
\usepackage[depth=4]{bookmark}
\usepackage{amssymb,amsmath,amsfonts}
\usepackage{stmaryrd}
\usepackage{mathtools}
\usepackage{galois}
\usepackage{xspace}
\usepackage{algorithmic, listings}
\usepackage{array}
\usepackage{orcidlink}             %
\usepackage{enumitem}
\usepackage{longtable, multirow, multicol, makecell}

\usepackage{trfrac} %

\usepackage{tikz-cd}
\usetikzlibrary{calc,decorations.pathmorphing}
\usetikzlibrary{arrows.meta,positioning,shapes,shadows}

\usepackage{changepage} %
\usepackage{framed} %

\usepackage{color}

\usepackage[misc,geometry]{ifsym}  %
\newcommand{\corresauthor}{$^{(\textrm{\Letter})}$} %

\newcommand{\figpath}{./Figures}
\newcommand{\bibpath}{../../../bibtex/clip}

\newcommand{\rr}{\mathbb{R}}
\newcommand{\qq}{\mathbb{Q}}
\newcommand{\zz}{\mathbb{Z}}
\newcommand{\nn}{\mathbb{N}}
\newcommand{\dd}{\mathcal{D}}
\newcommand{\pp}{\mathcal{P}}

\newcommand{\rri}{{\overline{\mathbb{R}}}}
\newcommand{\vn}{\vec{n}}
\newcommand{\va}{\vec{\alpha}}
\newcommand{\vy}{\vec{y}}
\newcommand{\vx}{\vec{x}}
\newcommand{\vyi}{\vec{y_{i}}}
\newcommand{\vyip}{\vec{y'_{i}}} %
\newcommand{\vysol}{\vec{y_{sol}}}
\newcommand{\vphib}{\vec{b}'}
\newcommand{\vb}{\vec{b}}
\newcommand{\vzero}{\vec{0}} %
\newcommand{\ymat}{\mathcal{Y}}
\newcommand{\ypmat}{\mathcal{Y}'} %

\newcommand{\ciao}{\texttt{Ciao}\xspace}
\newcommand{\ciaopp}{\texttt{CiaoPP}\xspace}

\newcommand{\lfp}{\mathrm{lfp}\,}
\newcommand{\gfp}{\mathrm{gfp}\,}
\newcommand{\fsol}{f_{sol}}
\newcommand{\Fix}{\mathrm{Fix}}
\newcommand{\PreFix}{\mathrm{Prefp}}
\newcommand{\PostFix}{\mathrm{Postfp}}

\NewCommandCopy{\rawPhi}{\Phi}
\renewcommand{\Phi}{\mathrm{\rawPhi}}
\NewCommandCopy{\rawPsi}{\Psi}
\renewcommand{\Psi}{\mathrm{\rawPsi}}
\NewCommandCopy{\rawDelta}{\Delta}
\renewcommand{\Delta}{\mathrm{\rawDelta}}

\newcommand{\Phil}{\Phi_{\mathrm{lin}}}

\newcommand{\sgn}{\mathrm{sgn}}
\newcommand{\id}{\mathrm{id}}
\newcommand{\ite}{\mathrm{ite}}

\newcommand{\Vect}{\mathrm{Vec}}
\newcommand{\Card}{\mathrm{Card}}

\newcommand{\dimdom}{d}
\newcommand{\inputset}{\mathcal{I}}
\newcommand{\trainingset}{\mathcal{T}}
\newcommand{\basefuns}{\mathcal{F}}
\newcommand{\candf}{\hat{f}}
\newcommand{\modelspace}{\mathcal{M}} %
\newcommand{\numfeatures}{p} %
\newcommand{\numtraining}{n} %

\newcommand{\sexact}{\checkmark}
\newcommand{\sexactnum}{\widetilde{\checkmark}}
\newcommand{\ssim}{\sim}
\newcommand{\ssimnum}{\approx}
\newcommand{\sTheta}{\mathrm{\Theta}}

\newcommand{\sbadlb}{\mathrm{\Omega}}
\newcommand{\sbadub}{\mathrm{O}}
\newcommand{\snone}{\texttt{None}}
\newcommand{\snotabound}{\times}

\newcommand{\rO}{\mathrm{O}} %

\newcommand{\Prog}{\texttt{Prog}}
\newcommand{\States}{\mathrm{\Sigma}}
\newcommand{\Labels}{\mathcal{L}}

\newcommand{\Env}{\texttt{Env}}
\newcommand{\Mem}{\texttt{Mem}}
\newcommand{\Vars}{\mathcal{V}}
\newcommand{\Vals}{\texttt{Data}}

\newcommand{\CostModel}{\texttt{CostModel}}
\newcommand{\sem}[1]{\llbracket#1\rrbracket}
\newcommand{\semc}[1]{\sem{#1}_c}
\newcommand{\semio}[1]{\sem{#1}_{io}}
\newcommand{\semic}[1]{\sem{#1}_{ic}}
\newcommand{\semioc}[1]{\sem{#1}_{ioc}}
\newcommand{\traces}{\mathcal{T}}

\newcommand{\ttc}{\mathcal{T}^c}
\newcommand{\ttcf}{\mathcal{T}^c_\mathrm{fin}}
\newcommand{\alphaioc}{\alpha_{ioc}}

\newcommand{\predicate}{\texttt{p}}
\newcommand{\predicatebis}{\texttt{q}}
\newcommand{\VarsInput}{{\Vars_{i}}}
\newcommand{\VarsOutput}{{\Vars_{o}}}

\newcommand{\Size}{\texttt{Size}}
\newcommand{\intervals}{\mathcal{I}}

\newcommand{\alphafun}{\overset{\rightarrow}{\alpha}}
\newcommand{\gammafun}{\overset{\rightarrow}{\gamma}}
\newcommand{\alphaB}{\alpha_\boundaries}
\newcommand{\gammaB}{\gamma_\boundaries}

\newcommand{\eqleq}{\sqsubseteq} %

\newcommand{\concretedom}{\mathfrak{D}}
\newcommand{\abstractdom}{\concretedom^\sharp}

\newcommand{\Eqs}{\texttt{Eqs}}
\newcommand{\EqsFull}{\texttt{EqsM}}
\newcommand{\EqsRestr}{\texttt{EqsA}}
\newcommand{\MaxOp}{{\textrm{Max}}}
\newcommand{\MinOp}{{\textrm{Min}}}
\newcommand{\ChoiceRestrOp}{{\textrm{Choice}}}

\newcommand{\SumOp}{{\textrm{Sum}}}
\newcommand{\ScaleOp}{{\textrm{Scale}}}
\newcommand{\ComposeOp}{{\textrm{Comp}}}
\newcommand{\PrecompOp}{{\textrm{Rec-Call}}}
\newcommand{\GroundOp}{{\textrm{Base-Case}}}

\newcommand{\rkfun}{r}

\newcommand{\Brs}{\texttt{Brs}}
\newcommand{\Bb}{\texttt{Bb}}
\newcommand{\Nrs}{\texttt{Nrs}}
\newcommand{\Nb}{\texttt{Nb}}

\newcommand{\hcir}{HC IR\xspace}

\newcommand{\bb}{\mathbb{B}}
\newcommand{\boundaries}{\mathfrak{B}}

\newcommand{\benchid}[2]{\ifthenelse{\equal{\nameorkey}{name}}{#1}{#2}}
\newcommand{\nameorkey}{key}
\newcommand{\catwithnoise}{with\_noise}
\newcommand{\runningexname}{runn\_ex} %
\newcommand{\runningexkey}{1}
\newcommand{\runningex}{\benchid{\runningexname}{\runningexkey}}
\newcommand{\noiseddtwotwoname}{noisy2} %
\newcommand{\noiseddtwotwokey}{2}
\newcommand{\noiseddtwotwo}{\benchid{\noiseddtwotwoname}{\noiseddtwotwokey}}
\newcommand{\fillandfreename}{ffree} %
\newcommand{\fillandfreekey}{3}
\newcommand{\fillandfree}{\benchid{\fillandfreename}{\fillandfreekey}}
\newcommand{\noiseddthreeonename}{noisy3} %
\newcommand{\noiseddthreeonekey}{4}
\newcommand{\noiseddthreeone}{\benchid{\noiseddthreeonename}{\noiseddthreeonekey}}
\newcommand{\catnonmonot}{nonmonot}
\newcommand{\incronename}{incr1} %
\newcommand{\incronekey}{5}
\newcommand{\incrone}{\benchid{\incronename}{\incronekey}}
\newcommand{\noisystrtonename}{noisy1} %
\newcommand{\noisystrtonekey}{6}
\newcommand{\noisystrtone}{\benchid{\noisystrtonename}{\noisystrtonekey}}
\newcommand{\memorythreename}{memory} %
\newcommand{\memorythreekey}{7}
\newcommand{\memorythree}{\benchid{\memorythreename}{\memorythreekey}}
\newcommand{\multiphaseonename}{mphase} %
\newcommand{\multiphaseonekey}{8}
\newcommand{\multiphaseone}{\benchid{\multiphaseonename}{\multiphaseonekey}}
\newcommand{\catdivandconq}{div\_and\_conq}
\newcommand{\binsearchname}{b\_search}  %
\newcommand{\binsearchkey}{9}
\newcommand{\binsearch}{\benchid{\binsearchname}{\binsearchkey}}
\newcommand{\qsortonename}{qsort1} %
\newcommand{\qsortonekey}{10}
\newcommand{\qsortone}{\benchid{\qsortonename}{\qsortonekey}}
\newcommand{\divandconqmultioneonename}{divcon\_m} %
\newcommand{\divandconqmultioneonekey}{11}
\newcommand{\divandconqmultioneone}{\benchid{\divandconqmultioneonename}{\divandconqmultioneonekey}}
\newcommand{\catsuperpoly}{superpoly}
\newcommand{\exponename}{exp1}
\newcommand{\exponekey}{12}
\newcommand{\expone}{\benchid{\exponename}{\exponekey}}
\newcommand{\exptwoname}{exp2}
\newcommand{\exptwokey}{13}
\newcommand{\exptwo}{\benchid{\exptwoname}{\exptwokey}}
\newcommand{\factname}{fact}
\newcommand{\factkey}{14}
\newcommand{\fact}{\benchid{\factname}{\factkey}}
\newcommand{\catmisc}{misc} %
\newcommand{\basiconename}{basic1}
\newcommand{\basiconekey}{15}
\newcommand{\basicone}{\benchid{\basiconename}{\basiconekey}}
\newcommand{\sumoscname}{sum-osc}
\newcommand{\sumosckey}{16}
\newcommand{\sumosc}{\benchid{\sumoscname}{\sumosckey}}
\newcommand{\divname}{div}
\newcommand{\divkey}{17}
\newcommand{\divben}{\benchid{\divname}{\divkey}}
\newcommand{\nestedname}{nested}
\newcommand{\nestedkey}{18}
\newcommand{\nested}{\benchid{\nestedname}{\nestedkey}}
\newcommand{\openzipname}{open-zip}
\newcommand{\openzipkey}{19}
\newcommand{\openzip}{\benchid{\openzipname}{\openzipkey}}
\newcommand{\mergename}{merge}
\newcommand{\mergekey}{20}
\newcommand{\mergeben}{\benchid{\mergename}{\mergekey}}
\newcommand{\loopustarjanname}{loopus\_t} %
\newcommand{\loopustarjankey}{21}
\newcommand{\loopustarjan}{\benchid{\loopustarjanname}{\loopustarjankey}}
\newcommand{\prototypename}{Prototype}
\newcommand{\linregonlyname}{LinReg only}
\newcommand{\srname}{SymbReg only}
\newcommand{\ramlname}{RaML}
\newcommand{\cofloconame}{Cofloco}
\newcommand{\mathematicaname}{Mathematica}

\setlist{nolistsep}

\usepackage[colorinlistoftodos,textwidth=2.2cm]{todonotes}
\newcommand{\lrnote}[1]{\todo[backgroundcolor=cyan!20]{{\bf LR:}~#1}}
\newcommand{\plgnote}[1]{\todo[backgroundcolor=orange!20]{{\bf PLG:}~#1}}
\newcommand{\mhnote}[1]{\todo[backgroundcolor=blue!20]{{\bf MH:}~#1}}
\newcommand{\jfmnote}[1]{\todo[backgroundcolor=green!20]{{\bf JM:}~#1}}
\renewcommand{\lrnote}[1]{}
\renewcommand{\plgnote}[1]{}
\renewcommand{\mhnote}[1]{}
\renewcommand{\jfmnote}[1]{}

\begin{document}
\pgfdeclarelayer{marx}
\pgfsetlayers{main,marx}
\providecommand{\cmark}[2][]{%
  \begin{pgfonlayer}{marx}
    \node [nmark] at (c#2#1) {#2};
  \end{pgfonlayer}{marx}
  }
\providecommand{\cmark}[2][]{\relax}

\title{An Order Theory Framework of Recurrence Equations
        for Static Cost Analysis
       -- Dynamic Inference of %
          Non-Linear Inequality Invariants
       }

\titlerunning{An Order Theory of Recurrence Equations for Cost Analysis}
\author{
  Louis Rustenholz\inst{1,3}\corresauthor\orcidlink{0000-0002-1599-2431}
  \and
  Pedro L\'{o}pez-Garc\'{\i}a\inst{2,3}\orcidlink{0000-0002-1092-2071}
  \and
  Jos\'{e} F. Morales\inst{1,3}\orcidlink{0000-0001-9782-8135}
  \and
  Manuel V. Hermenegildo\inst{1,3}\orcidlink{0000-0002-7583-323X}
}
\authorrunning{L. Rustenholz et al.}
%
\institute{
  Universidad Polit\'{e}cnica de Madrid (UPM), Spain
  \and
  Spanish Council for Scientific Research (CSIC), Spain\\
  \email{pedro.lopez@csic.es}
  \and
  IMDEA Software Institute, Spain\\
  \email{\{louis.rustenholz, josef.morales, manuel.hermenegildo\}@imdea.org}
}

\maketitle
\vspace{-2em}
\begin{abstract}

Recurrence
equations have played a central role in static cost analysis, where
they can be viewed as abstractions of programs and used to infer
resource usage information without actually running the programs with
concrete data. Such information is typically represented as functions
of input data sizes.
More generally, recurrence
equations have been increasingly used
to automatically obtain \emph{non-linear} numerical invariants.
However, state-of-the-art recurrence solvers and cost analysers
suffer from
serious limitations when dealing with the (complex) features of
recurrences arising from cost
analyses.
We address this challenge by developing a novel order-theoretical
framework where recurrences are viewed as operators and their
solutions as fixpoints, which allows leveraging powerful
pre/postfixpoint search techniques. We prove useful properties and
provide
principles and insights that enable us to develop
techniques
and combine them to
design new
solvers.
We have also implemented and experimentally evaluated an
optimisation-based instantiation of the proposed approach.
The results are quite promising: our prototype outperforms
state-of-the-art cost analysers and recurrence solvers, and can infer
tight non-linear lower/upper bounds, in a reasonable time, for complex
recurrences representing diverse program behaviours.

\keywords{Static Cost Analysis \and Resource Usage Analysis \and Static Analysis \and Recurrence Equations.}
\vspace{-1em}
\end{abstract}

\section{Introduction} %
\vspace{-1em}

Recurrence equations
have played a central role in static cost
analysis and verification
for decades, having been used in many automated 
tools~\cite{granularity-short,caslog-short,low-bounds-ilps97-short,resource-iclp07-short,ciaopp-sas03-journal-scp-shortest,AlbertAGP11a-short,plai-resources-iclp14-short,gen-staticprofiling-iclp16-short,montoya-phdthesis,resource-verification-tplp18-shortest,LommenGiesl23}
ever since Wegbreit's seminal work in 1975~\cite{Wegbreit75}.
The goal of static cost analysis
is to infer information about the resources used by programs
(e.g. \emph{execution time}, \emph{memory}, \emph{energy consumption})
without actually running them with concrete data.
Since program resource
consumption generally depends on input data values, the inferred
resource usage information is
typically presented as functions of some abstraction of the program
input (e.g. size) and possibly other (environmental) parameters. The resulting
cost functions are
generally desired
to be correct everywhere (in the sense of safe approximations,
i.e. lower/upper bounds).
This is a requirement in some applications, such as
verification of resource-related
specifications~\cite{resource-verification-tplp18-shortest},
automated parallelisation with
guarantees~\cite{granularity-short,DBLP:journals/pacmpl/WestrickFRA24},
security (e.g. detection of side channel attacks), and blockchain
(e.g. smart-contract gas analysis)~\cite{resources-blockchain-sas20-short,gasol-tacas-2020-short}.
Moreover, in these and other applications,
such as automatic program optimisation, asymptotic approximations
(i.e. just complexity orders) are insufficient, as the techniques
employed need functions that yield concrete, accurately approximated
values.
\mhnote{This paragraph could be skipped here if we really need the
  space --otherwise it is useful of course.}
Recent works in cost analysis have come to the realisation that, more deeply,
\emph{recurrences can be viewed as an abstraction of programs},
stripped from information irrelevant to cost,
or simply as programs themselves.
Recurrences can be executed, and we are
interested in overapproximating their semantics, i.e. in finding
bounds on their solutions. Various properties of programs can be
understood by focusing on their recurrence representations.

In a broader program analysis sense, recurrence
equations have been
increasingly used in the last few years
to automatically obtain \emph{non-linear} numerical
invariants~\cite{kovacs-phdthesis,kovacs17-short,kincaid2018,LommenGiesl23,Wang-PRS23-short}.
\plgnote{The following sentence can be removed/summarised if needed.}
Indeed, recurrence solving allows deducing highly non-linear (global)
invariants from simple (local) relations on program variables.

The challenge we address in this paper originates from the limitations
of state-of-the-art recurrence solving and cost analysis approaches.
First, general-purpose techniques are insufficient.
Since the numerical properties relevant to static cost analysis are
highly non-linear, we cannot simply use a general-purpose
abstract domain such as polyhedra, which would not directly support
programs with even only quadratic complexity.
Similarly, Computer Algebra Systems (CAS) typically focus on
monovariate \emph{difference} equations~\cite{Karr81} with very simple
recursive calls but complex %
coefficients, whereas the
opposite occurs in cost analysis, where coefficients are generally
simple but we are faced with multivariate equations, complex recursive
calls, conditional/non-linear/non-deterministic behaviour, etc.
Second, despite significant progress on specialised solvers, they
too do not support the breadth of
behaviours arising from programs.
As reported in~\cite{mlrec-tplp2024-preprint}, tools such as \texttt{RaML},
\texttt{\ciaopp}, %
\texttt{PUBS}, \texttt{Cofloco}, \texttt{KoAT}/\texttt{LoAT},
\texttt{Loopus}, or
\texttt{Duet}~\cite{DBLP:journals/toplas/0002AH12-short,ciaopp-sas03-journal-scp-shortest,AlbertAGP11a-short,montoya-phdthesis,GieslKoat22,loat-descr-22,Sinn17-loopus,kincaid2018}
all present important limitations and are each unable to infer tight
bounds for some classes of programs, that exhibit, e.g.
\lrnote{But this takes more space. Can we compress?}
non-linear recursive, amortised, non-monotonic and/or multiphase
behaviour, even when a simple closed-form solution exists.
Recently, approaches in the spirit of dynamic invariant generation
have been proposed to circumvent the difficulties of the equation
itself and focus instead on its solution~\cite{mlrec-tplp2024-preprint} (syntax versus
semantics), but are currently limited to solutions that can be
expressed \emph{exactly} using some templates, and are not able to
obtain (safe) bounds for equations that do not admit an exact solution
easily expressed in closed form.

Motivated by these limitations, %
we introduce a novel
approach to recurrence solving, which
leverages 
classical results of order theory,
and makes it possible to infer tight non-linear bounds for complex
recurrence equations
arising from
programs that exhibit diverse behaviour (e.g. divide-and-conquer,
multiphase, hardware-related noise, non-monotone resources), that do
not admit simple solutions.
More specifically, our contributions are as follows:
first, we introduce a novel order-theoretical viewpoint on recurrence solving,
that allows us to leverage powerful pre/postfixpoint search
techniques.
Indeed, we introduce a simple but powerful proof
principle: for \emph{monotone} equations, pre/postfixpoints
provide lower/upper bounds on the solution respectively.
Thanks to the theoretical framework we develop, we prove properties
and obtain a new understanding of the bound search problem in
recurrence solving.
We witness underlying algebraic and geometric structures that can
be utilised to design novel, more accurate %
solving methods. %
In particular, we investigate the completeness of our proof
principle. This unveils %
cones of postfixpoints, with rays related to ranking functions for
the equation, that provide strategies to repair and fine-tune
candidate bounds.%

These theoretical contributions enable the development of new and
accurate algorithms. We propose several instantiations of our approach,
and implement one of them as a proof-of-concept, that uses
dynamic invariant generation and constrained optimisation
techniques. %
Using a machine-learning vocabulary, its results suggest that,
when searching for models, enforcing ``local verifiability''
beyond local correctness can improve accuracy.
This paper extends previous work in several directions. In
particular, it improves the methods presented
in~\cite{mlrec-tplp2024-preprint} with the ability to infer \emph{safe}
approximations, i.e. bounds of the solution of difficult equations.
We perform an experimental evaluation on recurrences arising
from programs displaying diverse properties, and show that our
proof-of-concept solver performs better than the state-of-the-art
on various categories of benchmarks.

Finally, additional results are presented %
to facilitate and foster discussion.
The ``recurrences as (abstractions of) programs''
approach to cost analysis is described in a more
order-theoretical language than classical in the literature, and
Galois connections between abstract domains of functions
relevant to our viewpoint on recurrences are introduced, so that
they may be extended to new, non-discrete applications.

The rest of this paper is organised as follows.
Sections~\ref{sec:running-example}--\ref{sec:context} introduce the
setting with an illustrative example (Section~\ref{sec:running-example}),
background information and notations (Section~\ref{sec:preliminaries}), as
well as motivation and context from cost analysis presented in an
order-theoretical language (Section~\ref{sec:context}).
The main theoretical contributions are then described:
Section~\ref{sec:eqs-as-ops} introduces the
equations as operators viewpoint and our proof principle,
Section~\ref{sec:grammars-of-eqs} presents a syntax of equations
and inductively proves properties on them,
after which Section~\ref{sec:cones} investigates robustness and
extensions of our proof principle, unveiling algebraic and geometrical
structures.
We then move on to algorithmic contributions: candidate generation
(Section~\ref{sec:cand-generation}) and repair
(Section~\ref{sec:repair}) strategies are introduced.
In particular, Section~\ref{subsec:optim-approach} describes a
constrained optimisation-based approach, whose implementation and
experimental evaluation are presented in
Section~\ref{sec:implementation-and-evaluation}.
Finally, Section~\ref{sec:related-work} discusses related work, and
Section~\ref{sec:conclusion} summarises conclusions and lines for
future work.
Additionally, Appendix~\ref{sec:A-bounds} sketches a relevant abstract
domain framework with the aim to encourage discussion.

\newpage

\section{Illustrative Example}
\label{sec:running-example}

Consider a program with quadratic time complexity, performing a linear
time action at each step, and additionally interacting with a buffer,
slowly filled over time and flushed when it is full, so that rare but
costly events can occur.
For example,
imagine a program processing a list of initial length $n$, along
with a buffer containing $c$ elements and of maximal allowed size $100$.
At each step, until the list is empty, it goes through all of it (cost
$n$), selects the largest element, pops it, and pushes it to a buffer.
When the buffer is full, it flushes it and prints it to screen (cost
$300$), a rare but costly event. Finally, when the list is empty, all
unprinted elements are freed from the buffer (cost $c$).
Such a program, independently of the programming language it originates from, can
be represented as the following equation.
{\footnotesize
\begin{equation}\label{eq:running-ex-1}
     f(n, c) = \begin{cases}
       f(n-1, 0)   + n + 300 & \text{ if } n>0 \text{ and } c\geq100,\\
       f(n-1, c+1) + n       & \text{ if } n>0 \text{ and } c<100,\\
       c                     & \text{ if } n=0,\\
      \end{cases}
\end{equation}}%
We denote its solution by $\fsol$. If the buffer is initially empty,
the total cost as a function of $n$ is $f_0(n):=\fsol(n, 0)$.
While the quadratic cost component is simply
$\sum_{k\leq n}k=n(n+1)/2$, the noisy behaviour introduced by
the buffer interaction leads to an exact solution that is
difficult to express with a simple closed-form expression.
Current tools do not obtain tight bounds for this example but
coarse overapproximations. For example,
\texttt{\ramlname}~\cite{DBLP:journals/toplas/0002AH12-short} only infers
that
$\forall n,\,\frac{1}{2}n^2+\frac{1}{2}n \leq f_0(n)\leq\frac{1}{2}n^2+\frac{1}{2}n+300n,$ %
i.e. it overapproximates by triggering the bad event at either every
or no iteration. \texttt{\cofloconame}~\cite{montoya-phdthesis}
obtains even coarser overapproximations: for the
recursive case $n\geq 1$, $n \leq f_0(n) \leq 2n^2+300n+100$. Similar
or more serious difficulties can be observed with cost analysers such as
\texttt{PUBS}, \texttt{KoAT}/\texttt{LoAT}, \texttt{\ciaopp}'s builtin
solver, \texttt{Loopus} and \texttt{Duet}. CAS such as \texttt{Sympy},
\texttt{PURRS} or \texttt{Mathematica} simply do not handle such
equations.
The method of~\cite{mlrec-tplp2024-preprint}, purely based on
regression applied to input-output samples, produces intermediate
approximations such as $f_0(n)\approx 0.5 n^2 + 3.435n + 210.3$ for
$n\geq 1$, and hence witnesses that bad events are expected
to occur around once every $101$ steps ($(3.435-0.5)/300\approx
102.2$), but cannot prove this result, since this is not an exact
solution. Thus, it does not obtain any \emph{safe} approximation.

Our novel method views
$(\ref{eq:running-ex-1})$ as the following operator.
{\footnotesize
\begin{equation*}\label{eq:running-ex-operator}
  \begin{aligned}
  \Phi : (\nn^2 \to \rr) &\to (\nn^2 \to \rr)\\
         f &\mapsto (n,c) \mapsto \begin{cases}
       f(n-1, 0)   + n + 300 & \text{ if } n>0 \text{ and } c\geq100,\\
       f(n-1, c+1) + n       & \text{ if } n>0 \text{ and } c<100,\\
       c                     & \text{ if } n=0.
      \end{cases}
  \end{aligned}
\end{equation*}}

With this viewpoint, Theorem~\ref{th:proof-principle} and monotonicity
of $\Phi$ in $f$ make it possible to deduce lower and upper bounds on
$\fsol$ from pre- and postfixpoints of $\Phi$, which can be obtained
by leveraging a variety of classical techniques.
For example, we could discover a simple candidate upper bound
$\candf$, given by $\candf(n,c)=c$ if $n=0$, $n(n+1)/2+3n+3c$ if
$n\neq0$ and $c<100$; and $n(n+1)/2+3n+297$ otherwise. Then, check with a CAS
that $\Phi\candf\leq\candf$, which guarantees that $\fsol\leq\candf$.
In fact, besides other algorithmic ideas and theoretical framework,
we present in this paper a prototype that implements
a simple approach based on constrained optimisation, which obtains the
following bounds for all $n\geq1$.
{\footnotesize
\begin{equation*}\frac{n(n+1)}{2}+\frac{300}{101}n+100 \leq f_0(n) \leq
  \frac{n(n+1)}{2}+\frac{300}{101}n+100\cdot\frac{300}{101}
\end{equation*}}

This models cost precisely, in its quadratic and linear terms, without
simply assuming that the bad event always happens.
In fact, those are the best possible polynomial bounds on the exact
solution $\frac{n(n+1)}{2}+n+299+199\lfloor\frac{n-1}{101}\rfloor$.
More complete information is available in
Table~\ref{table:runnex_results} of Appendix~\ref{sec:benchmark-tables}.

Beyond this simple example, exact solutions can become even harder to
express with simple closed-form expressions, but we can still hope to
obtain \emph{bounds} on the solution, without losing too much
accuracy.
Our approach allows us to do so. Similarly
to~\cite{mlrec-tplp2024-preprint}, it has few restrictions on equation
syntax, focusing instead on its solution, but
unlike~\cite{mlrec-tplp2024-preprint}, it is not limited to guessing
exact solutions, and can search through the full space of (inductive)
bounds, which is much larger.

While our approach supports more difficult benchmarks, and is not
limited to polynomial bounds (in fact, any user-defined template can
be used), this example shows that tight bounds not discovered by the
state-of-the-art methods can be found, even for benchmarks restricted to
linear arithmetic. %

An additional motivation for this example is that, in future work, we
plan to use these techniques to improve interaction between cost/size
abstract domains, and abstract domains for low-level properties such as
cache~\cite{wilhelm-cache-analysis-survey-2016-short}, that allow
obtaining partial information on such ``rare bad events'' relevant to
cost properties. We believe this is a promising direction to bridge
the gap between two families of automated cost analysis techniques:
those that support general programs but typically use very
coarse-grained cost models, with those common in WCET communities that
support fine-grained models but only limited programming patterns,
often requiring user annotations to handle loops and recursion.

\section{Preliminaries}
\label{sec:preliminaries}

\subsection{Notation}

Boldface characters represent vectors. Depending on the context,
$\vzero$ is a zero finite-dimensional vector, a zero matrix, or the
constant function $\vn\mapsto 0$. Given two sets $A$ and $B$, $B^A$
and $A\to B$ both refer to the set of functions from A to B, the
former insisting on its natural vector space structure. $\nn$ and
$\rr_+$ both include $0$. Unless explicitly stated otherwise,
functions are ordered pointwise and vectors are ordered elementwise.
$f$ and $\candf$ both refer to arbitrary functions, but the latter is
used to convey that it is thought of as a candidate bound. $\fsol$
denotes an exact solution. Operators on functions are usually denoted
by $\Phi$ when they are interpreted as recurrence equations, and by
$\Psi$ otherwise. The function $\ite$ is defined by $\ite(\top,x,y)=x$
and $\ite(\bot,x,y)=y$.
Given $f:A\to B$ and $A'\subseteq A$, $f_{|A'}:A'\to B$ is the
restriction of $f$ to $A'$.

\subsection{Order Theory, Prefixpoints and Postfixpoints Conventions}

We recall some classical definitions of order-theory in cases where it may
be helpful to disambiguate. We refer the reader to \cite{davey-order-book,Cousot21-book} for more
background. %

\begin{definition}
  We say that a function $f:(X, \sqsubseteq) \to (Y, \leq)$ is
  \emph{monotone} whenever
  $\forall x,y\in X,\,  x\sqsubseteq y \implies f(x)\leq f(y)$,
  and that it is \emph{antimonotone} when
  $\forall x,y\in X,\,  x\sqsubseteq y \implies f(x)\geq f(y)$.
\end{definition}
\vspace{-0.1em}
We adopt the following convention.
\vspace{-0.1em}
\begin{definition}[Prefixpoints and postfixpoints]\-

  \noindent Given a function $f:(X, \sqsubseteq) \to (X, \sqsubseteq)$
  and $x\in X$, we say that
  \begin{itemize}
    \item $x$ is a \emph{fixpoint} of $f$ whenever $f(x) = x$,
    \item $x$ is a \emph{prefixpoint} of $f$ whenever $x \sqsubseteq f(x)$,
    \item $x$ is a \emph{postfixpoint} of $f$ whenever $f(x) \sqsubseteq x$.
  \end{itemize}
  The sets of fixpoints, prefixpoints, and postfixpoints of $f$ are
  respectively written as $\Fix(f)$, $\PreFix(f)$, and $\PostFix(f)$.
\end{definition}
Note that the opposite convention can be found, reversing prefixpoints and
postfixpoints~\cite{SmythPlotkin82,davey-order-book}.
The convention we choose is more common in abstract interpretation
literature~\cite{Cousot79}
because of its interpretation in the context of ascending/descending
Kleene sequences.
\begin{definition}
  The minimum element of $\Fix(f)$, if it exists, is called the
  \emph{least fixpoint} (lfp) of $f$, written $\lfp f$.
  Similarly, the maximum element of $\Fix(f)$, if it exists, is called the
  \emph{greatest fixpoint} (gfp) of $f$, written $\gfp f$.
\end{definition}

\noindent We will make use of the following classical result.
\begin{theorem}[Knaster-Tarski, \cite{Knaster28,Tarski55}]\-
  \label{th:knaster-tarski}

  \noindent
  Let $(L,\sqsubseteq)$ be a \emph{complete lattice} and $f:L\to L$ be monotone.
  \begin{itemize}
    \item $\Fix(f)$ is a complete lattice. In particular, it is not empty.
    \item $\lfp f = \bigsqcap \PostFix(f)$. In particular,
      $\forall x\in X,\, f(x) \sqsubseteq x \implies \lfp f \sqsubseteq x.$
    \item $\gfp f = \bigsqcup \PreFix(f)$. In particular,
      $\forall x\in X,\, x \sqsubseteq f(x) \implies x \sqsubseteq \gfp f.$
  \end{itemize}
\end{theorem}

\vspace{-1em}
\subsection{Function Comparison}
\label{subsec:prelim-funcomp}

The approach presented in this paper reduces the problem of checking
whether a candidate function is a bound of the (unknown) solution to a
(monotone) equation to the problem of comparing two functions. This
can be formulated as follows: given two functions $f,g:\dd\to\rr$,
where $\dd\subseteq\nn^\dimdom$, we want to decide the
universally quantified formula $\forall\vn\in\dd,\,f(\vn)\leq^{(?)}g(\vn)$.

This is a much simpler, but still non-trivial problem, whose
difficulty obviously depends on the class of functions under
comparison. In general, it is an undecidable problem, remaining
undecidable even when restricting ourselves to global multivariate
polynomial functions with integer
domains~\cite{Mat70}. Nevertheless, complete
algorithms exist for comparing piecewise multivariate polynomials with
real number domains, providing a decidable sufficient condition for
the comparison of polynomials on integers.

The related work on function comparison in the context of cost
analysis~\cite{resource-verification-tplp18-shortest,resource-verif-2012-short,Albert2015-cost-func-comp}
provides sufficient conditions and techniques
for dealing with large classes of functions, including
polynomials and expressions with exponentials, logarithms, summations,
$\max$/$\min$, etc. These methods perform well in practice, despite being
incomplete in general.

If we are only concerned with whether the global inequality $f\leq g$
holds, the techniques employed by the work cited above are also
provided by contemporary general-purpose CAS. Since the function
comparison problem is beyond the scope of this paper, to facilitate
the implementation of our prototype, we chose to simply use
\texttt{Mathematica}~\cite{mathematica-v13-2}, for which we obtained
the best results compared to other off-the-shelf tools.
Further discussion is included
in a longer version of the paper.%
\footnote{Reference to the technical report skipped to respect
  lightweight double-blind process.}
\lrnote{Link to it?}\mhnote{Easiest thing is to make a technical
  report (an entry in clip.bib suffices) and cite it. We can also
  upload it to ArXiV, but that is only worth if of course if we think
  it is ready for people to read. MH: Done to TR.}

\vspace{-1em}
\subsection{Quadratic Programming, Sparse Constrained Linear Regression}
\label{subsec:prelim-qp}

Quadratic Programming~\cite{OptimML-11} (QP) refers to the process of
solving optimisation problems with quadratic objective functions and
linear constraints on the variables and is a generalisation of Linear
Programming (LP).
QP problems, with $n$ variables and $m$ constraints, can be formulated
as the optimisation problem
$\min_{\vx\in\rr^n}\vx^TP\vx + \vec{q}^T\vx$
under the constraint $G\vx\leq\vec{h}$,
where $P$ is a $n\times n$ \emph{symmetric} matrix, $\vec{q}$ is an
$n$-dimensional vector, $G$ is $m\times n$, and $\vec{h}$ is
an $m$-dimensional, all over $\rr$.
In other words, the quadratic objective function
$\vx\mapsto\vx^TP\vx+\vec{q}^T\vx$ is optimised over the convex
polyhedra $\{\vx\,|\,G\vx\leq\vec{h}\}$.
In many applications, including this paper, we are mostly interested
in the case where $P$ is positive definite (or at least semidefinite).
In this case, the problem is convex, efficient\footnote{In particular
when iterative or \texttt{float}-based. Exact methods for rational
inputs exist, but do not scale as well, similarly to the case of
Linear Programming.} algorithms and tools are available, and the
problem is equivalent to \emph{constrained least square}. Hence, QP is
the natural framework for (linearly) \emph{constrained} linear
regression.
We refer the reader to~\cite{Hastie-09a,Hastie-15a} for more background
on linear regression.

As described in~\cite{mlrec-tplp2024-preprint}, (sparse) linear
regression can be used to obtain non-linear invariants, by searching
for models $\hat{f}$ of
\emph{functions} $f:\dd\to\rr$ in an affine model space spanned by
base functions $f_i:\dd\to\rr$, using a set of observations
$\{(\langle\dots,f_i(\vn),\dots\rangle,\,f(\vn)\,|\,\vn\in\inputset\}$
which are produced from finite samples $\inputset\subseteq\dd$, but
the linear relationship is searched between the $f_i$'s and $f$.
In this context, lasso regularisation can be described as a way to
perform feature selection, so that base functions $f_i$ most relevant
to model $f$ are selected.  With this viewpoint, regression may first
be performed with lasso to select a subset of relevant functions, and
then without lasso but with only the selected features to achieve
better precision.
However, constrained linear regression has not been used before in
this context. In this work (Section~\ref{subsec:optim-approach}), we
describe how we use it to generate candidate \emph{bounds} on
solutions to recurrence equations, as well as
candidate \emph{inductive bounds}.

\section{Context and Motivation}
\label{sec:context}

As mentioned earlier, an important motivation for the work presented
in this paper stems from automated static cost analysis and
verification of programs.

In this context, classical operational semantics are decorated using
some cost model,
provided by the user or inferred, which may be either very generic or
contain detailed low-level architectural information.
For such applications it may be \emph{a priori} necessary and relevant
to consider very fine-grained set of states $\States$ and semantics,
in order to take into account low-level information, such as cache
behaviour, which can be very impactful on cost properties.  We
represent this as $\States = \Labels\times\Env\times\Mem$, where
$\Labels$ is the set of control states (program points, etc.),
$\Mem=\Vars\to\Vals$ is the set of memory states (where $\Vars$ is the
set of program variables, and $\Vals$ is the set of possible values,
e.g.  integers, data structures or Herbrand terms), and $\Env$ is the
set of states for the rest of the environment.

Cost models may be seen as a way to add weights\footnote{Typically
taken from $\rr$ or $\nn$, but generalisations are possible.}  to
transitions $\States\times\States$, i.e.
$\CostModel:\States\times\States\to\rr$, or non-deterministically
$\CostModel:\States\times\States\to\pp(\rr)$.
Usual concrete trace semantics $\sem{\Prog}\in\pp(\traces)$, where
$\traces=\States^*\cup\States^\omega$ are (finite or infinite) traces
$\sigma_0\rightarrow \sigma_1 \rightarrow\dots$,
are extended into concrete trace semantics with costs
$\semc{\Prog}\in\pp(\ttc)$ where elements of $\ttc$ are ``weighted
traces'' of
the form
$\sigma_0\overset{c_1}{\rightarrow}\sigma_1\overset{c_2}{\rightarrow}\dots$,
with $\sigma_i\in\States$ and $c_i\in\rr$.
For simplicity, we will only consider finite weighted traces $\ttcf$
with costs in $\rr$, and their
\emph{total
cost},%
\footnote{We could also consider
the \emph{maximal footprint} $\max_{k\leq n}\sum_{i=1}^k c_i$,
e.g. in the case of analysis of memory usage, or to more complex
aggregations of weights, e.g. when analysing non-deterministic logic
programs with tree semantics, or searching for cost hotspots when
analysing parallel programs.} i.e.
the function
$\ttcf \to \rr$,
$\big(\sigma_0\overset{c_1}{\rightarrow}\sigma_1\overset{c_2}{\rightarrow}\dots\overset{c_n}{\rightarrow}\sigma_n\big)\mapsto\sum_{i=1}^n c_i$.
We use this function to abstract concrete trace semantics (with cost)
into input-output semantics (with cost), via
\begin{align*}
  \alphaioc:
  \pp(\ttcf) & \to \pp(\States\times\rr\times\States)\\
  T &\mapsto
   \Big\{\Big(\sigma_0,\sum_{i=1}^n c_i,\sigma_n\Big)\,\Big|\,
   \big(\sigma_0\overset{c_1}{\rightarrow}\sigma_1\overset{c_2}{\rightarrow}\dots\overset{c_n}{\rightarrow}\sigma_n\big)\in T\Big\},
\end{align*}
which can be abstracted further into $\pp(\States\times\States)$ and
$\pp(\States\times\rr)$ by dropping the cost and output components
respectively.

After having defined this low-level cost semantics in a language most
appropriate to describe concrete execution (e.g. ``on bare metal'' --
the low level language is chosen for the precision of the
corresponding cost model), the recurrence-based approach to cost
analysis then suggests translating the program into an intermediate
language more appropriate for analysis.
This can be done for example by a translation
into \emph{Horn clauses},
while preserving/encoding the low level cost semantics presented above
(which may be done automatically,
c.f.~\cite{resource-verification-tplp18-shortest,isa-energy-lopstr13-final-short,isa-vs-llvm-fopara-short,resources-bytecode09-short,decomp-oo-prolog-lopstr07-short}).
This intermediate representation as Horn clauses (\texttt{\hcir}) may
be viewed as a simple language with \emph{pure} functions operating on
symbolic terms, so that the main challenge of program analyses,
recursion, can be dealt with uniformly, as recursive function calls
become the core control structure.
Recurrence equations
then arise naturally as a representation of numerical recursive
programs, or, more generally, of programs abstracted to numerical
programs via \emph{size} abstractions, as explained below.

In this new setting, the ``concrete'' denotational semantics,
decorated with cost, of a function $\predicate$ with input and output
variables $\Vars_i$ and $\Vars_o$, is an element
\begin{align*}
  \semioc{\predicate}
  &\in
  \pp\big((\VarsInput\to\Vals)\times\rr\times(\VarsOutput\to\Vals)\big)\\
  &\simeq
  (\VarsInput\to\Vals)\to\pp\big(\rr\times(\VarsOutput\to\Vals)\big)\\
  &=
  \Vals^\VarsInput\to\pp\big(\rr\times\Vals^\VarsOutput\big),
\end{align*}
where the equivalence is just curryfication
$(A\times B \to 2) \simeq (A \to (B \to 2))$,
and functions are (naturally) ordered pointwise.
Dropping the cost and output components respectively, we then obtain
$\semio{\predicate}:\Vals^\VarsInput\to\pp(\Vals^\VarsOutput)$
and
$\semic{\predicate}:\Vals^\VarsInput\to\pp(\rr)$,
which may be seen as the concrete \emph{data} semantics and concrete
\emph{cost} semantics of
function $\predicate$ respectively.
Notice that to obtain a compositional framework, it is necessary to
take data semantics into account, even if we are only interested in
cost:
in order to compute the cost semantics of a composition
$\predicate\circ\predicatebis$, we need to know the cost semantics of
the outer function, \emph{and the data semantics} of inner functions,
i.e.  we expect to have
$\semic{\predicate\circ\predicatebis}=\semic{\predicate}\circ\semio{\predicatebis}+\semic{\predicatebis}$
(where functions $A\to\pp(B)$ are implicitly extended into
$\pp(A)\to\pp(B)$ in the natural way).

To move into the realm of computability and finitely-representable objects,
more abstractions are introduced.
This is achieved by using \emph{size measures}, which are functions
$m:\Vals\to\Size$ that transform concrete data elements in $\Vals$
into data \emph{sizes} in $\Size$ (representing properties like length
of lists, number of nodes of trees, value of integers,
etc.). Typically $\Size\subseteq\nn$ as in~\cite{caslog-short,granularity-short}, but
more general and refined approaches exist, e.g. extracting tuples of
natural numbers using sized types~\cite{plai-resources-iclp14-short},
where $\Size\subseteq\nn^k$.
Any
\emph{size measure} naturally provides Galois connections
\lrnote{Key point}
$\pp(\Vals)\leftrightarrows\pp(\Size)$,
$\big(\Vals\to\pp(\rr)\big)\leftrightarrows\big(\Size\to\pp(\rr)\big)$
and
$\big(\Vals\to\pp(\Vals)\big)\leftrightarrows\big(\Size\to\pp(\Size)\big)$,
where the last two connections allow
for abstraction of cost and data semantics, respectively.

Note that $\Vals\to\Size$ allows viewing programs operating on $\Vals$,
after size abstractions, as purely numerical programs.
Ideally, size abstractions would be defined, for a particular program, via
some quotient of the semantics with observational (cost) equivalence, so
that as little information on the cost semantics is lost, but behaviour
irrelevant to cost is abstracted away.
In practice, current size analyses use type analyses, heuristics and user
annotations to transform data structures into numerical summaries that
appear relevant to cost properties.
We refer the reader to~\cite{caslog-short,plai-resources-iclp14-short}
for more details.

Hence, at this point, we are left with non-deterministic numerical
programs operating on sizes, which may be viewed as abstractions of
the original program. Any such recursive numerical program defined
with pure functions may also be viewed directly as \emph{a system of}
(non-deterministic) \emph{recurrence equations}.

However, non-deterministic systems of recurrence equations
on cost functions $\nn^\VarsInput\to\pp(\rr)$ and size functions
$\nn^\VarsInput\to\pp(\nn^\VarsOutput)$, arising from the diversity of
numerical programs are far beyond what can be effectively analysed by
recurrence solving methods.
The recurrence
equations corresponding to programs via size abstractions therefore
need to be simplified to be handled.

A first natural step is to focus on obtaining parametric \emph{bounds}
on the cost of functions: numerical sets of sizes and costs (in
$\pp(\nn)$ and $\pp(\rr)$) are abstracted using intervals (or
generalised Kaucher intervals~\cite{Kaucher80} -- we omit these
details for the sake of conciseness).
The problem is further simplified by abstracting away output-output
relations in data semantics, i.e. by looking at each output independently.
We are left with the following interval-valued functions, which
abstract data and cost semantics via size abstraction, and may
equivalently be viewed as \emph{intervals whose endpoints are
functions}.
\begin{align*}
  \texttt{Size}^{(k)}_\predicate:\nn^\VarsInput\to\intervals(\nn)\;\text{ for each }k\in\VarsOutput,
  &&
  \texttt{Cost}_\predicate:\nn^\VarsInput\to\intervals(\rr)
\end{align*}
\noindent Separating lower and upper bounds of the intervals, we obtain
functions $\texttt{Size}^{(k)}_{\predicate,lb}$,
$\texttt{Size}^{(k)}_{\predicate,ub}$,
$\texttt{Cost}_{\predicate, lb}$ and
$\texttt{Cost}_{\predicate, ub}$,
all of which may be viewed as elements of $\nn^\dimdom\to\rri$,
for each function $\predicate$ (of arity $\dimdom$) defined in the source
program.\footnote{In general, $\dimdom$ can be different from the
arity of $\predicate$; for example, when using more refined size
measures such as sized types~\cite{plai-resources-iclp14-short}, which
can represent sizes of structures and substructures at any position
and depth by using tuples of natural numbers. Additionally, some sizes
can be inferred to be irrelevant and thus removed. However, in this
section, we use simpler concepts for the sake of presentation.}
We can then extract (ideal) \emph{deterministic} recurrence equations
on such \texttt{Size} and \texttt{Cost} bound functions, abstracting the
non-deterministic equations
above.
Solving these equations
or finding bounds on their solutions, finally
allows us to obtain overapproximations of the cost semantics
$\semic{\predicate}:\Vals^\VarsInput\to\pp(\rr)$
via
\begin{align*}
  \gamma:\big(\Size^\VarsInput\to\rri\big)\times\big(\Size^\VarsInput\to\rri\big)&\to\big(\Vals^\VarsInput\to\pp(\rr)\big)\\
    (\texttt{Cost}_{lb},\;\texttt{Cost}_{ub})&\mapsto\Big(\vx\mapsto
       \big[\texttt{Cost}_{lb}(m(\vx)),\;\texttt{Cost}_{ub}(m(\vx))\big]\Big),
\end{align*}
where the order $\geq\!\times\!\leq$, with $\geq$ and $\leq$ pointwise, is
used on the left hand side.
In practice, cost/size analysers such as~\cite{caslog-short,granularity-short,plai-resources-iclp14-short} do not necessarily
obtain ideal recurrence equations
on $\texttt{Size}$ and $\texttt{Cost}$
bound functions, but instead focus their effort on setting up safely simplified
equations assumed to be supported by recurrence solvers. %

Without loss of generality, in the rest of the paper, we will only
refer to equations on a single function $f:\dd\to\rr$ with discrete
domain $\dd$, instead of \emph{systems} of equations. Note that any
system of equations on $\{f_i:\dd_i\to\rr\}$ may always be translated
into an equivalent system on $f:\dd\to\rr$ with $\dd=\sqcup \,\dd_i$.
This allows us to simplify the presentation of our novel approach
for recurrence solving, by applying
it to a recurrence equation on a single unknown $f:\dd\to\rr$.  We
will describe
how to obtain (good) bounds on its solution, even if this exact
solution cannot be computed.

\section{The Equations as Operators Viewpoint}
\label{sec:eqs-as-ops}

A simple but effective insight enabling the
techniques presented in
this paper is that recurrence equations may be viewed as higher-order
operators. In our setting, a \emph{recurrence equation} is simply
an equation on a function $f:\dd\to\rr$, that can be written
as: $$\forall \vn\in\dd,\;f(\vn) = \Phi(f)(\vn),$$ where
the higher-order operator $\Phi:(\dd\to\rr)\to(\dd\to\rr)$ is used to
define $f(\vn)$ in terms of other values of $f$.

We say that $\Phi$ is \emph{the operator corresponding to such a
functional equation}. Additionally, the operator may be viewed as a
\emph{definition} of the equation.  A function $f$ is then a
\emph{solution} to the equation if and only if it is a \emph{fixpoint}
of the $\Phi$ operator.
A key observation of this paper is that, for a large class of
recurrence equations, the bound search problem can be reduced to the
search of pre/postfixpoints for a \emph{monotone} operator $\Phi$
corresponding to
an equation.
This section is devoted to this result, proved as
Theorem~\ref{th:proof-principle}.
We also discuss the intuition
behind the equation monotonicity assumption in
Section~\ref{subsec:monotonicity-discussion} -- importantly, this does
not require the solutions themselves to be monotone. %
Fig.~\ref{fig:gcd-ex} illustrates Theorem~\ref{th:proof-principle}
and operator dynamics: a linear upper bound candidate $\candf$ is
compared against a non-trivial program computing the greatest common
divisor, and the ``blanket'' $\candf$ slowly falls onto the \emph{a
priori} unknown solution.
\begin{figure}[t]
  \centering
  \begin{minipage}[c]{0.45\linewidth}
    \includegraphics[width=4.5cm,height=4.5cm]{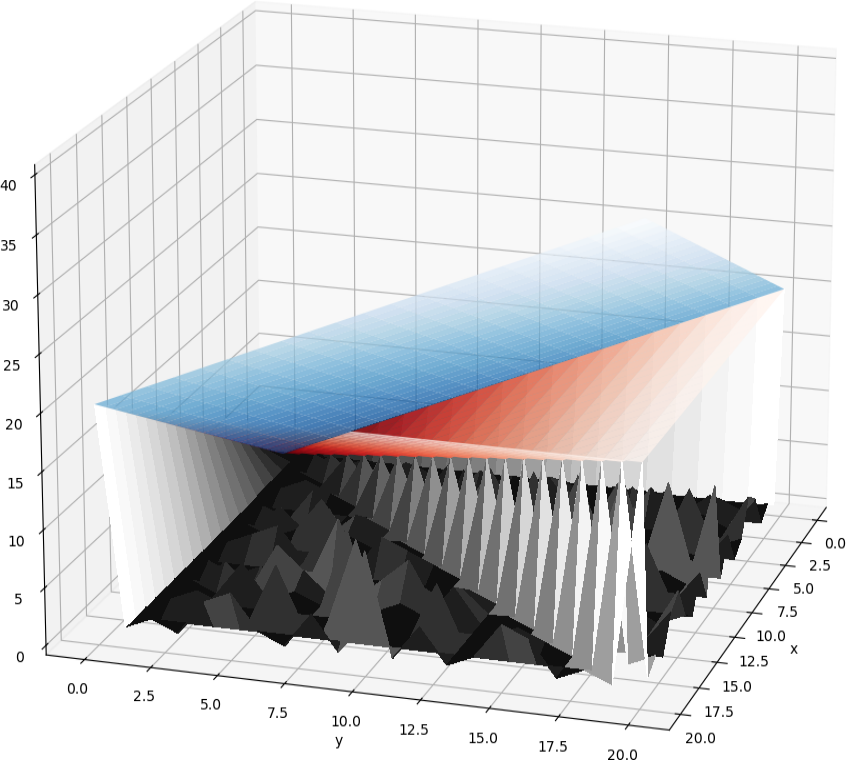}
  \end{minipage}
  \begin{minipage}[c]{0.45\linewidth}
    \includegraphics[width=4.5cm,height=4.5cm]{\figpath/gcd1\_1}
  \end{minipage}
  \caption{\footnotesize
  \textbf{Candidate and operator dynamics.}
  The operator $\Phi$ of a program computing the \texttt{gcd}
  interacts with a candidate upper bound $\candf(x,y)=x+y$.
  Plots represent the solution $\fsol$ (black), $\candf$ (blue left),
  $\Phi\candf$ (red left, blue right), and $\Phi^2\candf$ (red right).
  We observe that $\Phi\candf \leq\candf$, which is enough to prove that
  $\candf$ is an upper bound on $\fsol$. Iterates $\Phi^k\candf$
  progressively get closer to $\fsol$, but have more complex
  expressions.
  \label{fig:gcd-ex}
  }
  \vspace{-1em}
\end{figure}

\subsection{Equations as Operators}

The equations that arise in static cost analysis typically
correspond to \emph{monotone} operators for the pointwise order, hence
allowing the use of all insights and techniques coming from order
theory to study these equations and obtain bounds on their solutions.
\emph{Solutions} to equations naturally correspond to \emph{fixpoints}
on operators, %
and the idea of the ``equation-as-operator viewpoint'' is that we
can take this operator seriously: we can study the values taken by the
operator \emph{outside} of its fixpoints, and obtain information on
the solutions without having to compute them explicitly, which may be
impossible or impractical.
The equations we consider are typically defined piecewise as:
{\footnotesize
\begin{equation}\label{eq:piecewise-generic}
\begin{aligned}
  \Phi : (\dd \to \rr) &\to (\dd \to \rr)\\
         f &\mapsto \vn \mapsto \begin{cases}
           \;\dots&\\[-2pt]
           e_i(f,\vn) &\text{if } \varphi_i(f,\vn)\\[-2pt]
           \;\dots&
         \end{cases},
\end{aligned}
\end{equation}}%
as was illustrated in Section~\ref{sec:running-example}.
In our context, we usually take the domain $\dd$ to be discrete, e.g.
$\dd\subseteq\nn^\dimdom$ or $\dd\subseteq\zz^\dimdom$ with
$\dimdom\geq 1$.
In many cases, we will further suppose that the conditions $\varphi_i$
are independent of $f$ to simplify the problem. However, we plan to
take more general $\varphi_i$ into account in future work.

We bring order theory into the picture by equipping the set of
functions $\dd\to\rr$ with the \emph{pointwise order} derived from the
usual order on $\rr$: given two functions $f:\dd\to\rr$ and
$g:\dd\to\rr$, $f \leq_{\dd\to\rr} g$ if and only if for all inputs
$\vn\in\dd$ we have $f(\vn)\leq_\rr g(\vn)$. This provides the
expected definition of upper and lower bounds: if $f \leq_{\dd\to\rr}
g$, we say that $g$ is an upper bound of $f$ (or equivalently, that
$f$ is a lower bound of $g$).
In order to work with complete lattices, we extend $\rr$ with a top and a
bottom element, obtaining $\rri:=\rr\cup\{-\infty,+\infty\}$, where
$\leq_\rri$ is defined as expected. $\dd\to\rri$ equipped with the
pointwise order is then also a complete lattice, with meet and join
defined pointwise.

The rest of this paper focuses on the following class of recurrence
equations.
\begin{definition}[Monotone equation]
  \label{defi:monotone-equation}
  A \emph{monotone recurrence equation}, or simply \emph{monotone
  equation}, is a monotone function
  $$\Phi:(\dd\to\rri, \leq) \to (\dd\to\rri, \leq)$$
  We say that a function $f:\dd\to\rri$ is \emph{a} solution to the
  equation whenever $f$ is a fixpoint of $\Phi$.
  We respectively call
  $\lfp f$ and $\gfp f$ the \emph{least} and \emph{greatest solutions}
  to the equation.
  If $\lfp f = \gfp f$, the unique fixpoint of $\Phi$ is simply called
  \emph{the} solution to the equation.
\end{definition}
We argue that monotone equations are common in cost analysis
applications. Large classes of examples are given in
Section~\ref{sec:grammars-of-eqs}, e.g.
Example~\ref{example:piecewise-aff-sum}. We also refer the reader to
our benchmarks in Table~\ref{table:benchmarks} of
Appendix~\ref{sec:benchmark-tables}, all of which are monotone,
including some operating on non-monotone resources such as memory.

\begin{proposition}
  A monotone equation always admits a solution in $\dd\to\rri$.
\end{proposition}
\begin{proof}
  This is a direct corollary of Theorem~\ref{th:knaster-tarski}.
\end{proof}
\noindent The following observation,
also a direct corollary of Theorem~\ref{th:knaster-tarski},
enables us to reformulate our bound search problem
as a question of pre/post-fixpoints of monotone operators.
\begin{theorem}\label{th:proof-principle}
  Let $\Phi:\rri^\dd \to \rri^\dd$ be a monotone equation.
  For any $f:\dd\to\rri$,
  \begin{itemize}
    \item if $\Phi f \leq f$, i.e. if $f$ is a postfixpoint of $\Phi$,
          then $\lfp \Phi \leq f$, i.e. $f$ is an upper bound on the least
          solution to the equation,
    \item if $f \leq \Phi f$, i.e. if $f$ is a prefixpoint of $\Phi$,
          then $f \leq \gfp \Phi$, i.e. $f$ is a lower bound on the greatest
          solution to the equation.
  \end{itemize}
\end{theorem}

\begin{remark}
  In particular, notice that if $\Phi$ admits a unique solution
  $\fsol$, then any postfixpoint is an upper bound on $\fsol$, and any
  prefixpoint is a lower bound on $\fsol$.
  This is always the case when the equation is \emph{well-defined}, in
  the sense that it terminates everywhere as a program, as will be
  explained in Section~\ref{subsec:termination}.
\end{remark}
\noindent With this insight, we introduce the following synonyms.
\begin{definition}[Inductive bounds]
  \label{defi:inductive-bounds}
  Let $\Phi$ be a monotone equation.
  We say that $f:\dd\to\rri$ is an \emph{inductive upper bound} (resp.
  \emph{inductive lower bound}) whenever it is a postfixpoint (resp.
  prefixpoint) of $\Phi$.
  Additionally, we simply say that $f$ is an \emph{inductive bound}
  whenever it is either a prefixpoint or a postfixpoint.
\end{definition}
We compare below these notions with classical proofs by induction of
upper/lower bounds.
When the restriction $\Phi:(\dd\to\rr)\to(\dd\to\rr)$ is well-defined,
i.e. when infinities are unnecessary, we can also introduce the
following concept.
\begin{definition}[Inductivity measure]
  \label{defi:inductivity-measure}
  For any monotone equation $\Phi$ and $f:\dd\to\rr$, we call the
  function $\Phi f - f$ the (local) \emph{inductivity measure} of $f$,
  to reflect the intuition that local information on $\Phi f - f$
  gives partial information on whether $f$ is an inductive bound.
  Indeed, $f$ is an inductive upper (resp. lower) bound whenever
  $\sgn(\Phi f -f)$ is $(-)$ (resp. $(+)$) everywhere.
\end{definition}
This local viewpoint is leveraged in our algorithms, e.g. to
exploit dynamic invariant generation techniques.
We also later explain that quantitative information on
``offending'' values of the function $\Phi\candf - \candf$ may be used
to \emph{repair} candidates $\candf$, i.e. transform them into actual
inductive bounds with constant-sign inductivity measure, in accordance
with our motto of \emph{taking the operator seriously}.

\subsection{Termination of Equations}
\label{subsec:termination}

The exposition of multiple results and proofs in this paper can be
simplified if we assume that the equations we deal with always
terminate, in the viewpoint where equations are simply programs that can be
executed, with solution given by their semantics.
This is directly related to termination of derivations in the
operational semantics expressed in~\cite{mlrec-tplp2024-preprint}.
In particular, this will allow us to simply use the vector space
structure of $\dd\to\rr$ without extra technical care for infinities,
and to simply identify pre/postfixpoints with lower/upper bounds on
\emph{the} solution without having to distinguish between $\lfp$ and
$\gfp$.

Termination of equations is defined at the level of syntactical
representations. This can be done by introducing a boolean domain
$\bb=\{\bot,\top\}$, (corresponding to ``undefined'', ``defined'').
We can then abstract each operation in the syntax of expressions by
$\llbracket - \rrbracket^\flat(x_1,\dots,x_n)=\top\iff\forall i,x_i=\top$
(in particular, constants are abstracted to $\top$, and this
generalises trivially to infinite arities).
It is then easy to abstract $\Phi$ into $\Phi^\flat:(\dd\to\bb)\to(\dd\to\bb)$
(again, note that we are abusing notation: this depends \emph{a
priori} on the syntax of the definition of $\Phi$, and not only on its
intensive properties $(\dd\to\rr)\to(\dd\to\rr)$).
The domain of termination of $\Phi$ simply appears as
$(\lfp_{\leq_\flat}\Phi^\flat)^{-1}\{\top\}\subseteq\dd$.

\begin{definition}[Termination]\label{defi:eq-termination}
  We say that an equation \emph{terminates everywhere unconditionally}
  whenever $(\lfp_{\leq_\flat}\Phi^\flat)^{-1}\{\top\}=\dd$ for the
  above abstraction.
  For simplicity, we also require that
  $\Phi:(\dd\to\rri)\to(\dd\to\rri)$ restricts to
  $\Phi:(\dd\to\rr)\to(\dd\to\rr)$.
  When the context is clear, we will equivalently simply say that
  \emph{$\Phi$ terminates}, or even just that \emph{$\Phi$ is a
  well-defined equation}.

  Note that in the general case, the required number $\beta$ of
  iterations so that the execution $\Phi^\beta(\bot)(\vn)$ converges to
  $(\lfp\Phi)(\vn)$ may be transfinite.

  In this case, $\Phi$ admits a single solution solution
  $\lfp\Phi=\gfp\Phi=\fsol:\dd\to\rr$.
\end{definition}
Note that unconditional termination is implied by the existence of a
\emph{ranking function for the equation} (Definition~\ref{defi:rkfun}
later in the paper).
In the rest of this paper, unless explicitly said otherwise, we will
always assume that this is the case (reducing $\dd$ if required) for
the sake of simplicity and presentation.

\subsection{Discussion and Interpretation}
\label{subsec:monotonicity-discussion}

\subsubsection{Relation with classical proof by induction.}

Proofs by Theorem~\ref{th:proof-principle} may be seen as
easy-to-automate generalisations of proofs by induction,
which motivates the vocabulary introduced in
Definition~\ref{defi:inductive-bounds}.
For example, a proof by induction of a candidate upper bound $\candf$
for \texttt{\exptwoname} of Table~\ref{table:benchmarks} would start
by noticing that $\fsol(0)=3\leq\candf(0)$, and then observe that, for
$n\geq 1$, $\fsol(n)=2\fsol(n-1)+1\leq 2\candf(n-1)+1\leq\candf(n)$.
The direct connection with monotonicity of $\Phi$ is apparent.
However, automatically setting up and performing inductive proofs for
benchmarks such as \texttt{\multiphaseonename} or \texttt{\sumoscname}
is not easy, as it requires selecting an appropriate ``induction path''.
This difficulty is circumvented by Theorem~\ref{th:proof-principle}.

\subsubsection{Interpretation of \emph{equation} monotonicity.}

We informally discuss intuition to better visualise the main
hypothesis of Theorem~\ref{th:proof-principle}, which does \emph{not}
imply limitations to monotone $\fsol$ or monotone resources (e.g.
excluding memory that can be freed during execution), as shown by our
benchmarks.

Given an equation, we can picture its \emph{control-flow} by
organising $\dd$ as a graph, with arrows going from callers to callees
for each recursive call.
In fact, for cost properties, a better visualisation is to consider
the \emph{opposite} orientation, from callees to possible callers,
hence making base cases \emph{sources} (as opposite to \emph{sinks}),
from which cost \emph{flows}, and \emph{accumulates} towards further
recursive cases. This orientation, with logical content, coincides
with that of \emph{derivation trees} from logic programming, when
viewing equations as Horn clauses.

With this image, monotonicity can be interpreted dynamically, and
simply means that, if a cost value is increased \emph{upstream} (e.g.
an initial condition), then all values \emph{downstream} also increase.
This does \emph{not} imply monotonicity of $\lfp\Phi$ regarding an
underlying order of $\dd$, as it may not coincide with the direction
of recursive calls (consider \texttt{\incronename} in
Table~\ref{table:benchmarks}). %
This also does \emph{not} mean that values downstream are higher
than upstream: flow can be withdrawn along the stream without
affecting the monotonic response of values downstream -- this allows
to handle ``non-monotonic'' resources such as memory.

However, non-monotone equations do exist, as discussed in
Remark~\ref{rem:non-monotone-discussion-grammar} below. Examples
outside of pure program analysis appear in biochemical systems, where
an increase of population at some point in time may actually have a
negative impact on the population downstream. We can see that
monotonicity of $\Phi$ is connected with the possibility to compute
worst/best cases for an input size from worst/best cases of the
recursive calls.

\section{Grammars and Classes of Equations}
\label{sec:grammars-of-eqs}

We may break apart equations into simpler building blocks, introducing
grammar of operators. This allows us to easily define classes of
equations, and inductively prove properties on them. Additionally, this
viewpoint provides us with syntax of supported inputs for
Sections~\ref{sec:cand-generation}-\ref{sec:implementation-and-evaluation}
and enables the ``abstract equations'' viewpoint sketched in
Section~\ref{subsubsec:abstract-iteration}.

\subsection{Intrinsic definitions and properties}

We start by defining classes of equations through their properties.
Constructive subsets of these classes are obtained in
\ref{subsec:induct-defis} through grammar of equations.

The approach presented in this paper focuses on \emph{monotone
equations}, %
but other properties will be useful. We define them here.

\begin{definition}[Linear equation]
  \label{defi:linear-equation}
  We say that an equation $\Phil:(\dd\to\rri)\to(\dd\to\rri)$ is
  \emph{linear} whenever it restricts to a linear operator
  on the vector space $\dd\to\rr$, i.e. whenever
  $\forall f,g:\dd\to\rr,\,\Phil(f+g)=\Phil(f)+\Phil(g)$ and
  $\forall f:\dd\to\rr,\,\forall\alpha\in\rr,\,\Phil(\alpha f)=\alpha\cdot\Phil(f)$.
  In particular, in this case, $\Phil(\vzero)=\vzero$.
\end{definition}

\begin{definition}[Affine equation]
  \label{defi:affine-equation}
  Similarly, we say that $\Phi$ is \emph{affine} whenever it
  restricts to an affine operator on $\rr^\dd$,
  i.e. whenever there exists a linear $\Phil$ such that
  $\forall f\in\rr^\dd$, $\Phi(f)=\Phi(\vzero)+\Phil(f)$.
  In this case, we call $\Phil$ the \emph{linearisation} of $\Phi$,
  and we can observe that
  $\forall f,g\in\rr^\dd,\,\Phi(f-g)=\Phi(f)-\Phi(g)$.
\end{definition}

\noindent While our approach is only limited to monotone
\emph{operators}, without restrictions on candidates, the special case
of \emph{monotone candidates} is interesting.
Unless specified, we give $\dd$ the order induced by the product of
the usual order on $\nn$.
\begin{definition}[Monotone candidate]
  \label{defi:monotone-candidate}
  We call \emph{monotone candidate} any function
  $f:\dd\to\rri$ that is monotone, i.e. order-preserving.
\end{definition}

\begin{proposition}
  \label{prop:monotone-candidate-sublat}
  The set of monotone candidates is a complete sublattice of
  the set of functions $\dd\to\rri$ ordered pointwise.
\end{proposition}

\begin{definition}%
  \label{defi:monotonicity-preservation}
  We say that an equation $\Phi$ \emph{preserves candidate
  monotonicity} whenever $\Phi f$ is monotone for all monotone $f$.
\end{definition}

\begin{proposition}
  \label{prop:monotonicity-preservation}
  If $\Phi$ is a monotone equation that preserves candidate
  monotonicity, then it admits a solution that is monotone. In
  particular, if $\Phi$ admits a unique solution $\fsol$, e.g. if
  $\Phi$ terminates unconditionally, then $\fsol$ is monotone.

  More generally, if a monotone equation $\Phi$ preserves a
  sublattice of $\dd\to\rri$ defined by a property, then it admits
  solutions that have this property.
\end{proposition}

\subsection{Order}

We continue our algebraic treatment of equations as first-class
objects, by introducing an order on monotone equations themselves.
This higher-order order structure interacts well with the semantics of
equations (with their solutions), and can be used to simplify them,
i.e. to safely approximate them.
It is also useful for proofs, such as that of
Proposition~\ref{prop:hybrid-verif} on hybrid verification.

\begin{definition}[Order of equations]
  \label{defi:pointwise-eq-order}
  We define the order $\eqleq$ on equations by lifting pointwise (once
  more) the pointwise order of $\dd\to\rri$.\\
  In other words,
  $\Phi_1\sqsubseteq\Phi_2 \iff
    \forall f:\dd\to\rri,\,\forall \vn\in\dd,\,
    \Phi_1(f)(\vn) \leq \Phi_2(f)(\vn).$
  This gives $((\dd\to\rri)\to(\dd\to\rri),\eqleq)$ the structure of a complete
  lattice, of which the set of monotone equations is a complete
  sublattice.
\end{definition}

\begin{proposition}[Fixpoint transfer]
  Solution extraction is $\eqleq$-monotone.

  \noindent More precisely, if
  $\Phi_1,\Phi_2$ are two monotone
  equations with $\Phi_1\eqleq\Phi_2$, then $\lfp\Phi_1\leq\lfp\Phi_2$
  and $\gfp\Phi_1\leq\gfp\Phi_2$.
  In particular, when $\Phi_1$ and $\Phi_2$ terminate everywhere and
  thus admit a single solution, we can write $f_{sol,1}\leq f_{sol,2}$.
\end{proposition}

\begin{remark}
  Over- and underapproximations of equations for the $\eqleq$ order
  can be easily performed at the syntactical level, by using over- and
  underapproximations of the expressions used in the definition of the
  equation.
  For example, if $\Phi(f)(\vn):=\ite(\varphi(\vn),y_0,\Phi_1(f)(\vn))$
  and $\Phi'(f)(\vn):=\ite(\varphi(\vn),\overline{y_0},\Phi_1(f)(\vn))$
  with $y_0\leq\overline{y_0}$ (overapproximation of initial
  conditions), then $\lfp\Phi \leq \lfp\Phi'$.
  Similarly, when setting up an equation by a $\max$ in a branch, this
  can be (coarsely but safely) overapproximated by a sum, i.e.
  schematically $\max(\Phi_1,\Phi_2)\eqleq(\Phi_1+\Phi_2)$.
  Negligible or complex terms can also be simplified, e.g. to attempt
  to recover simple equations handled by CAS solvers.
  For example, $\Phi\eqleq\widetilde{\Phi}$ for
  $\Phi(f)(n)=\ite(n=0,0,2f(n-1)+n^2+3n+\lfloor\log_2(n)\rfloor+4+\frac{1}{n})$
  and $\widetilde{\Phi}(f)(n)=\ite(n=0,0, 2f(n-1)+11n^2)$.
  While these transformations may appear to be trivial and natural,
  our formalisation provides guarantees for them, and highlights an
  important hypothesis: \emph{equations $\Phi$ must be
  \textbf{monotone} in their unknown}. Multiple counter-examples can
  be produced if this hypothesis is violated.
\end{remark}

\subsection{Inductive definitions and grammars}
\label{subsec:induct-defis}

We finally explicitly introduce large classes of equations supported
by our approach, thanks to their monotonicity and other good
properties. These are the equations that we analyse in our experiments
and in Section~\ref{sec:implementation-and-evaluation}.
Note that we only prove that our implementation is correct for such
equations, and we assume that all inputs belong to this syntax so
as to make them monotone. %

\begin{proposition}
   Let $\Eqs$ denote the set of all monotone equations
   $\rr^\dd \to \rr^\dd$.
   Then, $\Eqs$ is stable by all operations presented in the grammar
   of Fig.~\ref{eqs:grammar-1}.
\end{proposition}
\begin{remark}[Warning]
  \label{rem:non-monotone-discussion-grammar}
  $f\mapsto -f$ is not a monotone operation, i.e. positivity on $a$ in
  (\ScaleOp) is necessary.
  More subtly, (\ChoiceRestrOp) also requires $\varphi$ to depend only
  on $\vn$ (and not on $f$).
  Non-monotone equations can be constructed with piecewise definitions
  whose conditions depend on the result of recursive calls (e.g. the
  \texttt{prs23\_1} benchmark in~\cite{mlrec-tplp2024-preprint,Wang-PRS23-short}). The
  order theory of such equations is interesting but subtle, and
  requires a lot of care in handling $\varphi$ -- this is left for
  future work.
\end{remark}
\begin{figure}\footnotesize
  \vspace{-1em}
  \begin{equation*}
  \begin{gathered}
        \trfrac[\;(\GroundOp)]{P:\dd\to\rr}{(f\mapsto \vn\mapsto P(\vn)):\Eqs}
        \qquad
        \trfrac[\;(\PrecompOp)]{\phi:\dd\to\dd}{(f\mapsto f\circ\phi):\Eqs}
        \\
        \trfrac[\;(\SumOp)]{\Phi_1:\Eqs,\;\Phi_2:\Eqs}{\big(f\mapsto \Phi_1(f)+\Phi_2(f)\big):\Eqs}
        \qquad
        \trfrac[\;(\ScaleOp)]{\Phi:\Eqs,\;a:\dd\to\rr_+}{\big(f\mapsto\vn\mapsto a(\vn)\cdot\Phi(f)(\vn)\big):\Eqs}
        \\
        \trfrac[\;(\ChoiceRestrOp)]{\Phi_1:\Eqs,\;\Phi_2:\Eqs,\;\varphi:\dd\to\{\top,\bot\}}{\Big(f\mapsto\vn\mapsto\ite\big(\varphi(\vn),\Phi_1(f)(\vn),\Phi_2(f)(\vn)\big)\Big):\Eqs}
        \;
        \trfrac[\;(\ComposeOp)]{\Phi_1:\Eqs,\;\Phi_2:\Eqs}{\Phi_2\circ\Phi_1:\Eqs}
        \\
        \trfrac[\;(\MaxOp)]{\Phi_1:\Eqs,\;\Phi_2:\Eqs}{\Big(f\mapsto\max\big(\Phi_1(f),\Phi_2(f)\big)\Big):\Eqs}
        \quad
        \trfrac[\;(\MinOp)]{\Phi_1:\Eqs,\;\Phi_2:\Eqs}{\Big(f\mapsto\min\big(\Phi_1(f),\Phi_2(f)\big)\Big):\Eqs}
  \end{gathered}
  \end{equation*}
  \vspace{-1em}
  \caption{A grammar of operators to construct equations. All
    constructors presented here preserve monotonicity of equations.
    Note that no assumptions are made on $P$ and $\phi$ in (\GroundOp)
    and (\PrecompOp), which form the base cases for structural
    induction.
    (\SumOp) and (\ScaleOp) allow analysing programs with non-linear
    recursion.
    (\ChoiceRestrOp) corresponds to (restricted) conditional choice,
    and enables piecewise definitions.
    (\ComposeOp) is convenient for combining equations.
    Finally, (\MaxOp) and (\MinOp) give simple ways to analyse worst- and
    best-cases in branching programs.
    \label{eqs:grammar-1}}
  \vspace{-1em}
\end{figure}
\begin{definition}
  We call $\EqsFull$ the set of equations generated by all constructors
  in Fig.~\ref{eqs:grammar-1}, i.e. the corresponding least fixed point in
  $\pp(\rr^\dd\to\rr^\dd)$.
\end{definition}
\begin{remark}
  For implementation, we also need to define the syntax of functions
  $P$, $a$, $\phi$ and $\varphi$ that can appear in
  Fig.~\ref{eqs:grammar-1}, even if we have no restriction in
  principle.
  In the artifact attached to this paper, we circumvent this
  technicality by using an embedded language, i.e. using syntax of
  numerical expressions naturally available in languages such as
  \texttt{Python} and \texttt{Wolfram}, and we allow users to
  extend it via a bank of custom functions (e.g. $\log_2$,
  $\texttt{ceiling}$, $\texttt{factorial}$, etc.).
\end{remark}
All proofs of this section are easy, by induction on the syntax.
\begin{proposition}
  All equations in $\EqsFull$ are monotone.
\end{proposition}

\begin{definition}
  We call $\EqsRestr$ the subset of $\EqsFull$ containing all
  equations generated by the constructors in
  Fig.~\ref{eqs:grammar-1}, (\MaxOp) and (\MinOp) excluded.
\end{definition}

\begin{proposition}
  All equations in $\EqsRestr$ are both monotone and affine.
\end{proposition}

\begin{example}\label{example:piecewise-aff-sum}
  All equations of the following form belong to $\EqsRestr$.
  \begin{equation}\label{eq:piecewise-aff-sum}
    \begin{aligned}
    \Phi : (\dd \to \rr) &\to (\dd \to \rr)\\
         f &\mapsto \vn \mapsto \begin{cases}
           \;\cdots&\\
           \sum_{j=1}^{k_i} \big(a_{i,j}(\vn)\cdot f(\phi_{i,j}(\vn))\big) + b_i(\vn)
              &\text{if } \varphi_i(\vn)\\
           \;\cdots&
         \end{cases},
  \end{aligned}
  \end{equation}
  where all the $a_{i,j}$ are positive, and $\phi_{i,j}:\dd\to\dd$,
  $b_i:\dd\to\rr$ are arbitrary.
\end{example}
This class contains many examples.
In fact, we can prove the following.
\begin{proposition}
  All elements of $\EqsRestr$ are of the
  shape~$(\ref{eq:piecewise-aff-sum})$.
\end{proposition}
\begin{corollary}
  Elements of $\EqsFull$ can be understood as equations of
  shape~$(\ref{eq:piecewise-aff-sum})$,
  interleaved with nested $\min$ and $\max$ operators.
\end{corollary}

\begin{proposition}
  Candidate monotonicity in preserved by (\SumOp), (\ComposeOp),
  (\MaxOp), (\MinOp).
  It is preserved by (\GroundOp), (\PrecompOp) and (\ScaleOp) whenever
  $P$ (resp. $\varphi$, resp. $a$) are monotone.
  It is not preserved in general by (\ChoiceRestrOp), but sufficient
  conditions can be obtained to deal with special cases (e.g. initial
  conditions).
\end{proposition}
The grammar of Fig,~\ref{eqs:grammar-1} can be extended to wider
types of equations, featuring nested calls (self-composition), and
other very nonlinear behaviours, on the condition of being a bit more
careful about the lattice we work in.
For the sake of presentation, we do not include these constructors
in our grammar, as they would require precise hypotheses, and only
give the following few results.

\begin{proposition}[Self-composition]
  The operation $(f,g)\mapsto f\circ g$ is monotone in both variables
  for $f$, $g$ in the sublattice of monotone candidates, when this
  operation is well-defined, i.e. when $\dd\subseteq\nn$,
  $g(\nn)\subseteq\nn$. Generalisations can be obtained.
  Note that this allows us to deal with the \texttt{\nestedname} benchmark
  in Table~\ref{table:benchmarks}.
\end{proposition}

\begin{proposition}[Monotone post-composition]
  Operations (\SumOp), (\ScaleOp), (\MaxOp), (\MinOp) can be unified
  by noting that for all monotone $P:\rr^k\to\rr$, the equation
  $f\mapsto P(\Phi_1(f),\dots,\Phi_k(f))$ is monotone for all monotone
  $\Phi_1,\dots\Phi_k$.
  For example, $f\mapsto n\mapsto\ite(n=0,1,f(n-1)^2+1)$ is monotone
  on positive candidates.
\end{proposition}

\begin{proposition}[Unbounded optimisation]
  Operations (\MaxOp), (\MinOp) can be immediately extended to all
  finite arity $\max$ and $\min$ by iterating them.
  More generally, even \emph{unbounded} $\max$ and $\min$ operations
  are monotone. This allows, e.g., dealing with explicit
  non-deterministic divide-and-conquer equations.%
\end{proposition}

\subsection{Application to inductivity measure}
\label{subsec:grammars-extra-results}

The compositional and axiomatic viewpoint presented in this section
allows us to easily obtain and combine results on the inductivity measure
$\Phi f - f$ of Definition~\ref{defi:inductivity-measure}.
For example, it admits a simple expression for affine equations
applied to affine combinations of base functions.
This will be useful to enforce inductivity of candidates in our
optimisation-based approach of Section~\ref{subsec:optim-approach}, as
it will be possible to enforce inductivity locally using only linear
constraints in the $\alpha_i$.
\begin{proposition}\label{prop:aff-on-aff}
  Let $\Phi$ be an affine equation, and $f=(\sum_i \alpha_i f_i) + b$ be an affine
  combination of functions.
  Then, $\Phi f = (\sum_i \alpha_i\cdot\Phil f_i) + \Phi(b)$, where
  $\Phil := f \mapsto \Phi(f) - \Phi(\vzero)$ is the linear component
  of $\Phi$, i.e. the linearisation of $\Phi$.
  In other words, the behaviour of affine equations on affine
  candidates is easy to express explicitly.
  In particular, the inductivity measure can be expressed as
  $\Phi f - f = (\sum_i \alpha_i\cdot(\Phil f_i - f_i)) + (\Phi(b)-b)$.
\end{proposition}
In Section~\ref{sec:cones}, we are even more explicit, and give
expressions of $(\Phi f - f)$ that provide, for basic building blocks
of equations, simple sufficient conditions for a bound to be inductive.
For now, we focus on how such results can be combined through the
constructors of our grammar. In other words, taking a deductive
approach, inductivity $(\Phi f - f)$ of candidates can be analysed by
induction on the structure of the equation -- such results could be
expressed as type rules.
For example, it is preserved by (full) conditional choice and convex
combinations of equations (which may be viewed as a restricted form of
probabilistic choice).

\begin{proposition}
  Let $\Phi_1,\Phi_2$ be two monotone equations.
  Then, for all $f$,
  if $\sgn(\Phi_1(f)-f)$ and $\sgn(\Phi_2(f)-f)$ are equal and
  globally constant, $\sgn(\Phi(f)-f)$ is identical, in both of the following
  cases:
  when $\Phi(f)(\vn)$ is defined piecewise via
  $\varphi:(\dd\to\rr)\times\dd\to\{\top,\bot\}$
  as $\ite(\varphi(f,\vn), \Phi_1(f)(\vn), \Phi_2(f)(\vn))$ and is monotone;
  and when
  $\Phi=\lambda\Phi_1+\mu\Phi_2$ for $\lambda,\mu\in\rr_+$,
  $\lambda+\mu=1$.
  In other words, to check that $f$ is inductive for $\Phi$, it is
  sufficient to do so for $\Phi_1$ and $\Phi_2$.
\end{proposition}
We may also observe that for affine equations, inductivity conditions
can be obtained from conditions for linear equations, by relativisation.
\begin{proposition}
  For all affine $\Phi$ and all $f$,
  $\Phi f - f = \Phil(f-\fsol)-(f-\fsol)$.
\end{proposition}
For example, conditions in Section~\ref{sec:cones} on the growth of
$f$ will be translated into conditions on the growth of the difference
$f-\fsol$ with the (unknown) solution.
Similarly, sufficient conditions for inductivity of bounds can be
obtained for $\max/\min$ of operators or for composition, and should
be easily obtained by the interested reader, but we do not describe
them here for brevity.

Just as inductivity $(\Phi f - f)$ of candidates can be analysed by
induction on the structure of the equation, the same is true of
\emph{variation} of inductivity measures under a transformation $\Psi$
of expressions, i.e. of $(\Phi\Psi f - \Psi f) - (\Phi f - f)$. This
is relevant for candidate \emph{repair}, i.e. to build inductive
bounds from candidate bounds only ``nearly inductive'', which is quickly sketched
in Section~\ref{sec:repair}.

\section{Structure of sets of pre- and postfixpoints -- Robustness and
  Extensions of the proof principle}
\label{sec:cones}

We now investigate the \emph{converse} of the proof principle
introduced in Section~\ref{sec:eqs-as-ops}, i.e. we question the
\emph{completeness} of the approach: when are bounds inductive?

The results of our exploration are summarised as a rather explicit
description of $\PreFix(\Phi)$ and $\PostFix(\Phi)$ in terms of convex
cones (Theorem~\ref{th:explicit-cone-of-aff} for affine equations,
partially extended to the full grammar in
Proposition~\ref{prop:cones-of-maxmin}).
The conditions we obtain can be interpreted as conditions on the
relative growth of candidates with respect to the solution, \emph{in
the direction of the recursive calls}, as exemplified at the end of
Section~\ref{subsec:cones-1}.
They fully explain the nature of our counter-examples, and suggest
ways to repair them (as sketched in Section~\ref{sec:repair}).
This view is further empowered by additional algebraic structure of
these cones (Proposition~\ref{prop:cone-multiplication}) and by
relationship with (quasi-)ranking functions for the equation.

Additionally, Section~\ref{subsec:proof-extended} proposes tools to
extend our proof principle, so as to handle some of the non-inductive
bounds and simplify the proof process, via \emph{hybrid verification}
and introduction of \emph{subdomain-wise} reasoning enabling
control-flow analysis and refinement (CFA/CFR) for equations.

\subsection{Structure of the set of postfixpoints}
\label{subsec:cones-1}

\begin{proposition}[Structure of $\PostFix$/$\PreFix$ for affine equations]
  \label{prop:struct-cone-affine}
  Let $\Phi$ be a (monotone) affine equation.
  Then, the sets $\PostFix(\Phi)$ and $\PreFix(\Phi)$ are affine convex cones
  described by $\PostFix(\Phi) = \lfp \Phi + C$ and
  $\PreFix(\Phi) = \gfp \Phi - C$, where $C$ is a linear convex cone.
  Moreover, the linear cone $C$ is equal to $\PostFix(\Phil)$, where
$\Phil$ is the linearisation of $\Phi$.%
\end{proposition}
\begin{proof}
  For simplicity, we only present the case of $\dd\to\rr$, without infinities.
  Recall that for all $h, h'$, $\Phil(h-h') = \Phi(h)-\Phi(h')$,
  so that for any $f:\dd\to\rr$,
  $\Phi f - f = (\Phi f - \Phi(\lfp\Phi)) - (f - \lfp \Phi)
   = \Phil(f - \lfp\Phi) - (f - \lfp\Phi),$
  hence
    $\PostFix(\Phi) = \lfp\Phi + \PostFix(\Phil)$.
  We could actually show that $\PostFix(\Phi) = \Fix(\Phi) + \PostFix(\Phil)$,
  but we choose the former for its
  geometrical interpretation.

  \noindent We now prove that $\PostFix(\Phil)$ is a linear convex cone.
  Let $f,g\in\PostFix(\Phil)$,  $\alpha,\beta\in\rr_+$.
  We need to show that $\alpha f + \beta g \in \PostFix(\Phil)$,
  which follows from linearity:
  $
    \Phil(\alpha f + \beta g) - (\alpha f + \beta g)
    = \alpha\cdot(\Phil f - f) + \beta\cdot(\Phil g - g)
    \leq \vzero.
  $

  Finally, we need to prove that $\PreFix(\Phi) = \gfp \Phi - \PostFix(\Phil)$,
  showing that it is also an affine convex cone, and that its corresponding
  linear cone is the same as that of $\PostFix(\Phi)$ up to a sign reversal.
  This follows directly from the fact that, for all $f$,
  $\Phil(\gfp \Phi - f) - (\gfp \Phi - f)
  = (\Phi (\gfp \Phi) - \Phi f) - (\gfp \Phi - f)
  = f - \Phi f,$
  hence
  $f \in\PreFix(\Phi) \iff \gfp \Phi-f \in \PostFix(\Phil) \iff f \in \gfp \Phi-\PostFix(\Phil)$.
\end{proof}
The next lemma aids interpretation and to later unveil
multiplicative structure. \vspace{-1em}
\begin{lemma}[Cone positivity]
  \label{lem:cone-pos}
  Assume that $\Phi$ is a ``well-defined equation'', i.e. that it
  terminates everywhere unconditionally (c.f.
  Definition~\ref{defi:eq-termination}).
  Then, with the assumptions of
  Proposition~\ref{prop:struct-cone-affine}, all
  $f\in C=\PostFix(\Phil)$ are such that $f\geq \vzero$.
\end{lemma}

\begin{proof}[Sketch]
  By transfinite induction on the height of derivations.
\end{proof}

\noindent We finally give an explicit description of the cones for the
affine equations we work with, by unfolding the
definition of $\PostFix(\Phil)$. It can be interpreted both
externally, as relative growth conditions, and internally via rays
akin to ``quasi-ranking functions''.
\begin{theorem}\label{th:explicit-cone-of-aff}
  Let $\Phi$ be a (monotone) affine equation written as in
  Example~\ref{example:piecewise-aff-sum}, i.e.
  $\Phi : f \mapsto \vn \mapsto
           \sum_{j=1}^{k_i} \big(a_{i,j}(\vn)\cdot f(\phi_{i,j}(\vn))\big) + b_i(\vn)
              \text{ if } \varphi_i(\vn),$
  with $a_{i,j}\geq\vzero$.
  As usual, we write $\dd = \sqcup_i \dd_i$, where $\dd_i = \{\vn\in\dd\,|\,\varphi_i(\vn)\}$.

  Then $\PostFix(\Phi) = \lfp \Phi + C$ and $\PreFix(\Phi) = \gfp \Phi - C$,
  where
  \vspace{-0.5em}
  \begin{equation*}
    C = \big\{g:\dd\to\rri \,\big|\,
          \forall i,\, \forall \vn\in\dd_i,\,
          g(\vn) \geq \sum\nolimits_j a_{i,j}(\vn)\cdot g(\phi_{i,j}(\vn))
        \big\}.
  \vspace{-0.5em}
  \end{equation*}
  In other words,
{\footnotesize
  \begin{align*}
    \PostFix(\Phi) &= \Big\{f\in\rri^\dd \,\Big|\,
          \forall i\,\forall \vn\in\dd_i\;
          (f-\lfp\Phi)(\vn) \geq \sum_j a_{i,j}(\vn)\cdot (f-\lfp\Phi)(\phi_{i,j}(\vn))
        \Big\},\\[-4pt]
    \PreFix(\Phi) &= \Big\{f:\rri^\dd \,\Big|\,
          \forall i\,\forall \vn\in\dd_i\;
          (\gfp\Phi-f)(\vn) \geq \sum_j a_{i,j}(\vn)\cdot (\gfp\Phi-f)(\phi_{i,j}(\vn))
        \Big\}.
  \end{align*}
}
\vspace{-1.5em}
\end{theorem}
The results above can be extended to the full grammar of
Fig.~\ref{eqs:grammar-1} (with $\min/\max$) using the
following proposition, with the caveat of replacing equalities with
inclusions, hence obtaining sufficient but not necessary conditions
for inductiveness.
\vspace{-1em}
\begin{proposition}\label{prop:cones-of-maxmin}
  If $\Phi = \max(\Phi_1,\Phi_2)$, %
  then $\PostFix(\Phi) = \PostFix(\Phi_1)\cap \PostFix(\Phi_2)$,
  and $\PreFix(\Phi) \supseteq \PreFix(\Phi_1)\cap \PreFix(\Phi_2)$.
  Similarly, if $\Phi = \min(\Phi_1,\Phi_2)$, then
  $\PostFix(\Phi) \supseteq \PostFix(\Phi_1)\cap \PostFix(\Phi_2)$,
  and $\PreFix(\Phi) = \PreFix(\Phi_1)\cap \PreFix(\Phi_2)$.

  As the intersection of two affine convex cones is an affine convex cone, this
  implies that for all equations $\Phi$ in the family defined in
  Fig.~\ref{eqs:grammar-1}, their sets $\PostFix(\Phi)$ and $\PreFix(\Phi)$ also
  contain affine convex cones.
  Moreover, $\PostFix(\Phi)$ and $\PreFix(\Phi)$ are stable by addition from the
  corresponding \emph{linear} convex cones.
\end{proposition}
\begin{remark}
  An intuitive way to visualise sets stable by additions from cones is
  as ``icebergs'', with sharp points at minimal elements for the induced cone preorder.
\end{remark}

\noindent Under mild assumptions, these cones have additional
algebraic structure.
\begin{proposition}[Stability by multiplication]\label{prop:cone-multiplication}
  Consider an affine equation given explicitly with the same
  notation as above. Assume that it terminates everywhere so that
  Lemma~\ref{lem:cone-pos} can be applied, and that all $a_{i,j}$
  are such that $\forall\vn,a_{i,j}(\vn)^2\geq a_{i,j}(\vn)$, i.e.
  $a_{i,j}(\vn)=0$ or $a_{i,j}(\vn)\geq 1$.

  Then, the linear cone $C=\PostFix(\Phil)$, appearing in the affine convex
  cones $\PostFix(\Phi)=\lfp\Phi+C$ and $\PreFix(\Phi)=\gfp\Phi-C$, is stable by
  (pointwise) multiplication,
  i.e. $\forall f,g\in C, (\vn\mapsto f(\vn)\cdot g(\vn))\in C.$
\end{proposition}
This makes it possible to build fast-growing rays (e.g. polynomials)
from simpler elements, e.g. as obtained below with ranking functions.

This is useful to have in mind when considering a final intuition
related to the cone structure presented in this section: ``if we
degrade the precision of a candidate long enough, we will eventually
reach an inductive bound''.
Indeed, intuitively, the presence of such cones guarantee that
inductive bounds are ``reasonably common'' in the space of all bounds,
and not only isolated points.\footnote{This is not the case anymore,
e.g., for the \texttt{\nestedname} benchmark of
Table~\ref{table:benchmarks}.}
This suggests that inductive upper bounds could be searched simply by
``going up and eventually encountering the cone''.
Three issues have to be considered to make this idea concrete:
the codimension of the cone to stay in an appropriate subspace,
the use of \emph{interior rays} to get closer to the cone, and the
necessity to produce fast-growing rays to eventually reach it (indeed,
unlike in finite dimensions $\forall\vn\,\exists N,\, N\cdot f(\vn) \geq
g(\vn) \not\Rightarrow \exists N, N\cdot f \geq g$).

\subsubsection{(Quasi)-ranking functions for rays.}

We now observe that elements of $C$ can be constructed from ranking
functions for the equation, either directly or through an exponential
transformation.
In particular, we are comforted that $C$ is not reduced to $\vzero$:
inductive bounds other than the exact solution exist. Additionally,
this observation is useful to explore candidate repair, making it
possible to leverage prior work on ranking function inference.

\begin{definition}[Ranking function for an equation]
  \label{defi:rkfun}
  A ranking function for a monotone equation
  $\Phi:(\dd\to\rr)\to(\dd\to\rr)$ is simply a function
  $\rkfun:\dd\to\nn$ (or more generally to a well-ordered set) such
  that $\rkfun$ decreases strictly along recursive calls, i.e. such
  that for each $\vx\in\dd$, the function $f\mapsto\Phi(f)(\vx)$
  depends only on the values of $f(\vy)$ for
  $\rkfun(\vy)<\rkfun(\vx)$.
\end{definition}
Since we focus on $\Phi$ that terminate, we will consider that such
functions exist.
For linear-recursive $\Phi$, they directly provide elements and hence
rays of $C$.
\begin{definition}
  We say that an affine equation is \emph{linear-recursive} when,
  given explicitly with the same notation as in
  Theorem~\ref{th:explicit-cone-of-aff}, there is a single recursive
  call in each subdomain ($\forall i, k_i=1$) with $a_{i,1}=1$.
\end{definition}
\begin{proposition}\label{prop:rkfun-are-rays}
   For all affine linear-recursive $\Phi$, ranking functions
   $\rkfun:\dd\to\nn$ belong to the linear cone $\PostFix(\Phil)$.
\end{proposition}
\begin{proof}
  Simply observe the expression of $C$ in
  Theorem~\ref{th:explicit-cone-of-aff} for linear-recursive $\Phi$.
  The ranking function definition corresponds to $\rkfun(\vn)\geq
  0$ in base cases and $\rkfun(\vn)\geq r(\phi_i(\vn))+1$ in recursive
  cases, which implies $\rkfun(\vn)\geq \rkfun(\phi_i(\vn))$.
\end{proof}
For affine $\Phi$, ranking functions can be simply
described as postfixpoints of a modified equation.
\begin{proposition}\label{prop:rkfun-are-postfixp}
  Consider a monotone affine equation $\Phi$ given explicitly as in
  Theorem~\ref{th:explicit-cone-of-aff}, and assume that $\Phi$ is
  terminating. Then, a function
  $\rkfun:\dd\to\rr$ is a ranking function for $\Phi$ if and only if
  it verifies the following inequation.
  {\footnotesize\begin{equation*}
    \forall\vn\in\dd,\,\rkfun(\vn) \geq \begin{cases}
           1 + \max_{1\leq j\leq k_i} r(\phi_{i,j}(\vn))
            &\text{if }\vn\in\dd_i\text{ with $\dd_i$ recursive}\\
           0
            &\text{if }\vn\in\dd_i\text{ with $\dd_i$ non-recursive}\\
         \end{cases}
  \end{equation*}}
\end{proposition}
Not all elements of $C$ are ranking functions.
We can describe more of them with the following generalisation, where
we allow real codomain and non-strict decrease. Such quasi-ranking
functions do not provide termination guarantees.
\begin{definition}[Quasi-ranking function]
  \label{defi:quasirkfun}
  For a monotone $\Phi:(\dd\to\rr)\to(\dd\to\rr)$, we call
  \emph{quasi-ranking function} any (non-negative) function
  $\rkfun:\dd\to\rr_+$ such that $\rkfun$ decreases along recursive
  calls, i.e. so that for each $\vx\in\dd$, the function
  $f\mapsto\Phi(f)(\vx)$ depends only on the values of $f(\vy)$ for
  $\rkfun(\vy)\leq\rkfun(\vx)$.
\end{definition}
\begin{proposition}
  For all affine linear-recursive $\Phi$, quasi-ranking functions
  $\rkfun:\dd\to\rr_+$ belong to the linear cone $\PostFix(\Phil)$.
\end{proposition}
\begin{remark}
  Constants are trivial quasi-ranking functions. Hence, for
  linear-recursive $\Phi$, inductive upper and lower bounds
  can always be degraded by a constant.
\end{remark}
This is not true anymore for non-linear-recursive equations, as seen
with (\texttt{\exptwoname}). However, we can still build elements in
the linear cone $C=\PostFix(\Phil)$ of general affine equations by
using their ranking functions.
\begin{proposition}
  Let $\Phi$ be a terminating affine equation, explicitly described
  as in Proposition~\ref{prop:rkfun-are-postfixp} and
  Equation~$(\ref{eq:piecewise-aff-sum})$.
  If all $a_{i,j}$ are bounded above by a constant $a$, and the $k_i$
  are bounded above by a constant $k$, then for all ranking functions
  $\rkfun$ of $\Phi$, the function $g:\vn\mapsto (ka)^{\rkfun(\vn)}$
  belongs to the linear cone $\PostFix(\Phil)$.
\end{proposition}
\begin{proof}
  This comes directly by observing that for $\vn$ in a recursive
  subdomain $\dd_i$,
  $g(\vn) = (ka)^{\rkfun(\vn)}
   \geq ka\cdot\max_{1\leq j\leq k_i} (ka)^{\rkfun(\phi_{i,j}(\vn))}
   \geq \sum_{1\leq j\leq k_i}a_{i,j}(\vn)\cdot g(\phi_{i,j}(\vn))$.
\end{proof}
By being a bit more careful, we could generalise this to more kinds of
equations (e.g. unbounded $a$, $\max/\min$), and obtain $g$ that are
more precise by avoiding the use of the exponential $a^{\-}$ in subdomains
with a single recursive call.

\subsubsection{Examples.}

As mentioned before, the results above provide inductivity conditions
that can be interpreted as relative growth conditions \emph{in the
direction of recursive calls}, which may be visualised in light of
the discussion of Section~\ref{subsec:monotonicity-discussion}.
Space considerations lead us to only present simple illustrations in
this section, leaving a unified description for a longer version of the
paper (it can be phrased through operator algebra, discussing
attractivity of modified discrete derivatives).
\mhnote{Same comment as before re longer version}
Additional examples, e.g. for equations with $\min/\max$, can easily be
created.

For linear recursive equations, relative growth is identical to
absolute growth, and the only relative condition is that inductive
bounds must be bounds.
\begin{example}
  Consider the equation given by $f(n)=f(n-1)+n$ if $n>0$ and
  $f(0)=3$.
  The results above tell us that an upper bound $g\geq\fsol$ is
  checkable with our proof principle whenever $\forall n>0,\, g(n)
  \geq g(n-1)$, i.e. whenever $g$ is an upper bound and has increasing
  discrete derivative.
\end{example}

\begin{example}
  Consider equation \texttt{\sumoscname} in
  Table~\ref{table:benchmarks}.
  An upper bound $g\geq\fsol$ is inductive whenever
  $\forall x,y>0,g(x,y) \geq g(x-1,y)$ and
  $\forall y>0,g(0,y)\geq g(1,y)$,
  i.e. whenever it has increasing derivative ``in the direction of
  recursive calls''.
\end{example}

\noindent For non-linear recursion, simple growth is insufficient: it
must be exponential.

\begin{example}
  \label{ex:exp2-explanation}
  Consider once again the \texttt{\exptwoname} equation,
  $f(n)=2f(n-1)+1$ if $n>0$ and $f(0)=3$, with solution
  $\fsol(n)=4\cdot 2^n-1$, which gives an example of
  non-inductive upper bound $\candf(n)=4\cdot2^n$.
  From the results above,
  an upper bound $g\geq\fsol$, is inductive whenever
  $\forall n>0,\, (g(n)-\fsol(n))\geq 2(g(n-1)-\fsol(n-1)).$

  This explains the non-checkability of $\candf$: while the candidate
  itself grows exponentially, its difference with $\fsol$ is constant,
  which is insufficient to permit an inductive proof.
  However, this also suggests ways to ``repair'' the candidate: if we
  replace $\candf$ by $\candf\circ(n\mapsto 1.2n)=n\mapsto
  4\cdot2^{1.2n}$, we get a new candidate that grows fast enough, and
  is an inductive bound.
  We come back to this idea in Section~\ref{sec:repair}.
\end{example}

\subsection{Extensions of the proof principle}
\label{subsec:proof-extended}

The space of checkable candidates can be expanded beyond the cones
mentioned above.

\textbf{Hybrid verification.}
Inductive proofs may be safely mixed with proofs obtained by other
means, in the following sense. For example, if inductivity fails only
on a finite subdomain, but the candidate still seems to be a correct bound,
it is safe to simply check it by concrete execution on this subdomain.

\begin{proposition}[Hybrid verification]
  \label{prop:hybrid-verif}
  Let $\dd = \dd_1 \sqcup \dd_2$.
  If $\fsol \leq f$ in $\dd_1$,
  and $\Phi f \leq f$ in $\dd_2$,
  then $\fsol \leq f$ everywhere.
  The dual is true for lower bounds.
\end{proposition}

\begin{lemma}[For hybrid verification]
  \label{lem:hybrid-verif}
  Let $\Phi:(\dd\to\rri)\to(\dd\to\rri)$, $\dd=\dd_1\sqcup\dd_2$,
  and suppose that $g:\dd_1\to\rri$ is such that
  $g\geq\lfp\Phi$ (in $\dd_1$).
  Consider
  {\footnotesize
  \begin{equation*}
    \widetilde{\Phi}:=f\mapsto\vn\mapsto\begin{cases*}
      \min(\Phi(f)(\vn),g(\vn)) &if $\vn\in\dd_1$,\\
      \Phi(f)(\vn)              &if $\vn\in\dd_2$.
    \end{cases*}
  \end{equation*}}
  Then, $\lfp\widetilde{\Phi} = \lfp\Phi$.
  The dual result holds for $\gfp$, $\max$, and $g\leq\gfp\Phi$ in $\dd_1$.
\end{lemma}

\begin{proof}[Sketch]
  First, $\widetilde{\Phi} \sqsubseteq \Phi$,
  hence $\lfp\widetilde{\Phi} \leq \lfp\Phi$.
  Then, observe that $\lfp\Phi$ is also fixpoint of
  $\widetilde{\Phi}$
  and that there are no strictly smaller fixpoints (by contradiction).
\end{proof}
\begin{remark}
  This proves Proposition~\ref{prop:hybrid-verif} by noticing that,
  for a candidate solution $g$ of equation $\Phi$,
  if $\lfp\Phi\leq g$ in $\dd_1$ and $\Phi g \leq g$ in $\dd_2$,
  then $\widetilde{\Phi}(g)\leq g$ (with the notations of the lemma),
  so that $g\geq\lfp\widetilde{\Phi}=\lfp\Phi$ (everywhere).
\end{remark}

Such proof principle may be used for interactive bound search
approaches, mixed with local repair as sketched in
Section~\ref{subsubsec:interactive-approach}.
Additionally, it provides safety guarantees for substitution of
partial bounds or solutions (known to be valid only in a subdomain),
which is particularly useful to define a compositional approach to
solving systems of equations.

\textbf{Subdomain-wise reasoning, CFR and CFA for equations.}
\label{subsubsec:subdom-wise-reasoning}

It can be useful to consider control-flow graphs for equations, and
follow the topological order when searching candidates for subdomains,
generalising classical approaches used for systems of
equations~\cite{caslog-short}.

\begin{example}
  Consider equation \texttt{\multiphaseonename} of
  Table~\ref{table:benchmarks}, with subdomains $\dd_1,\dd_2,\dd_3$
  corresponding to $i\geq n$, $i<n \wedge r>0$ and $i<n \wedge r = 0$
  respectively.
  The following dependency graph is obtained by asking, for each
  $\vn$, what are the values $\vn'$ such that
  $\Phi(f)(\vn)$ may depend on $f(\vn')$. %
  Up to orientation reversal, this is just the control flow of the program
  from which this equation originates (Fig 1.12,
  \cite{montoya-phdthesis}).
\begin{center}
\begin{tabular}{c@{\hspace{4em}}c}
\begin{tikzcd}[column sep=1.5em]
                                                             & \dd_1 \arrow[rd] &                                                              \\
\dd_2 \arrow[loop, distance=2em, in=235, out=305] &                  & \dd_3 \arrow[ll] \arrow[loop, distance=2em, in=235, out=305]
\end{tikzcd}
&
\begin{lstlisting}[language=c, numbers=none, basicstyle=\scriptsize]
uint __cost = 0;

while(i < n){
  if(r > 0){
    i=0;
    r--;
  }else{
    i++;
  }
  __cost++;
}
\end{lstlisting}
\end{tabular}
\end{center}
  In the interpretation of
  Section~\ref{subsubsec:interactive-approach}, cost flows from $\dd_1$ to $\dd_3$
  to $\dd_2$, and we may thus perform our bound search in this order.
  For example, take the candidate $g:(i,n,r)\mapsto \max(0,n+r-i+1)$.
  We have
  {\footnotesize
  \begin{equation*}
    \Phi g=(i,n,r)\mapsto\begin{cases}
      0 &\text{in }\dd_1\\
      n+r+1 &\text{in }\dd_2\\
      n-i+1 &\text{in }\dd_3,
    \end{cases}
  \end{equation*}}%
  hence $\Phi g \leq g$ in $\dd_1 \cup \dd_3$,
  but $\Phi g \geq g$ in $\dd_2$,
  we thus cannot apply our proof principle directly.
  However, since there are no recursive calls to $\dd_2$ coming from $(\dd_1 \cup \dd_3)$,
  we can restrict our attention to $\dd_1 \cup \dd_3$,
  and the observation above proves that
  $\fsol \leq g$ on $\dd_1 \cup \dd_3$.
  To go further, we may thus later search another candidate upper
  bound only in $\dd_2\to\rr$, using the bound $g_{|\dd_1 \cup \dd_3}$
  in recursive calls.
  (In this particular example, admittedly a bit contrived, we can also
  realise that if we had searched a bit further, we would have
  obtained a global bound directly, since $\Phi\Phi g \leq \Phi g$
  everywhere, hence $\fsol \leq \Phi g$ on the whole domain.)
\end{example}

Subdomain-wise reasoning also suggests to take particular attention to
\emph{interfaces}, where inductive proofs are most likely to fail if
we are not careful. As this intuition is particularly useful to define
good heuristics (e.g. for input sample generation in the approach of
Section~\ref{subsec:optim-approach}), we define this notion here.

\begin{definition}[Subdomain interface]
  \label{defi:subdom-interface}
  Let $\Phi:\rr^\dd\to\rr^\dd$ be an equation, and $\dd_1,
  \dd_2\subseteq\dd$ two disjoint subdomains.
  We say that a value $\vn\in\dd_1$ is at the \emph{interface} between
  $\dd_1$ and $\dd_2$ whenever there exists two functions
  $f_1,f_2:\dd\to\rr$, identical outside of $\dd_2$, such that
  $\Phi(f_1)(\vn)\neq\Phi(f_2)(\vn)$.
\end{definition}

\section{Candidate Generation}
\label{sec:cand-generation}

The order ``equation-as-operator'' viewpoint and the theoretical
framework developed in this paper lead to many avenues for development
of new recurrence solving algorithms, especially in the form of
``guess, check, and repair/refine'' loops (including abstract
interpretation approaches).

For brevity, we focus on a particular instance, that
has been implemented as a proof-of-concept.
It is an effective although simple approach based on dynamic invariant
generation and constrained optimisation, presented in
Section~\ref{subsec:optim-approach}.
We also sketch other possible approaches in
Section~\ref{subsec:alternative-approaches} (which we believe offer
other advantages, particularly in terms of generalisability), and
discuss them further 
in the longer version of the paper. 
\mhnote{Same}

\vspace{-0.55em}
\subsection{Dynamic optimisation-based approach}
\vspace{-0.15em}
\label{subsec:optim-approach}

\noindent While we do not \emph{a priori} have access to a closed-form
for the solution $\fsol=\lfp \Phi$ to our equation, we can always
observe its projection on finite-dimensional subspaces of
the form $\rr^\inputset$ with $\inputset\subseteq\dd$ finite.
Indeed, we can take advantage of the fact that we are working in a
\emph{discrete} setting where $\dd\subseteq\zz^\dimdom$, with
equations that terminate in finite time, to obtain the values of
$\fsol$ on each $\vn\in\inputset$ by evaluating $\Phi$ as a program
($\fsol(\vn) = \Phi^\beta(\bot)$ for some $\beta<\omega$).
Note that this would not be immediate in non-discrete settings, e.g.
for  differential equations.

This observation allows us to draw upon techniques developed in the line of
work on \emph{dynamic invariant analysis}, e.g.
\cite{nguyen2012-shortest,nguyen2014-shortest,ishimwe2021:dynaplex-oopsla,nguyen2022:invariants-tse,mlrec-tplp2024-preprint},
which proposes to identify likely invariants from a sample of program traces,
before attempting to prove them.
Note that we apply these techniques to \emph{recurrences viewed as programs},
i.e. executable abstractions of programs, instead of programs
directly.

Thus, we generate a \emph{training set}
$\trainingset = \{(\vn, \fsol(\vn))\,\big|\,\vn\in\inputset\}$,
by executing $\Phi$.

We generalise the approach of~\cite{mlrec-tplp2024-preprint}, which
uses regression techniques on $\trainingset$ to search for $\fsol$
among various model spaces, including finite-dimensional affine
subspaces of $\rr^\dd$ generated by \emph{base functions} taken from a
finite set $\basefuns$ hoped to be representative of common size/cost
functions. Our generalisation allows the \emph{guess} stage to search
through the space of all \emph{bounds} on $\fsol$, which are a priori
easier to guess in general, and thus makes it possible to obtain
bounds where~\cite{mlrec-tplp2024-preprint} only obtains unsafe
approximations.

We do so by leveraging the order-theoretical framework developed in
the rest of this paper, and perform dynamic invariant generation on
the set of observations $\trainingset$ using \emph{constrained}
optimisation methods, that enforce \emph{accuracy}, \emph{parsimony},
\emph{safety} and local \emph{inductivity}.
Generalisations of the approach can easily be imagined, but for
simplicity, efficiency, and to obtain a first prototype, we present an
instantiation based on quadratic programming (QP), using affine model
spaces (made of affine combinations of base functions). The full
version presented here is a priori limited to \emph{affine} equations,
but can be generalised to the full grammar of
Fig.~\ref{eqs:grammar-1}, or directly applied without
\emph{inductivity} constraints for simplicity.
For the sake of exposition, the method is presented for upper bounds;
the case of lower bounds is symmetrical.

We first present our objectives in a general setting, before
presenting the instantiation with QP, affine equations and affine
template spaces, which can be efficiently handled in practice.

We want the optimisation-based method to search for ``good'' candidate bounds
$\candf$ of $\fsol=\lfp\Phi$, inside some tractable model space
$\modelspace\subseteq(\dd\to\rr)$, ideally such that they can be checked by our method,
i.e. such that they are \emph{inductive}: $\Phi\candf \leq \candf$.
We formalise these requirements in the following way.
\begin{enumerate}
    \item \label{enum:accuracy}
      \textbf{Accuracy} on the training set.
      We want to minimise the $\ell_2$ norm of $\candf-\fsol$ on $\rr^\inputset$,
      i.e. minimise
      $||(\candf-\fsol)_{|\inputset}||^2_2 = \sum_{\vn\in\inputset}(\candf(\vn)-\fsol(\vn))^2.$
    \item \label{enum:parsimony}
      \textbf{Parsimony}.
      In order to avoid overfitting, we need to penalise model complexity,
      i.e. minimise some measure of complexity of $\candf$.
    \item \label{enum:bound}
      \textbf{Safety}.
      We search for functions that are bounds on all of $\dd$,
      so candidates should at least be bounds on $\inputset$, and we impose
      $\forall \vn\in\inputset,\, \fsol(\vn)\leq\candf(\vn).$
    \item \label{enum:inductivity}
      Local \textbf{Inductivity}.
      Ideally, we search for bounds that we will be able to prove, i.e. such
      that $\Phi\candf \leq \candf$.
      To extract a finite number of constraints from this, we can at least
      require that this inequality holds on the training set:
      $\forall \vn\in\inputset,\, \Phi(\candf)(\vn)\leq\candf(\vn).$
\end{enumerate}

\noindent These goals can be achieved efficiently if we restrict
ourselves to model spaces and candidates that can be handled by QP
methods. Generalisations to larger template spaces are possible with
other optimisation methods, at the cost of efficiency and theoretical
guarantees.

We choose a finite number $\numfeatures$ of \emph{base functions}
$\basefuns\subset\rr^\dd$ (intended to represent common complexity
orders),
we impose $(\vn\mapsto 1)\in\basefuns$ to handle biases uniformly,
and we set
$\modelspace = \Vect(\basefuns)
  = \big\{\sum_{i=1}^p \alpha_i f_i\,\big|\,\forall i,\,\alpha_i\in\rr,\,f_i\in\basefuns\big\}$.
For simplicity, for a generic candidate $\candf\in\modelspace$ of the
form $f_{\va} := \sum_{i=1}^p\alpha_i f_i$, we refer to the vector of
coefficients $\va$ as ``the candidate'' itself, although a single
function may admit multiple representations as a linear combination of
elements of $\basefuns$.

To simplify the notation, let $\vysol$ and $\vyi$ respectively denote the
projection on $\inputset\to\rr$ of $\fsol$ and each $f_i\in\basefuns$, i.e.
$\vysol=(\fsol(\vn))_{\vn\in\inputset}$ and $\vyi=(f_i(\vn))_{\vn\in\inputset}$.
Similarly we write $\vy_{\va}$, or simply $\vy$, to represent the projection
$f_{\va}$, so that $\vy = \sum_i \alpha_i\vyi$.
This can be rewritten as the matrix multiplication $\vy = \ymat \va$, where
$\ymat$ is a $\numtraining\times\numfeatures$ matrix whose columns are the
$\vyi$, with $\numtraining=\Card(\inputset)$.

Once this is done, the goals presented above can be handled as follows.
    \textbf{Accuracy} corresponds to minimisation of the Euclidean norm
    $||\vysol-\ymat\va||_2^2$.
    To achieve \textbf{parsimony}, we can use lasso regularisation, i.e. add
    a penalisation of the $\ell_1$ norm $||\va||_1:=\sum_i |\alpha_i|$, which
    can serve as an approximation of the $\ell_0$ pseudo-norm
    $||\va||_0:=\Card\{i\,|\,\alpha_i\neq 0\}$, as described in~\cite{Hastie-15a}.
    Local \textbf{safety} is simply represented as a linear
    inequality $\vysol\leq\ymat\va$.
    Local \textbf{inductivity} corresponds to
    $\Phi(f_{\va})_{|\inputset} \leq (f_{\va})_{|\inputset}$, which is
    a bit more complicated as it is not in general a conjunction of
    affine inequalities on $\va$.
    \lrnote{TODO: Improve notations $\vyip$, $\ypmat$ and $\vphib$?}
    The situation simplifies if we additionally assume that $\Phi$ is
    \emph{affine}, in which case we use Proposition~\ref{prop:aff-on-aff}.
    Calling $\vyip$ the projections of $\Phil f_i$ on $\inputset\to\rr$,
    $\ypmat$ the $\numtraining\times\numfeatures$ matrix whose columns are the
    $\vyip$, and $\vphib$ the projection of the bias $\Phi(\vzero)$, we obtain
    $\Phi(f_{\va})_{|\inputset} = \ypmat \va + \vphib$, so that the
    inequality above can be rewritten as the polyhedron
    $\ypmat\va + \vphib \leq \ymat\va$.
Putting all of this together, and calling $\lambda\in\rr_+$ the hyperparameter
associated to lasso regularisation, we obtain the following
constrained optimisation problem.
\begin{equation}
\begin{gathered}
  \min_{\va\in\rr^\numfeatures}||\vysol-\ymat\va||_2^2+\lambda||\va||_1\\[-2pt]
  \text{under the constraints}
  \;
  \vysol\leq\ymat\va %
  \;\text{and}\;
  \ypmat\va + \vphib \leq \ymat\va. %
\end{gathered}
\vspace{-0.5em}
\end{equation}

\noindent This can be transformed into a QP problem. We give the full
version directly for completeness, but it can be simplified if we
ignore terms enforcing parsimony, safety and/or local
inductivity.
\vspace{-0.5em}
$$\min_{\substack{\vx\in\rr^\numfeatures\times\rr^\numfeatures\\G\vx\leq\vec{h}}}\vx^TP\vx
+ q^T\vx,\;\;\;\text{where}$$
\vspace{-2em}
{\footnotesize
\begin{equation*}
  \vx = \begin{pmatrix}
    \va \\ \vec{u}
  \end{pmatrix},
  \;
  P = \begin{pmatrix}
    \ymat^T \ymat && \vzero \\
    \vzero        && I_{\numfeatures}
  \end{pmatrix},
  \;
  q = \begin{pmatrix}
    -\ymat^T\vysol \\ %
    \lambda \vec{1}_{\numfeatures}
  \end{pmatrix},
  \;
  G = \begin{pmatrix}
           -\ymat  && \vzero \\ %
     \ypmat-\ymat  && \vzero \\ %
    -I_{\numfeatures} && -I_{\numfeatures}\\ %
     I_{\numfeatures} && -I_{\numfeatures}   %
  \end{pmatrix},
  \;
  h = \begin{pmatrix}
     \vysol \\ %
    -\vphib \\ %
    \vzero \\ %
    \vzero    %
  \end{pmatrix}.
\end{equation*}
}

\noindent As mentioned in Section~\ref{subsec:prelim-qp}, to our
knowledge, the use of such bounds and inductivity constraints is
novel in this context, but the rest of the setup is classical.
The terms $\ymat^T \ymat$ and $\ymat^T\vysol$ correspond to the classical QP
representation of linear regression problems. The first two lines of $G$ and $h$
correspond respectively to bound constraints and local inductivity
constraints. Finally, $\vec{u}$ are auxiliary variables used to obtain a
QP-encoding of the $\ell_1$ minimisation term.

Note that in practice, to gain precision and as in~\cite{mlrec-tplp2024-preprint},
we use a modified model space composed of piecewise functions, with subdomains
extracted from the definition of $\Phi$.
More precisely, for $\Phi$ defined piecewise as in
Equation~\ref{eq:piecewise-generic} with $\dd_i=\{\vn\,|\,\varphi_i(\vn)\}$,
we classify the subdomains into non-recursive (if $e_i(f,\vn)$ does not depend
on f) and recursive (otherwise or if nothing can be proven), and use
{\footnotesize
\begin{equation*}
  \widetilde{\modelspace} = f_{base} +
    \bigoplus_{\dd_i\;\text{recursive subdomain}}
    \Vect\Big(\Big\{
         \vn\mapsto f(\vn)\cdot \vec{1}_{\dd_i}(\vn)
      \;\Big|\;
      f\in\basefuns\Big\}\Big),
  \vspace{-0.5em}
\end{equation*}}

\noindent where $f_{base}(\vn) = \Phi(\vzero)(\vn)$ if $\vn\in\dd_i$ for $\dd_i$
non-recursive, and $f_{base}(\vn)=0$ otherwise,
and $\vec{1}_{\dd_i}:\vn\mapsto\ite(\varphi_i(\vn),1,0)$ is the
characteristic function of $\dd_i$.
In other words, $\widetilde{\modelspace}$ is simply composed of functions which are
piecewise combinations of the base functions from $\basefuns$ in the recursive
subdomains, and which are equal to the solution $\fsol$ on the non-recursive
domains, in which this solution is explicitly available as a base case.
$\widetilde{\modelspace}$ is then an affine vector space of dimension
$r_{\texttt{rec}}\times\Card(\basefuns)$, where $r_{\texttt{rec}}$ is
the number of recursive subdomains. $\widetilde{\modelspace}$ can also
simply be viewed as a shifted version (by $f_{base}$) of a model space
$\Vect(\widetilde{\basefuns})$ with a modified set of base functions
$\widetilde{\basefuns}$ containing all the restrictions to each
recursive subdomain of each base function in $\basefuns$.
The expressions presented above generalise easily to this case by
adding a bias term $\vb$ so that $\vy = \ymat\alpha+\vb$, and
reflecting this in the definition of $\vphib$ as well as $q$ and $h$.
Such strategy, that improves precision by using piecewise templates on
selected subdomains, can easily by improved by using more refined
strategies for subdomain selection than the purely syntactic method
mentioned here, e.g. CFA/CFR techniques from~\cite{DomenechGG19-short}.

At this point, a QP problem is set up. It is handed to a QP solver,
which performs candidate generation, outputting an $\va$ which
corresponds to a candidate $\candf=f_{\va}$. Depending on QP
solver used, the resulting $\candf$ may be expressed on floating
points rather than rational numbers, which hinders verification.
Strategies to handle this issue are discussed along with our
implementation in Section~\ref{sec:implementation-and-evaluation}.

\subsection{Alternative approaches}
\label{subsec:alternative-approaches}

\hspace{\parindent} \textbf{Interactive and dichotomy-like approach.}
\label{subsubsec:interactive-approach}
A first simple idea is to view Theorem~\ref{th:proof-principle} and
its extensions as ``a proof principle/system for humans''.
Indeed, it leads to very short proofs for both introductory and
non-trivial examples, such as
$\Phi(f)(n)=\ite(n\leq1,0,n + \max_{0 \leq k \leq
  n-1}(f(k)+f(n-1-k))$,
i.e. worst-case non-deterministic quicksort, where the $\max$ over $k$
presents unique challenges.
Using Theorem~\ref{th:proof-principle}, it is immediate to prove that
$n^2$ is an upper bound, since for $n\geq2$,
$\Phi(n\mapsto n^2)(n)=n+\max_{0 \leq k \leq n-1}(k^2+(n-1-k)^2)=n+(n-1)^2\leq n^2.$
Large parts of such proofs can be mechanised. This leads to an
interactive proof system, where the user iteratively proposes
candidate bounds, and the system either proves them or provides
information to improve the guess. Hybrid verification and
local repairs are particularly valuable here.
To automate manual exploration strategies, further research of
the structure of candidate (sub)lattices is required.

\textbf{Quantifier elimination: a partial result on completeness.}
As mentioned in Section~\ref{subsec:prelim-funcomp}, the first-order theory
of reals (as an ordered field) is decidable, and effective algorithms such
as CAD are available to perform quantifier elimination.
This provides partial results on completeness, in the sense that we
can obtain all \emph{inductive} (piecewise) polynomial bounds (locally
of shape $f_{\va}(\vn)=\sum_i \alpha_i \prod_j n_j^{e_{i,j}}$), for
equations that preserve the set of piecewise polynomials, after
integer-to-real relaxation.
Indeed, the formula
$\exists\va\forall\vn\,\Phi(f_{\va})(\vn)\leq f_{\va}(\vn)$
can be fully handled on the reals, and CAD produces
a logical formula on $\va$ that can be viewed as the representation of
an element of the domain of (sets) of polynomial bounds, as described
in Appendix~\ref{sec:A-bounds}.
A heuristic to temper the loss of precision from
integer-to-real relaxation is to use a finite number of disjunctions
to enforce integrality close to the origin and to interfaces between
the $\dd_i$, where most spurious counterexamples are found (e.g. $1/2$
for $n\leq n^2$).

As CAD does not scale well, this approach is mainly valuable for the
understanding it provides, but it can be applied practically when the
number of template parameters is moderate. For example, using CAD
through \texttt{Mathematica}'s \texttt{Resolve}, we obtain the exact
solution to \texttt{\sumoscname} in $20$ seconds, and recover the
ideal solution of Table~\ref{table:runnex_results} when fixing the
leading term for our running example.

\textbf{Abstraction-based approach, Abstract Iteration.}
\label{subsubsec:abstract-iteration}
We conclude this section with an idea, included here to enable
discussion, and as a foundation for future work, that we believe opens
interesting research avenues.
Since recurrence equations can be viewed as programs and thus as
operators $\Phi:\concretedom\to\concretedom$ on a concrete domain
$\concretedom=\dd\to\rri$, whose semantics/solution is $\lfp\Phi$, it
is very tempting to design \emph{abstract domains} $\abstractdom$ of
abstract functions and obtain sound
$\Phi^\sharp:\abstractdom\to\abstractdom$ via \emph{abstract transfer
functions}, computed on simpler operators used to syntactically
construct $\Phi$.
From this, upper bounds on solutions, can be obtained thanks to
fixpoint transfer theorems, as limits of a Kleene sequence
$(\Phi^\sharp)^\beta(\bot^\sharp)$ -- possibly using widenings. Such
computation may be called \emph{abstract iteration} to reflect the
intuition of executing an abstract recurrence equation.

Preliminary versions on these ideas have been described in the
technical report~\cite{rustenholz-aars-msc}, using domains
$\boundaries$-bounds presented in Appendix~\ref{sec:A-bounds},
and can be extended to various template spaces.
However, a complete description of the abstractions and transfer
functions designed as extension of this line of work is outside the
scope of this paper, and will be reported upon in future work.

Abstract iteration is illustrated in Fig.~\ref{fig:absiter-nested},
on a small example difficult to handle deductively.
The analysis in this example simply uses affine bounds, and could thus
be performed with polyhedra,\footnote{This %
(interprocedural abstract interpretation by polyhedra+disjunctions to
obtain affine bounds, if possible, on solutions of equations
represented as programs) can be considered ``folklore'', but we are
not aware of full descriptions in the literature.}
but the value of the ``$\boundaries$-bound abstract domain'' viewpoint
lies in the fact that it can be extended beyond affine relationships,
hence beyond programs limited to linear complexity. Our intuition and
experience is that transfer functions are easier to design in this
context than for general numerical domains, thanks to the
``relational-to-functional constraint abstraction'' mentioned in
Proposition~\ref{prop:rel2fun-cnstr}.

In comparison with the approach chosen for our proof-of-concept
implementation, clear advantages of abstraction interpretation methods
are liveness guarantees,
lower sensibility to numerical errors,
and absence of need for sampling.
However, they appeared to require extra implementation effort, for the
design of transfer functions and widenings adapted to each family of
templates, although this might be simplified by integrating transfer
function synthesis methods (on basic operators and small compositions
of them).

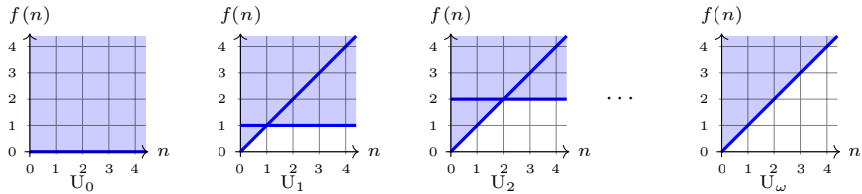
\begin{figure}[h!]
\vspace{-2em}
\begin{tabular}{p{2.7cm} p{2.7cm} p{3.5cm} p{3cm}}
{\begin{tikzpicture}[scale=0.35]
    \begin{scope}[transparency group]
        \begin{scope}[blend mode=multiply]
  \draw[very thin,color=gray] (-0.1,-0.1) grid (4.4,4.4);
  \draw[->] (-0.2,0) -- (4.5,0) node[right] {\scriptsize $n$};
  \draw[->] (0,-0.2) -- (0,4.5) node[above] {\scriptsize $f(n)$};
  \foreach \x in {0,...,4} {
        \node [anchor=north] at (\x,-0.5-1ex) {\tiny $\x$};
    }
    \foreach \y in {0,...,4} {
        \node [anchor=east] at (-0.5-0.5em,\y) {\tiny $\y$};
    }
    \draw[-, very thick, blue] (0, 0) -- (4.4, 0);
    \fill[fill=blue!20] (0,0)--(4.4,0)--(4.4,4.4)--(0,4.4);
        \end{scope}
    \end{scope}
    \node at (2, -1.2) {\scriptsize $\mathrm{U}_0$};
\end{tikzpicture}}
&
{\begin{tikzpicture}[scale=0.35]
    \begin{scope}[transparency group]
        \begin{scope}[blend mode=multiply]
  \draw[very thin,color=gray] (-0.1,-0.1) grid (4.4,4.4);
  \draw[->] (-0.2,0) -- (4.5,0) node[right] {\scriptsize $n$};
  \draw[->] (0,-0.2) -- (0,4.5) node[above] {\scriptsize $f(n)$};
  \foreach \x in {0,...,4} {
        \node [anchor=north] at (\x,-0.5-1ex) {\tiny $\x$};
    }
    \foreach \y in {0,...,4} {
        \node [anchor=east] at (-0.5-0.5em,\y) {\tiny $\y$};
    }
    \draw[-, very thick, blue] (0, 0) -- (4.4,4.4);
    \draw[-, very thick, blue] (0, 1) -- (4.4,1);
    \fill[fill=blue!20] (0,0)--(1,1)--(4.4,1)--(4.4,4.4)--(0,4.4);
        \end{scope}
    \end{scope}
    \node at (2, -1.2) {\scriptsize $\mathrm{U}_1$};
\end{tikzpicture}}
&
{\begin{tikzpicture}[scale=0.35]
    \begin{scope}[transparency group]
        \begin{scope}[blend mode=multiply]
  \draw[very thin,color=gray] (-0.1,-0.1) grid (4.4,4.4);
  \draw[->] (-0.2,0) -- (4.5,0) node[right] {\scriptsize $n$};
  \draw[->] (0,-0.2) -- (0,4.5) node[above] {\scriptsize $f(n)$};
  \foreach \x in {0,...,4} {
        \node [anchor=north] at (\x,-0.5-1ex) {\tiny $\x$};
    }
    \foreach \y in {0,...,4} {
        \node [anchor=east] at (-0.5-0.5em,\y) {\tiny $\y$};
    }
    \draw[-, very thick, blue] (0, 0) -- (4.4,4.4);
    \draw[-, very thick, blue] (0, 2) -- (4.4,2);
    \fill[fill=blue!20] (0,0)--(2,2)--(4.4,2)--(4.4,4.4)--(0,4.4);
        \end{scope}
    \end{scope}
    \node at (6.5, 2) {$\cdots$};
    \node at (2, -1.2) {\scriptsize $\mathrm{U}_2$};
\end{tikzpicture}}
&
{\begin{tikzpicture}[scale=0.35]
    \begin{scope}[transparency group]
        \begin{scope}[blend mode=multiply]
  \draw[very thin,color=gray] (-0.1,-0.1) grid (4.4,4.4);
  \draw[->] (-0.2,0) -- (4.5,0) node[right] {\scriptsize $n$};
  \draw[->] (0,-0.2) -- (0,4.5) node[above] {\scriptsize $f(n)$};
  \foreach \x in {0,...,4} {
        \node [anchor=north] at (\x,-0.5-1ex) {\tiny $\x$};
    }
    \foreach \y in {0,...,4} {
        \node [anchor=east] at (-0.5-0.5em,\y) {\tiny $\y$};
    }
    \draw[-, very thick, blue] (0, 0) -- (4.4,4.4);
    \fill[fill=blue!20] (0,0)--(4.4,4.4)--(0,4.4);
        \end{scope}
    \end{scope}
    \node at (2, -1.2) {\scriptsize $\mathrm{U}_\omega$};
\end{tikzpicture}}
\end{tabular}
\vspace{-2em}
    \caption{\footnotesize
    Abstract iteration with sets of affine bounds on benchmark \texttt{\nestedname}.
    Candidate upper bounds (light blue) are progressively removed by
    abstract iterations. It is sufficient to apply $\Phi^\sharp$ on extremal
    bounds (dark blue). Between $\mathrm{U}_1$ and $\mathrm{U}_2$, $(n\mapsto n)$ is fixed,
    which shows that it is a correct upper bound on $\fsol$. When the
    fixpoint $\mathrm{U}_\omega$ is reached, only correct upper bounds
    are left.
    In this case, the ideal solution is obtained.
    \label{fig:absiter-nested}}
\vspace{-1em}
\end{figure}

\vspace{-2.5em}
\section{On Candidate Repair}
\vspace{-0.5em}
\label{sec:repair}
\lrnote{Has been compressed but would benefit from a pass to make it
  flow and cut details, this is a bit cramped}

Consider a candidate (e.g. upper) bound $\candf$ on the solution to $\Phi$.
When the inductivity check fails, i.e. $\Phi\candf\not\leq\candf$,
but the candidate $\candf$ seems ``close to be correct'' (e.g. if
$\fsol(\vn)\leq\candf(\vn)$ on a large sample, or if $\Phi\candf-\candf$ stays close to zero when it is
positive), it would be useful to ``repair'' this candidate by applying
small transformations to $\candf$.
This is particularly relevant for the optimisation-based
approach presented in Section~\ref{subsec:optim-approach}, where
numerical errors may make a candidate uncheckable even though it is
extremely close to a tight inductive upper bound.
More generally, we could imagine using these repair strategies on
candidates obtained by other means (e.g. interactive loops of
Section~\ref{subsubsec:interactive-approach}), and integrate them
in ``guess, check, and repair/refine'' bound inference loops.

To reformulate our question more precisely, can we find a
transformation $\Psi:\rr^\dd\to\rr^\dd$, such that
$(\Phi\Psi\candf-\Psi\candf)\leq\vzero$? Can this $\Psi$ be obtained
by combining simple transformation rules $\Psi_1,\dots,\Psi_k$, or can
we at least obtain safety conditions to ensure that such $\Psi_j$ will
not degrade checkability, although they may degrade precision?
The concepts introduced in this paper are useful
tools to investigate these questions, through the analysis of
variations $(\Phi\Psi\candf-\Psi\candf)-(\Phi\candf-\candf)$ of
inductivity measures -- a transformation \emph{helps} (for upper
bounds) when this quantity is negative.
Three important families of transformations are (briefly) introduced below.\\

\textbf{Local repair -- $\Psi: f \mapsto \mathrm{ite}(\varphi,g,f)$.}
\label{subsec:repair-local}
Consider a candidate upper bound $\candf$ such
$\Phi\candf\not\leq\candf$ only on a finite subdomain $\dd_1$, i.e.
such that the inductivity proof \emph{succeeds} on a cofinite set
$\dd_2:=\dd\setminus\dd_1$, e.g. ``at infinity''.
We can then attempt to \emph{repair} $\candf$ by computing a correct
upper bound $g$ on $\dd_1$,
(by executing $\Phi$ on the finite sample $\dd_1$), and
by replacing the value of $\candf$ by that of $g$ in $\dd_1$, i.e.
$\Psi:f\mapsto\mathrm{ite}(\vec{1}_{\dd_1},g,f)$.
Using Section~\ref{subsec:proof-extended}, we can show that to check
the inductiveness of such a repair, it is sufficient to check the
result ($\Phi\Psi\candf\leq\Psi\candf$ or $\fsol\leq\Psi\candf$) on
the \emph{interface} of $\dd_1$ and $\dd_2$.
This empowers interactive guess/repair loops presented in
Section~\ref{subsubsec:interactive-approach}.
An illustration is given in Fig.~\ref{fig:local-repair-ex}.

\begin{figure}
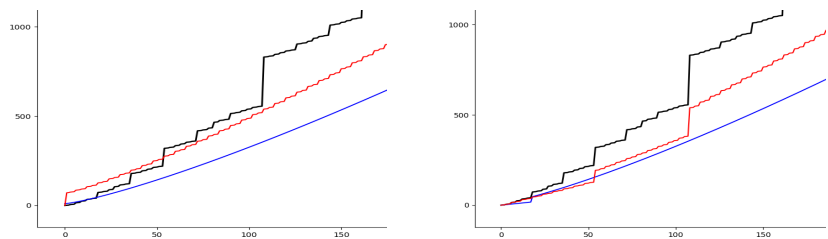

  \centering
  \begin{minipage}[c]{0.47\linewidth}
    \includegraphics[width=5cm,height=3cm]{\figpath/candidate\_repair\_ex1\_before}
  \end{minipage}
  \begin{minipage}[c]{0.47\linewidth}
    \includegraphics[width=5cm,height=3cm]{\figpath/candidate\_repair\_ex1\_after\_2}
  \end{minipage}

  \caption{\footnotesize
  Local repair of candidate lower bound $\candf:n\mapsto n^{1.25}+10$
  on the solution to the equation
  $\Phi(f)(n)=\ite(n=0,1,2f(\lfloor n/3 \rfloor)+5f(\lfloor n/6 \rfloor)+n)$.
  Plots represent
  $\fsol$ (black), $\candf$ (blue, left), $\Phi\candf$ (red, left),
  $\Psi\candf$ (blue, right), and $\Phi\Psi\candf$ (red, right).
  $\candf$ fails the inductivity check in $0$.
  Sample evaluation of $\fsol-\candf$ allows to guess
  that $\candf$ is a correct lower bound for $n\geq 25$.
  This suggests a fix such as
  $\Psi\candf:n\mapsto\ite(n<25,n,\candf(n))$.
  $\Psi\candf$ is a correct lower bound to $\fsol$. To check this, we
  verify that $\Psi\candf(n)\leq\fsol(n)$ on the \emph{interface}
  $[25,149]$, and that $\Phi\candf\geq\candf$ for $n\geq150$ (left).
  \label{fig:local-repair-ex}
  }
  \vspace{-1em}
\end{figure}

\textbf{Additive repair -- $\Psi: f \mapsto f+g$.}
\label{subsec:repair-additive}
Additive transformations %
interact particularly well with affine equations $\Phi$, where
$(\Phi(f+g) - (f+g)) - (\Phi f - f) = \Phil g - g$ for all
for $f,g$.
In other words, to repair a candidate upper bound,
it suffices to choose a $g$ that corrects the defect of inductivity of
$f$, i.e. such that $g - \Phil g \geq \max(\vzero, \Phi f - f)$.
This is directly connected to our discussion on the
geometry of $\PostFix$ and $\PreFix$ in Section~\ref{sec:cones}, and
to the description of their cones via (generalised) ranking functions
for $\Phi$.
Indeed, elements $g$ such that $\Phil g - g \leq 0$ can be constructed
from such generalised ranking functions, and we typically have a very
good understanding of the values taken by $g - \Phil g$.
This suggests that we can attempt to choose such $g$ appropriately,
to fix $\candf$ ``just enough'', i.e. to repair the offending positive
values in $(\Phi\candf-\candf)$ without degrading too much the precision of
the new candidate.
Such choice of repair obtained by precise accounting of the offending
numerical values in $(\Phi\candf-\candf)$ may be called a ``quantitative repair
strategy''.

The interested reader can ponder (precise)
repair of the following candidate in the simple case of our running
example, obtained numerically, and incorrect because of an unfortunate
rounding.
  {\vspace{-.5em}\footnotesize
  \[\candf(n,c)=\begin{cases*}
       \texttt{0.5*n**2 + 3.47*n + 297.02} & if $n>0 \wedge c\geq 100$\\[-2pt]
       \texttt{0.5*n**2 + 3.47*n + 2.97*c} & if $n>0 \wedge c <   100$\\[-2pt]
       \texttt{c}                         & if $n=0$,
    \end{cases*}\vspace{-.5em}\]}

\noindent For the prospect of future implementation, we note that in our manual
experiments, it was particularly useful to have access to both
\emph{piecewise} $g$ (e.g. provided by~\cite{Urban-phd} but
not~\cite{BagnaraMPZ10TR}) to localise the fix, and to have access to
large \emph{sets} of $g$ (e.g. as in~\cite{BagnaraMPZ10TR} but
not~\cite{Urban-phd}) to be able to choose and combine them. Automated
inference and combination of such $g$ is an interesting direction for
future work.\\

\textbf{Repair by precompositions: shift and rescale -- $\Psi:f\mapsto f\circ\phi$.}
\label{subsec:repair-precomp}
A final convenient repair strategy is to simply shift or rescale the
graph of the candidate $\candf$, i.e. to apply a precomposition
$\Psi:f\mapsto f\circ\phi$.
Compared to additive repair, this allows to more naturally conserve
and use the shape of the candidate expression $\candf$.
Moreover, compared to local repair, which can be difficult to apply to
some multivariate cases if $\candf$ is verified in an infinite but not
cofinite set, repair by precomposition can often eliminate this issue.
Relevant examples of shifts can be observed in two dimensions for
\texttt{\openzipname}, and for non-linear recursion, with speedups, as
already presented on \texttt{\exptwoname}.

To give a feeling of the kind of safety conditions that can be
computed for such $\Psi$,
consider the example of a linear-recursive program with a single
variable, i.e. an equation $\Phi:=f\mapsto f\circ\varphi + P$
(we allow partial functions $\varphi:\nn\rightharpoonup\nn$, with
default value $0$ for $f\circ\varphi$). Then, for $\Psi =
-\circ\phi$, when both sign in the following tuple are identical, then
$(\Phi\Psi f - \Psi f) - (\Phi f - f)$ has the same sign.
{\vspace{-.5em}%
$$\big(\sgn(\Delta^2 f)\cdot\sgn(\id-\varphi)\cdot\sgn(\id-\phi),\;\sgn(\Delta f)\cdot\sgn(\id-\varphi)\cdot\sgn(\Delta(\id-\phi))\big)\vspace{-.5em}$$
}

\noindent Such results can be conveniently obtained via computations of
discrete calculus and operator algebra, focusing on commutators
($[a,b]:=ab-ba$) for basic operators built on discrete derivatives,
pre/postcomposition and products.

\section{Implementation (of the optimisation-based approach) and Experimental Evaluation}
\label{sec:implementation-and-evaluation}

As a proof-of-concept of the potential of our order-theory viewpoint
on recurrence equations for cost analysis, we have implemented the
instantiation described in Section~\ref{subsec:optim-approach},
which is submitted as the artifact for the paper.
We have performed an experimental comparison with state-of-the-art cost
analysers and recurrence solvers, which demonstrates that inductive
bounds more precise than the bounds obtained by other tools exist and
can be efficiently discovered.

\paragraph{\textbf{Prototype.}}

\begin{figure}
  \centering
\begin{adjustwidth}{-2em}{0em}
\resizebox{1.0\linewidth}{!}{
   \pgfdeclarelayer{background}
\pgfdeclarelayer{foreground}
\pgfsetlayers{background,main,foreground}

\tikzstyle{io}  = [
  chamfered rectangle, chamfered rectangle corners=north west,
  draw=cyan!80!black!100, fill=cyan!20,
  minimum width=3cm, minimum height=2.5cm,
  text centered, font=\sffamily
]

\tikzstyle{output-large}  = [
  chamfered rectangle, chamfered rectangle corners=north west,
  draw=red!80!black!100, fill=red!20,
  minimum width=4cm, minimum height=7cm,
  text centered, font=\sffamily
]

\tikzstyle{output}  = [
  chamfered rectangle, chamfered rectangle corners=north west,
  draw=red!80!black!100, fill=red!20,
  minimum width=3cm, minimum height=2cm,
  text centered, font=\sffamily
]

\tikzstyle{verif-output}  = [
  chamfered rectangle, chamfered rectangle corners=north west,
  draw=black!100, fill=black!30,
  minimum width=2cm, minimum height=1cm,
  text centered, font=\sffamily
]

\tikzstyle{large-process}=[
  rectangle, rounded corners,
  fill=yellow!20, draw=black!50,
  minimum width=11cm, minimum height=8cm,
  thick,
  font=\sffamily
]

\tikzstyle{process-of-large}=[
  rectangle, rounded corners,
  draw=green!50!black!100, fill=green!40,
  minimum width=3cm, minimum height=4cm,
  thick,
  font=\sffamily
]

\tikzstyle{medium-process}=[
  rectangle, rounded corners,
  fill=yellow!20, draw=black!50,
  minimum width=4cm, minimum height=3.5cm,
  thick,
  font=\sffamily
]

\tikzstyle{unimplemented-process}=[
  rectangle, rounded corners,
  draw=black!50, fill=green!30!black!10,
  minimum width=2.5cm, minimum height=2cm,
  densely dashed,
  font=\sffamily
]

\tikzstyle{io-component}=[
  rectangle, chamfered rectangle corners=north west,
  draw=orange!50!black!100, fill=orange!40,
  minimum width=2.5cm, minimum height=1cm,
  thick,
  text centered, font=\sffamily
]

\tikzstyle{process-component}=[
  rectangle, rounded corners,
  draw=orange!50!black!100, fill=orange!40,
  minimum width=2.5cm, minimum height=0.75cm,
  thick,
  text centered, font=\sffamily
]

\tikzstyle{unimplemented-process-component}=[
  rectangle, rounded corners,
  draw=orange!50!black!100, fill=orange!30!black!30,
  minimum width=2.25cm, minimum height=0.75cm,
  densely dashed,
  text centered, font=\sffamily
]

\tikzstyle{decision}=[
  diamond,
  draw=green!50!black!100, fill=green!40,
  minimum width=0.5cm, minimum height=0.5cm,
  text centered, font=\sffamily
]

\tikzstyle{arrow} = [thick,->]

\begin{tikzpicture}
  [every node/.style={align=center}, >=latex, scale=1]

\node (recurrence) [io]
   {\Large\textbf{Recurrence} \\ [1mm] \Large\textbf{Relation} \\ [1mm]
   \large$\begin{pmatrix}\texttt{Find }\candf\texttt{ s.t.}\\[1mm]\fsol\leq\candf\end{pmatrix}$};

\node (operator) [io, below=0.5cm of recurrence]
   {{\huge $\Phi$} \\ [2mm]
   \large$\begin{pmatrix}\texttt{Find }\candf\texttt{ s.t.}\\[1mm]\Phi\candf\leq\candf\end{pmatrix}$};

\node (guesser) [large-process, right=0.5cm of operator, yshift=1.5cm] { };
\path [color=black, font=\sffamily] (guesser.north)+(0,-2.5em) node (guesser-txt)
      {\Large \textbf{Guesser}};

\node (training) [process-of-large, right=1cm of operator, minimum height=5.5cm,yshift=1.5cm] { };
\path [color=black, font=\sffamily] (training.north)+(0,-2em) node (training-txt)
      {\large \textbf{Training}};
\node (execution) [process-component, below=0.1cm of training-txt]
      {\large \textbf{Execution}};
\node (input-generation) [process-component, below=0.3cm of execution]
      {\large \textbf{Input}\\\large Generation\\\large Strategies};
\node (templates) [io-component, below=0.3cm of input-generation]
      {\large \textbf{Templates}};

\node (featselect) [process-of-large, right=0.5cm of training] { };
\path [color=black, font=\sffamily] (featselect.north)+(0,-2em) node (featselect-txt)
      {\large \textbf{Feature} \\[1mm] \large \textbf{Selection}};
\node (execution) [process-component, below=0.3cm of featselect-txt]
      {\large \textbf{QP}\\[2mm] \large w/ Lasso};

\node (finalguess) [process-of-large, right=0.5cm of featselect] { };
\path [color=black, font=\sffamily] (finalguess.north)+(0,-2em) node (finalguess-txt)
      {\large \textbf{Final} \\[1mm] \large \textbf{Guess}};
\node (execution) [process-component, below=0.3cm of finalguess-txt]
      {\large \textbf{QP}\\[2mm] \large no Lasso};

\node (rounding) [medium-process, below=1cm of guesser] { };
\path [color=black, font=\sffamily] (rounding.north)+(0,-2em) node (rounding-txt)
      {\large \textbf{Rounding}};
\node (continued-fractions) [process-component, below=0.3cm of rounding-txt]
      {\large \textbf{Continued}\\[2mm] \large \textbf{Fractions}};
\node (rounding-dots) [color=black, font=\sffamily, below=0.3 cm of continued-fractions]
      {\huge \textbf{$\cdots$}};

\node (cand-float) [io, left=0.5cm of rounding]
   {\Large$\candf=f_{\va}$\\[2mm]\large$\va\in(\texttt{float})^\numfeatures$};

\node (cand-rat) [io, right=0.5cm of rounding]
   {\Large$\candf=f_{\va}$\\[2mm]\Large$\va\in\qq^\numfeatures$};

\node (checker) [medium-process, right=2cm of guesser, minimum width = 4cm] { };
\path [color=black, font=\sffamily] (checker.north)+(0,-2em) node (checker-txt)
      {\Large\textbf{Checker}};
\node (cas) [process-component, below=0.5cm of checker-txt]
      {\Large\textbf{CAS}};

\node (repair) [unimplemented-process, below=2cm of checker]
   {\large Repair\\[1mm]\large or\\[1mm]\large Refine};

\node (verif-result) [output-large, right=1cm of checker] {};
\path [color=black, font=\sffamily] (verif-result.north)+(0,-2.5em) node (verif-result-txt)
      {\Large \textbf{Verification}\\[2mm]\Large \textbf{Info}};
\node (checked) [verif-output, below=0.4cm of verif-result-txt, fill=green!80]
      {\Large Checked};
\node (counterex) [verif-output, below=0.3cm of checked, fill=red!80]
      {\Large Counterexample};
\node (idontknow) [verif-output, below=0.3cm of counterex]
      {\Large I don't know};

\node (cand-result) [output] at (verif-result|-cand-rat)
      {\Large \textbf{Executable}\\[2mm]\Large \textbf{Output Bound}};

\path (cand-result)--(verif-result) node (plus-results)
      [color=black, font=\sffamily, midway]{\Huge\textbf{+}};

\path[arrow] (operator)   edge (guesser.west |- operator);
\path[arrow] (recurrence) edge (operator)
             (training)   edge (featselect)
             (featselect) edge (finalguess)
             (cand-float) edge (rounding)
             (rounding)   edge (cand-rat);
\draw[arrow] (guesser.east) -- ++(0.5cm,0cm) node (pathp1) {}
             -- ++(0,-4.5cm) node (pathp2) {}
             -- (operator |- pathp2) node (pathp3) {}
             -- (pathp3 |- cand-float) -- (cand-float);
\path[arrow] (checker.east) edge (checked.west)
             (checker.east) edge (counterex)
             (checker.east) edge (idontknow.west);
\draw[arrow] (cand-rat.east) -- ++(1.65cm,0cm) node (pathq1) {}
             -- (repair |- pathq1) node (pathq2) {}
             -- (cand-result);
\draw[arrow] (pathq1 |- cand-rat) -- (pathq1 |- checker) -- (checker);
\path[arrow, dashed] (cand-result.north west)  edge (repair)
                     (verif-result.south west) edge (repair)
                     (repair)       edge (checker);

\draw[very thick, {Latex[length=2.5mm,width=2.5mm]}-]
   (featselect.south)+(-4.8209mm,-2mm) arc (130:410:7.5mm);

\end{tikzpicture}
}
  \caption{Architecture of our optimisation-based recurrence solver.
    \label{fig:implem-arch}
  }
\end{adjustwidth}
\end{figure}
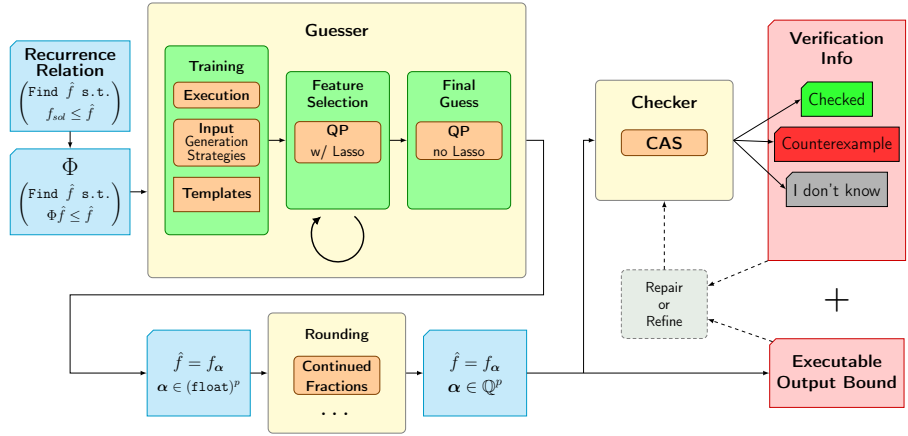

\begin{figure}
  \centering

\begin{tabular}{m{5cm} | m{9cm}}
\begin{lstlisting}[language=Python, basicstyle=\scriptsize\ttfamily]
def f(n,c):
  if n>0 and c>=100:
    return f(n-1, 0) + n + 300
  if n>0 and c<100:
    return f(n-1, c+1) + n
  return c
\end{lstlisting}
&
\begin{lstlisting}[language=Python, basicstyle=\scriptsize\ttfamily]
def f_ub(n, c):
  if ((n > 0) and (c >= 100)):
    return 1/2*n*n + 701/202*n + 30000/101*1
  elif ((n > 0) and (c < 100)):
    return 1/2*n*n + 701/202*n + 300/101*c
  else:
    return c

# ===========================================
# Result: {}.
# Candidate checked (no counter-examples).
# ===========================================

def f_lb(n, c):
  if ((n > 0) and (c >= 100)):
    return 1/2*n*n + 701/202*n + 100*1
  elif ((n > 0) and (c < 100)):
    return 1/2*n*n + 701/202*n + 300/101*c \
           + -19900/101*1
  else:
    return c

# ===========================================
# Result: {}.
# Candidate checked (no counter-examples).
# ===========================================
\end{lstlisting}
\end{tabular}
\caption{Input recurrence (left) corresponding to
  Equation~\ref{eq:running-ex-1}, Section~\ref{sec:running-example},
  and output (right), including closed-form bound functions, of our
  prototype using \texttt{lden} rounding strategy (based on continued
  fractions) and training sample $\inputset=[0,300]\times[0,300]$.
\label{fig:input-output-py}
  }
\end{figure}

\noindent
The full pipeline of our prototype is described in
Fig.~\ref{fig:implem-arch}. Inputs and outputs are given as Python
code, decorated with comments for the verification results, as
illustrated in Fig.~\ref{fig:input-output-py}.

Equation evaluation is sped up by memoisation decorators.
Operators are constructed by AST manipulations.
Inputs and outputs are executable and can be connected with other
analyses.

We use the \texttt{cvxopt}~\cite{cvxopt} library for
(\texttt{float}-based, iterative) QP solving, and the
\texttt{Mathematica} CAS solver~\cite{mathematica-v13-2} to decide
function inequality.
The verification step can be turned off if
\texttt{Mathematica} is not available on the user's machine, or can be
performed separately on candidates obtained by other means for further
experimentation.
We augment these off-the-shelf tools with simple
heuristics, e.g. we allow the CAS to discover extra simplification
opportunities for divide-and-conquer benchmarks by hinting towards
case reasoning on the value on inputs modulo a small product of
primes.

In the ``guess'' stage, several rounds of feature selection are
performed, which restrict the set of base functions to the most
relevant ones, using only a portion of the training sample to
accelerate the process. This diminishes the dimension of the problem,
which tempers numerical errors and speeds up the final guess, where we
allow more QP iterations to improve precision.

Additionally, to recover candidates represented with exact rational
coefficients from floating-point versions of them, we allow the user
to choose among several heuristics that we call \emph{rounding
strategies}, including a version of the
\texttt{fractions.limit\_denominator} function in \texttt{Python}'s
standard library, modified so as to obtain best lower and upper
rational bounds (of bounded denominator) of any floating-point number,
using continued fraction expansions~\cite{Kimberling1997,Lang1995}.
Note that, in principle, we could apply a more complete method by
using approaches such as \cite{exactLP07,Monniaux_CAV09}, which
recover rational solutions to an LP problem from an approximate
\texttt{float} solution, and applying them to
obtain a feasible solution of our QP problem close to the initially
suggested candidate, without including the quadratic objective
function.

\paragraph{\textbf{Benchmarks.}}

We evaluate our prototype and other tools on $21$ benchmarks organised
into five categories that we believe are representative of the
challenging features exhibited by recurrences arising from cost
analysis, capturing the essence of a much larger set of
benchmarks. For brevity, we provide a high-level overview of them in
this section and encourage the reader to interpret them through the
lens of the equation–program correspondence, as discussed in
Section~\ref{sec:running-example}, since each equation can be viewed
as the size/cost abstraction of multiple programs. Nevertheless, a
full description of these benchmarks, in equation format, is given in
Appendix~\ref{sec:benchmark-tables}, Table~\ref{table:benchmarks}.

For comparison with external tools, benchmarks have been manually
translated into a suitable input format using a strategy similar
to~\cite{mlrec-tplp2024-preprint} to avoid accuracy loss from
automated translation and to better understand actual tool
capabilities. The full translations can be found in our artifact.

Category \texttt{\catwithnoise} contains benchmarks with programs
whose cost function is difficult to express exactly, because of
underlying noise that can be interpreted as rare cost events coming
from interaction with hardware or data structures. This includes
variations of our running example and a ``fill and free'' loop
monitoring its memory usage.
Category \texttt{\catnonmonot} contains benchmarks with non-monotone
solutions, which can typically arise from the particular control-flow
of imperative programs, but is less common in declarative code.
Several are inspired from~\cite{montoya-phdthesis}, and we also encode
a program with a memory leak caused by an incorrect boundary condition.
Category \texttt{\catdivandconq} presents equations that exactly encode
some divide-and-conquer programs. For them too, simple and exact
closed-form solutions are not typically available, although
asymptotic behaviours can be obtained at the expense of safe bounds
and initial values. They present recursive calls with divisions and
floor/ceil that are particularly hard for classical CAS if not helped
by abstractions.
Category \texttt{\catsuperpoly} presents simple equations with fast-growing solutions, typically better supported by CAS than by cost
analysers.
The remaining examples are in category \texttt{\catmisc}, including a
simple one for sanity check, and harder benchmarks testing diverse
features, such as oscillating subdomain paths, arithmetic computations,
nested calls, $\max$ expressions in equation or ideal solution, and
amortised cost in non-deterministic programs.

\paragraph{\textbf{Experimental Setup and Hyperparameters.}}

Experiments are performed with the following hyperparameters.
Tuning them can affect performance and precision.
We assume here that $\fsol$ has domain $\nn^\dimdom$.

For input sample generation, $\inputset$ is by default the union of
two sets, the hypercube $[0,\Bb]^\dimdom$ with $\Bb$ chosen so that it
contains around $\Nb$ elements, and a random sample of $\Nrs$ elements
in $[0,\Brs]^\dimdom$, with $\Nb=5\cdot10^4$, $\Brs=2\cdot10^3$ and
$\Nrs=2\cdot10^4$. If this leads to overflows in evaluation (if the
solution is fast-growing), we retry with small inputs $\Bb=\Brs=15$.
We use the model space of piecewise affine combination of base
functions presented in Section~\ref{subsec:optim-approach}.
For simplicity (and unlike~\cite{mlrec-tplp2024-preprint}), we select
simple sets of base functions, which can be customised.
Experiments are performed twice, first with with all multivariate
polynomials of small bounded degree (defaults to $2$),
\lrnote{Slightly embarrassing datapoint. Is ok.}
then adding all products of them with $\log_2$ in one variable. If
overflows are encountered in input generation, we use instead
polynomials, small products of exponentials (at most $2$ variables,
base at most $3$), and factorials.
By default, memoisation is enabled, lasso parameter is
$\lambda=10$, and half of the training set is sampled for feature
selection, done in $2$ rounds. For each call to the QP solver, $100$
KKT steps are performed, and we tolerate less error in feasibility
(with \texttt{cvxopt} conventions, $10^{-15}$ tolerance, close to
\texttt{float} maximal precision)
than in optimality ($10^{-10}$
absolute and $10^{-9}$ relative tolerance). \texttt{Mathematica}
verification tasks timeout at \texttt{1min}. The default max
denominator in the rounding strategy by continued fractions is $10^5$.

We perform extra experiments to account for non-affine benchmarks,
disabling inductivity constraints and enforcing positivity of the
coefficients (which enforces monotonicity of the candidate for our
templates).
Finally, we use benchmark \texttt{\noiseddthreeonename} to show the
value of additional training data, improving QP convergence and
precision, but also to suggest that implementation of repair
strategies (e.g. via ranking functions) could avoid
this time-consuming step. With default values, we obtain reasonable
but imperfect candidates. We rerun the experiment with more training
data ($\Nrs=10^6$ and $\Brs=300$), obtaining the ideal, tightest
polynomial bounds for this benchmark, although it takes $20\times$
more time.

\lrnote{Table has to be manually placed to avoid large white
  spaces, at the very final stage}
\clearpage
\newpage
\begin{table}[h]
\centering
\begin{adjustwidth}{-6em}{-3em}
{\scriptsize
 \setlength{\extrarowheight}{2.5pt}%
\begin{tabular}{|c|c|c|c|c|c|c|}
\hline
\multicolumn{2}{|c|}{\textbf{Bench}}  
& \texttt{\prototypename}
& \texttt{\linregonlyname}
& \texttt{\ramlname}
& \texttt{\cofloconame}
& \texttt{\mathematicaname}
\\ \hline \hline
\multirow{9}{*}{\rotatebox[origin=c]{90}{\texttt{\catwithnoise}}}
& \texttt{\runningex}
& \makecell[c]{
  \phantom{.}\\[-0.8em]
  $[\ssim, \ssim]$\\
  Tight in all coeffs.,\\
  but $\rO(1)$.
  }
& \makecell[c]{
  $\snotabound$\\
  In $\frac{1}{2}n^2+\rO(n)$\\
  for $n\to\infty$
  }
& \makecell[c]{
  \phantom{.}\\[-0.8em]
  $[\ssim, \sbadub]$\\
  lb: $0.5n^2+0.5n$\\
  ub: $0.5n^2+300.5n+c$
  }
& \makecell[c]{
  \phantom{.}\\[-0.8em]
  $[\sbadlb, \sTheta]$\\
  lb: $n$\\
  ub: $2n^2+\rO(n)$
}
& $\snone$
\\ \cline{2-7}
& \texttt{\noiseddtwotwo}
& \makecell[c]{
  \phantom{.}\\[-0.8em]
  $[\ssim, \ssim]$\\
  Tight in all coeffs.,\\
  but $\rO(1)$.
  }
& \makecell[c]{
  $\snotabound$\\
  In $\sTheta(n^2)$\\
  for $n\to\infty$
  }
& \makecell[c]{
  $[\sTheta,\snone]$\\
  lb: $\sTheta(n^2)$
}
& \makecell[c]{
  \phantom{.}\\[-0.8em]
  $[\sbadlb, \sTheta]$\\
  lb: $n$\\
  ub: $302n^2+\rO(n)$
}
& $\snone$
\\ \cline{2-7}
& \texttt{\fillandfree}
&
\makecell[c]{
  \phantom{.}\\[-0.8em]
  $[\fsol-\sTheta(1), \fsol+\sTheta(1)]$\\
  $[0, 0]$     ($n=0$)\\
  $[-c, 99-c]$ ($n>0$)\\[1pt]
  }
& \makecell[c]{
  $\snotabound$\\
  ($\approx\texttt{cst}-\alpha\cdot c$)
  }
& \makecell[c]{
  $[\snone, \sbadub]$\\
  ub: $n$
  }
& \makecell[c]{
  $[\snotabound, \sbadub]$\\
  lb: $-n$\\
  ub: $n$
  }
& $\snone$
\\ \cline{2-7}
& \texttt{\noiseddthreeone}
& \makecell[c]{
  \phantom{.}\\[-0.8em]
  $[\ssim, \ssim]$\\
  Tight in all coeffs.,\\
  but $\rO(1)$.
  }
& \makecell[c]{
  $\snotabound$\\
  In $\frac{1}{2}n^2+\rO(n)$\\
  for $n\to\infty$
  }
& $[\snone, \snone]$
& $[\snone, \snone]$
& $\snone$
\\ \hline \hline
\multirow{4.7}{*}{\rotatebox[origin=c]{90}{\texttt{\catnonmonot}}}
& \texttt{\incrone}
& $[\sexactnum,\sexactnum]$
& \makecell[c]{
  $\snotabound$\\
  In $\sTheta(1)$
}
& \makecell[c]{
  \phantom{.}\\[-0.8em]
  $[\sTheta, \snone]$\\
  lb: $1$
  }
& $[\sexact, \sexact]$
& $\snone$
\\ \cline{2-7}
& \texttt{\noisystrtone}
& $[\sexactnum,\sexactnum]$
& $\sexact$
& \makecell[c]{
  \phantom{.}\\[-0.8em]
  $[\snone, \ssim]$\\
  ub: $n$
  }
& $[\sexact, \sexact]$
& \makecell[c]{
  $\ssim$/$\snotabound$\\
  ``$\exists \texttt{c},\, \fsol(n)=n+\texttt{c}$''
}
\\ \cline{2-7}
& \texttt{\memorythree}
& $[\sexactnum,\sexactnum]$
& $\sexact$
& $[\snone, \snone]$
& $[\sexact, \sexact]$
& $\snone$
\\ \cline{2-7}
& \texttt{\multiphaseone}
& $[\sexact,\sexact]$
& $\sexact$
& $[\snone, \snone]$
& $[\sexact, \sexact]$
& $\snone$
\\ \hline \hline
\multirow{5}{*}{\rotatebox[origin=c]{90}{\texttt{\catdivandconq}}}
& \texttt{\binsearch}
& \makecell[c]{
    $[\sTheta, \ssim]$\\
    lb: $\approx \frac{2}{3} \log_2(n)$
  }
& \makecell[c]{
    $\snotabound$\\
    $\approx \log_2(n)$
  }
& \makecell[c]{
    $[\snone, \sbadub]$\\
    ub: $\sTheta(n)$
  }
& $[\snone, \snone]$
& \makecell[c]{
    $\snone$ without help\\
    $[\ssim, \ssim]$ with help
  }
\\ \cline{2-7}
& \texttt{\qsortone}
& $[\ssimnum, \ssim]$
& \makecell[c]{
    $\snotabound$\\
    $\approx 1.2 n\log_2(n)$
  }
& \makecell[c]{
    $[\snone, \sbadub]$\\
    ub: $\sTheta(n^2)$
  }
& $[\snone, \snone]$
& \makecell[c]{
    $\snone$ without help\\
    $[\ssim, \ssim]$ with strong help
  }
\\ \cline{2-7}
& \texttt{\divandconqmultioneone}
& \makecell[c]{
    $[\sbadlb, \sbadub]$\\
    lb: $\approx 2n\log_2(n)$\\
    ub: $\approx 1.2 n^2$
  }
& \makecell[c]{
    $\snotabound$\\
    $\approx 20 n\log_2(n)$
  }
& \makecell[c]{
    $[\sbadlb, \sbadub]$\\
    lb: $\sTheta(n)$\\
    ub: $\sTheta(n^2)$
  }
& $[\snone, \snone]$
& $\snone$
\\ \hline \hline
\multirow{7.5}{*}{\rotatebox[origin=c]{90}{\texttt{\catsuperpoly}}}
& \texttt{\expone}
& $[\sexact, \sexact]$
& $\sexact$
& \makecell[c]{
  $[\sbadlb, \snone]$\\
  lb: $\sTheta(n^d)$\\
  (Poly lb in\\selected degree)
  }
& $[\snone, \snone]$
& $\sexact$
\\ \cline{2-7}
& \texttt{\exptwo}
& $[\sexact, \sexact]$
& $\sexact$
& \makecell[c]{
  $[\sbadlb, \snone]$\\
  lb: $\sTheta(n^d)$\\
  (Poly lb in\\selected degree)
  }
& $[\snone, \snone]$
& $\sexact$
\\ \cline{2-7}
& \texttt{\fact}
& $[\sexact, \sexact]$
& $\sexact$
& \makecell[c]{
  $[\sbadlb, \snone]$\\
  lb: $\sTheta(n^d)$\\
  (Poly lb in\\selected degree)
  }
& \makecell[c]{
  $\snone$\\
  $\snotabound$ for $n\geq 2$
  }
& $\sexact$
\\ \hline \hline
& \texttt{\basicone}
& $[\sexact, \sexact]$
& $\sexact$
& $[\sexact, \sexact]$
& $[\sexact, \sexact]$
& $\sexact$
\\ \cline{2-7}
& \texttt{\sumosc}
& $[\sexact, \sexact]$
& $\sexact$
& \makecell[c]{
  $[\snotabound, \sbadub]$\\
  lb: $\frac{1}{2}y^2+\frac{3}{2}y+1$\\
  \phantom{lb: } ($\snotabound$ in $x=0$)\\
  ub: $\frac{1}{2}y^2+\frac{3}{2}y+x+1$\\
  \phantom{ub: } ($\sbadub$ in $y=0$)
  }
& \makecell[c]{
  $[\snotabound, \sbadub]$\\
  lb: $\sTheta(1)$, $\snotabound$ in $(0,1)$\\
  ub: $2y^2+xy+\rO(y)$
  }
& $\snone$
\\ \cline{2-7}
& \texttt{\divben}
& \makecell[c]{
  $[\sbadlb, \sbadub]$\\
  lb: $\sTheta(1)$\\
  ub: $\max(x-y+1,0)$
  }
& \makecell[c]{
  $\snotabound$\\
  For $x\geq y\geq 1$,\\
  $\approx 0.015x-0.025y+\rO(1)$
  }

& $[\snone, \snone]$
& \makecell[c]{
  $[\snone, \sbadub]$\\
  ub: $\max(x-y+1,0)$
  }
& \makecell[c]{
  $\snone$\\
  $\snone$ without help\\
  ``x/y'' (no int) with help
  }
\\ \cline{2-7}
\multirow{3}{*}{\rotatebox[origin=c]{90}{\texttt{\catmisc}}}
& \texttt{\nested}
& $[\sexact, \sexact]$
& $\sexact$
& $[\sexact, \sexact]$
& $[\snone, \snone]$
& $\snone$
\\ \cline{2-7}
& \texttt{\openzip}
& \makecell[c]{
  $[\sTheta, \snotabound]$\\
  Exact for $x=0\vee y=0$\\
  For $x>0\wedge y>0$,\\
  lb: $\approx \frac{1}{2}x+\frac{1}{2}y$\\
  ub: $\approx 0.99x+0.99y+\rO(1)$
  }

& \makecell[c]{
  $\snotabound$\\
  Exact for $x=0\vee y=0$\\
  For $x>0\wedge y>0$,\\
  $\approx 0.6x+0.6y+\rO(1)$
  }
& \makecell[c]{
  $[\sbadlb, \sTheta]$\\
  lb: $y$\\
  ub: $x+y$\\
  }
& $[\sexact, \sexact]$
& $\snone$
\\ \cline{2-7}
& \texttt{\mergeben}
& $[\sexactnum, \sexactnum]$
& $\sexact$
& \makecell[c]{
  $[\snone, \ssim]$\\
  ub: $x+y$\\
  }
& \makecell[c]{
  $[\sbadlb, \sexact]$\\
  lb: $\min(x, y)$\\
  (lb(bc) instead\\of lb(wc))
  }
& $\snone$
\\ \cline{2-7}
& \texttt{\loopustarjan} %
& $[\sexact, \sexact]$
& $\sexact$
& $[\snone, \sexact]$
& \makecell[c]{
  $[\sTheta, \sexact]$\\
  lb: $i$\\
  (lb(bc) instead\\of lb(wc))
  }
& $\snone$
\\ \hline
\end{tabular}}
\end{adjustwidth}
\vspace*{1em}
\caption{Experimental evaluation and comparison.\label{table:results-full}}
\end{table}

\clearpage
\newpage

\paragraph{\textbf{Legend.}}

The results of our experimental evaluation and comparison are
presented in Table~\ref{table:results-full}.
For the sake of simplicity, and for space reasons, output accuracy is
reported with only a few symbols defined below.
This is non-trivial, and information is necessarily lost in the
process, as outputs are often incomparable.
Asymptotic order of scales ($\sTheta$) are too coarse-grained for the
applications we are interested in (e.g. granularity control), and are
difficulty to define well for multivariate outputs (e.g.
\texttt{termCOMP}~\cite{termcomp19} compresses reports on $\vx$ into
reports for a single input $\sum_i|x_i|$).
At the opposite end of the spectrum, numerical accuracy scores (e.g.
$\ell_2$ error on a sample) are not robust, and do not allow
evaluation of qualitative properties (e.g. tightness of a bound).
For these reasons, we use the following list of symbols, first defined
for single variable functions $f:\nn\to\rr$, then extended for
multivariate functions, illustrated with an example in the context of
evaluating an upper bound on $\fsol(n) = n^2 - 2n - 3$.
\begin{center}
\scalebox{0.95}{\footnotesize\begin{tabular}{|c|c|c|}
  \hline
  \textbf{Symbol} &
  \textbf{Meaning} & \textbf{Example}
  \\\hline
    $\sexact$
  & \makecell[c]{\scriptsize exact solution}
  & $n^2 - 2n - 3$
  \\\hline
    $\sexactnum$
  & \makecell[c]{\scriptsize exact solution, with small\\[-3pt]
                 \scriptsize numerical errors on coefficients}
  & $1.001 n^2 - 1.998n - 2.997$
  \\\hline
    $\ssim$
  & \makecell[c]{\scriptsize classical asymptotic equivalence\\[-3pt]
                 \scriptsize $f \sim g \iff \lim_{n\to+\infty} f(n)/g(n) = 1$}
  & $n^2$
  \\\hline
    $\ssimnum$
  & \makecell[c]{\scriptsize as above, with small numerical errors,\\[-3pt]
                 \scriptsize ratio tends to a number very close to 1}
  & $1.0001 n^2$
  \\\hline
    $\sTheta$
  & \makecell[c]{\scriptsize usual big theta, bounded ratio,\\[-3pt]
                 \scriptsize not very close to 1}
  & $4 n^2$
  \\\hline
  $\sbadlb$/$\sbadub$
  & \makecell[c]{\scriptsize non-trivial lower/upper bound,\\[-3pt]
                 \scriptsize but incorrect asymptotic order}
  & $n^3$
  \\\hline
  $\snone$
  & \makecell[c]{\scriptsize trivial bound or no output}
  & $+\infty$
  \\\hline
  $\snotabound$
  & \makecell[c]{\scriptsize not a bound (unsound or does\\[-3pt]
                 \scriptsize not claim to be a bound)}
  & $\frac{1}{2}n^2$
  \\\hline
\end{tabular}}
\end{center}
Numerical errors that cause unsoundness are reported with
$\snotabound$.
There are several options to compress multivariate functions.
We choose to use the best symbol that holds ``in all directions'',
i.e. $f\,\mathcal{R}\,g$ whenever
$(f\circ u)\,\mathcal{R}\,(g\circ u)$ holds for all sequences
$u:\nn\to\dd$ such that $||u_n||$ tends to infinity (i.e. at
least one component tends to infinity).
For example, $\max(x,y)\in\sTheta(x+y)$, but $(x,y)\mapsto x+y-1$ is
only $\sbadub$ for the \texttt{merge} benchmark (consider
$u_n = (n,0)$, where $\fsol\circ u = \vzero$).

For each tool reporting both lower and upper bounds, we report its
result as an interval of symbols.
Some additional information is often helpful to interpret the results,
and is then included below this interval in each cell.

For further details, we refer the reader to
Table~\ref{table:runnex_results} in
Appendix~\ref{sec:benchmark-tables} to interpret full outputs in the
case of our running example, and to our artifact otherwise.

\paragraph{\textbf{Evaluation and Discussion.}}

Overall, the experimental results are quite promising, showing that
our approach is accurate, discovers bounds not obtained by other
tools, and is reasonably efficient.
\lrnote{$\sim$Done, but needs cleanup, polishing, possibly compression}
More specifically, these results are compared in
Table~\ref{table:results-full}
with further analysers from different paradigms and one CAS. The
features of related tools can be found
in~\cite{mlrec-tplp2024-preprint}, which provides additional
motivation for the choice of tools presented here.
Compared to (\texttt{\linregonlyname}),
i.e.~\cite{mlrec-tplp2024-preprint} with hyperparameters similar to
our prototype, we infer bounds in cases where it does not obtain any
safe approximation (e.g. categories \texttt{\catwithnoise} and
\texttt{\catdivandconq}).
\texttt{\mathematicaname}'s results support our claim that CAS
are not appropriate for the problems we tackle.
Compared to \texttt{\ramlname}, we obtain better bounds in multiple
cases. While \texttt{\ramlname} consistently obtains reasonable bounds
for the features it is designed for, and handles some difficult
benchmarks (e.g. \texttt{\nestedname}), it is hindered by its
limitation to global polynomial templates. Furthermore, it does not
support integer comparisons, and list encodings are not enough to
make it succeed consistently on categories such as
\texttt{\catwithnoise} and \texttt{\catdivandconq}.
We also obtain overall better results than \texttt{\cofloconame} on our
benchmarks, in part because of the coarse overapproximations it performs
(e.g. in \texttt{\catwithnoise} or \texttt{\catmisc}) and its limited
support for non-linear recursion (\texttt{\catdivandconq},
\texttt{\catsuperpoly}).

Our prototype is also reasonably efficient. Runtimes depend on the
amount of training data, but we can report the time taken for the
candidates of Table~\ref{table:results-full}, with experiments run on
a small Linux laptop (\texttt{1.1GHz Intel Celeron N4500} CPU,
\texttt{4~GB}, \texttt{2933~MHz DDR4} memory). All of them (both
bounds combined, timeouts excluded) can be obtained and checked in
\mhnote{For final version we probably want to report times with fewer
  significant digits, since execution times are rarely repeatable with
  very high precision.}
\texttt{9.047s} to \texttt{47.475s}, around half of which is taken by
verification, at the exception of \texttt{\noiseddthreeonename} for
which ideal tight bounds are only obtained with extra training data
requiring \texttt{22min}, although reasonable candidates are obtained
in less than \texttt{1min}.

However, it must be noted that, unlike \texttt{\ramlname} and
\texttt{\cofloconame}, the optimisation-based approach implemented
here is susceptible to numerical errors and sensible to hyperparameter
choice, so that it can require extra analyses or domain knowledge for
fine-tuning.
Because of this, \texttt{\cofloconame} outperforms our prototype in
category \texttt{\catnonmonot}, where it obtains exact solutions in
cases where our prototype only infer (accurate) approximations.
For example, in the memory leak benchmark, we only find bounds
$[\frac{4589}{4599},\frac{4609}{4599}]$ for $s=0$ instead of $[1,1]$,
and in the \texttt{\openzipname} benchmark, a counter-example caused
by numerical errors is found for our candidate upper bound.
We may note that approaches mentioned in
Section~\ref{subsec:alternative-approaches}~and~\ref{sec:repair}, not
implemented in this prototype, could alleviate this limitation. A full
implementation and evaluation of these techniques is left for future
work.

Overall, we can conclude that our approach performs quite well.  Most
importantly, it demonstrates the existence of bounds that are
inductive in the sense of this paper, which cannot be
discovered by current
methods,
as well as the feasibility of finding them, showing the potential of
the order-theoretical viewpoint introduced in this paper.

\section{Related Work}
\label{sec:related-work}

\textbf{Exact Recurrence Solvers.}
Centuries of work have created a large body of
knowledge, %
with classical results including closed forms of \emph{C-recursive
sequences}~\cite{Petkovsek2013sketch},
or more recently for hypergeometric solutions of \emph{P-recursive
sequences}~\cite{Petkovsek92}.
Several algorithms have been implemented
in CAS, based on mathematical frameworks such as
\emph{Difference Algebra}~\cite{Karr81,Levin08}, \emph{Finite
Calculus}~\cite{ConcreteMathematics} or \emph{Generating
Functions}~\cite{Flajolet09} mixed with template-based methods.
However, as explained in the introduction, these techniques have a
different focus than ours.

\textbf{Cost Analysis via (Generalised) Recurrence Equations.}
Since the seminal work of Wegbreit~\cite{Wegbreit75}, implemented in
\textsc{Metric}, multiple authors have tackled the problem of cost
analysis of programs (either logic, functional, or imperative) by
automatically setting up and solving recurrence equations.
This includes %
early work such
as~\cite{Le88-short,Rosendahl89} and the
\ciao line of work (c.f.\ references in introduction), %
as well as tools such as \texttt{PUBS}~\cite{AlbertAGP11a-short} and
\texttt{Cofloco}~\cite{montoya-phdthesis},
which emphasise the
shortcomings of using too simple recurrence equations.
They introduce
the vocabulary of \emph{cost relations} to focus on properties such as
non-determinism, inequations, non-monotonic solutions and complex
control flow.
Additionally, recent work in \texttt{KoAT} also uses recurrence solving~\cite{LommenGiesl23}.%

\textbf{Recurrence Solving for Invariant Synthesis and Verification.}
In other areas of program analysis, important lines of work use
recurrences as key ingredients in generation of invariants, with
applications such as loop summarisation.
This is well illustrated by the work initiated
in~\cite{kovacs-phdthesis}, with a recommended overview available
in~\cite{kovacs23}.
This is also key to compositional recurrence
analysis~\cite{kincaid2018}, which has a stronger focus on
abstraction-based techniques (e.g., the \emph{wedge} abstract domain).
We can also note that \emph{size} analysis can sometimes be seen as a
form of numerical invariant synthesis, as illustrated
by~\cite{LommenGiesl23}.

\textbf{Other Approaches to Static Cost Analysis.}
Automatic static cost analysis of programs is an active field of
research, and approaches with different abstractions than recurrence
equations
have been proposed. We will simply mention here the
potential method implemented in
\texttt{RaML}~\cite{DBLP:journals/toplas/0002AH12-short}, which utilises
polynomial templates in the binomial basis, and refer the reader
to~\cite{mlrec-tplp2024-preprint} for more references.

\textbf{Dynamic Inference of Invariants/Recurrences.}
The instantiation of our approach implemented in
our prototype generalises~\cite{mlrec-tplp2024-preprint}, by using
inductivity constraints to obtain bounds, and is directly related to
dynamic invariant
analysis~\cite{ernst-phdthesis,daikon-tse,nguyen2012-shortest,nguyen2014-shortest,nguyen2022:invariants-tse}.
A key difference between these works and ours, beyond different
technical solutions and underlying theory, is the setting: we do not
sample program traces but instead abstract semantics of recurrence equations 
that are intermediate abstractions of programs, obtained via size
abstractions.

\textbf{Fixpoints and/or Order for Equations and Sequences.}
To the best of %
our knowledge, our work is the first to use
both a fixpoint and an order-theoretical viewpoint to infer bounds on
the solutions to recurrence equations on their whole domain, in a way
that enables support for a large class of templates. Nevertheless, we
would like to mention other work, that we believe is related.

\lrnote{PLG/LR: try to phrase this optimally for the final version, it
  is important for the novelty claim}
\lrnote{Tentative reformulation. I do not know the appropriate balance
  between strong claims and acknowledgements of origins of ideas.}
First, while the equations-as-operators viewpoint, for the most
general sense of equation, is established in abstract interpretation
for handling \emph{semantic equations}~\cite{Bourdoncle93}, it is much
less common in more explicit numerical settings, where the goal is to
solve concrete recurrence or differential equations.
In
our context
(cost analysis in particular, but more generally for numerical
equations considered globally), we believe that our framework is
unique. %

Concepts introduced in this paper are related to the \emph{flowpipes}
developed for differential
equations~\cite{GoubaultHSCC17-short,GoubaultCAV18}.
This work also exploits a fixpoint semantics of equations viewed as
operators, but this fixpoint is fundamentally \emph{topological},
rather than \emph{order-theoretical}, using the 
Picard-Lindelöf/Cauchy-Lipschitz theorem instead of Tarski's -- it is
difficult to naturally view the Picard operator as monotone in a
global sense.
Flowpipe-based approaches do recover under and overapproximations, building
on the ``classical interval Picard-Lindelöf approach'' developed by
Nedialkov~\cite{Nedialkov99}, but the inclusion proofs for intervals
are non-obvious, and can require restriction of the time horizon, so
that only \emph{local} bounds (with respect to time) are obtained, as
opposed to the \emph{global} bounds we discuss in this paper.

In a different context, probabilistic program analysis, techniques
introduced by Karp~\cite{Karp94} have been later understood in
light of pre/postfixpoints, which enabled the approach
of~\cite{qava-pldi21}. However,
such an approach is less general than
ours, and only obtains bounds (here on probabilities) at explicit
values, rather than fully parametric, functional bounds. This order
approach appears to not have been further pursued, as recent work on
probabilistic recurrence
equations~\cite{sun-cav23} reverts to more
classical tail-bound concentration techniques (e.g. Markov
inequality).

Finally, other work has considered the idea of abstract domains of
discrete functions, in particular~\cite{Feret05-short} for (one-dimensional)
algebraic-geometric sequences, applied to filter analysis. More
generally, several abstract interpretation works explore numerical
abstract domains for non-linear
invariants~\cite{Feret04,RodriguezKapurSas04}. However, they cannot be
applied directly in our context, as transfer functions for key
operators on functions are not included, and their design can require
overcoming serious theoretical difficulties.

\section{Conclusions and Future Work}
\label{sec:conclusion}

We have introduced a
novel approach to recurrence solving in the context of cost analysis:
equations are viewed as monotone operators, which allows leveraging
classical order-theoretical results and techniques, and obtaining safe
and tight lower/upper bounds on the solutions to
complex recurrences as pre/postfixpoints.
 
The novel order-theoretical framework we have developed under this
viewpoint has not only provided new insights on the properties of
recurrences, by unveiling underlying algebraic and geometric
structures, but has also laid the groundwork for more accurate and
robust methods in cost analysis: the properties and principles we have
proved allow us to develop techniques that can be combined for
designing new solvers that infer safe overapproximations.

Indeed, as a proof-of-concept, we have implemented a prototype
instantiation of the proposed approach
that uses constrained optimisation
techniques. Our experimental evaluation showed quite promising
results: it infers tight non-linear lower/upper bounds for complex
recurrences not supported by the state-of-the-art recurrence solvers
and cost analysers.
It is particularly effective for programs with complex control flow
and for equations obtained while incorporating fine-grained cost
models, which result in multivariate solutions with non-monotonic
properties, complex recursive calls and disjunctive multiphase
behaviours. Indeed, these equations present
complex features that make them hard to analyse directly, but also do
not admit simple closed-form exact solutions, so that template-based
methods have to be extended to support inequalities while
accommodating complex piecewise and non-linear templates.

\mhnote{MH: Not strictly necessary to talk about future research if
  shortening needed (but I think length now is really OK for SAS).}
Finally, our work opens new avenues for future research.
We have presented multiple directions for designing new algorithms to
explore spaces of bounds and integrate information available locally
to generate candidate bounds, check their correctness, evaluate their
accuracy, and improve or repair them.

Additionally, while our proof-of-concept prototype is focused on dynamic
invariant generation techniques, using sample inputs conveniently
available in discrete settings, we have proposed abstraction-based
strategies that do not require such data.
Therefore, an exciting direction for future work is to apply our
viewpoint on functional equations beyond domains of integers and
codomains of reals. This would open exploration of new application
areas, e.g. cyberphysical systems, biochemical networks, and
probabilistic settings.
To fully enter these new territories, we would need to consider
support for non-monotonic equations, e.g. by introducing abstract
codomains such as intervals, and find ways to handle continuous
domains and differential operators, e.g. via non-standard analysis.
More concretely, another promising and interesting line for future work
is to integrate and implement our algorithms inside existing cost
analysis pipelines, while developing new size abstractions, able to
interact with fine-grained cost models to generate complex equations,
closely modelling cost properties and making the most of our new
recurrence solving methods.

\begin{credits}
  \subsubsection{\ackname}
This work has been partially supported by MICINN projects
PID2019-108528RB-C21 \emph{ProCode}, TED2021-132464B-I00
\emph{PRODIGY}, and FJC2021-047102-I, and the Tezos foundation.

\subsubsection{\discintname}
The authors have no competing interests to declare that are relevant
to the content of this article.
\end{credits}

\lrnote{Most references can be shortened to use less lines.}
\addcontentsline{toc}{section}{References}
\bibliographystyle{splncs04}
\bibliography{\bibpath/clip,\bibpath/general}

\newpage

\appendix

\section{Tables of Experimental Data}
\lrnote{WARNING, this changed format -- may be a false alarm, but
  there is build bug on my side}
\label{sec:benchmark-tables}

{\footnotesize
 \renewcommand{\arraystretch}{2}%
\begin{longtable}[c]{| c | c | c | p{5.5cm} | p{4.8cm} |}
\caption{Benchmarks\label{table:benchmarks}}\\
\hline
\multicolumn{3}{|c|}{\textbf{Bench}} %
& \makecell[c]{\textbf{Equation}} & \makecell[c]{\textbf{Solution}} \\ \cline{1-5} \cline{1-5}
\endfirsthead
\caption[]{Benchmarks (continued)}\\
\cline{1-5}
\multicolumn{3}{|c|}{\textbf{Bench}} %
& \makecell[c]{\textbf{Equation}} & \makecell[c]{\textbf{Solution}} \\ \cline{1-5} \cline{1-5}
\endhead
\multirow{7}{*}{\rotatebox[origin=c]{90}{\texttt{\catwithnoise}}}
 & \texttt{\runningexname}
 & \texttt{\runningexkey}
 & \gape{\scalebox{0.7}{$f(n,c)=\begin{cases*}
      f(n-1, 0)   + n + 300 & if $n>0 \wedge c\geq 100$\\[-2pt]
      f(n-1, c+1) + n       & if $n>0 \wedge c <   100$\\[-2pt]
      c                     & if $n=0$\\
    \end{cases*}$}}
 & \gape{\scalebox{0.7}{$\begin{cases*}
      \frac{n(n+1)}{2} + n + 299 + 199\lfloor\frac{n-1}{101}\rfloor\\
           \qquad \text{ if } n>0 \wedge c\geq 100\\[-2pt]
      \frac{n(n+1)}{2} + n + c + 199\lfloor\frac{n+c}{101}\rfloor\\
           \qquad \text{ if } n>0 \wedge c <   100\\[-2pt]
      c    \qquad\!\!\! \text{ if } n=0
    \end{cases*}$}}
 \\ \cline{2-5}
 & \texttt{\noiseddtwotwoname}
 & \texttt{\noiseddtwotwokey}
 & \gape{\scalebox{0.7}{$f(n,c)=\begin{cases*}
      f(n-1, 0)   + n + 300n & if $n>0 \wedge c\geq 100$\\[-2pt]
      f(n-1, c+1) + n        & if $n>0 \wedge c <   100$\\[-2pt]
      c                      & if $n=0$\\
    \end{cases*}$}}
 & \gape{\makecell[l]{
     \scalebox{0.7}{Hard to express exactly, but}\\
     \scalebox{0.7}{for $c\geq 100$, and $n\to\infty$,}\\
     \scalebox{0.7}{$=\frac{n(n+1)}{2}+\frac{300}{101}\frac{n(n+101)}{2} + \rO(1)$}
     }}
 \\ \cline{2-5}
 & \texttt{\fillandfreename} 
 & \texttt{\fillandfreekey} 
 & \gape{\scalebox{0.7}{$f(n,c)=\begin{cases*}
      f(n-1, 0)   - c & if $n>0 \wedge c\geq 99$\\[-2pt]
      f(n-1, c+1) + 1 & if $n>0 \wedge c <   99$\\[-2pt]
      0               & if $n=0$\\
    \end{cases*}$}}
 & \gape{\scalebox{0.7}{$\begin{cases*}
      (n-1)\% 100 - c & if $n>0 \wedge c\geq 99$\\[-2pt]
      (n+c)\% 100 - c & if $n>0 \wedge c <   99$\\[-2pt]
      0               & if $n=0$\\
    \end{cases*}$}}
 \\ \cline{2-5}
 & \texttt{\noiseddthreeonename}
 & \texttt{\noiseddthreeonekey} 
 & \gape{\makecell[l]{$f(n,c_1,c_2)=$\\
   \scalebox{0.6}{$\begin{cases*}
      f(n-1,    0,    0) + n + 300 + 20  & if $n>0 \wedge c_1\geq 100 \wedge c_2\geq 77$\\[-2pt]
      f(n-1,    0, c_2+1) + n + 300      & if $n>0 \wedge c_1\geq 100 \wedge c_2<77$\\[-2pt]
      f(n-1, c_1+1,    0) + n + 20       & if $n>0 \wedge c_1<100 \wedge c_2\geq 77$\\[-2pt]
      f(n-1, c_1+1, c_2+1) + n           & if $n>0 \wedge c_1<100 \wedge c_2<77$\\[-2pt]
      0                                  & if $n=0$\\
    \end{cases*}$}}}
 & \gape{\makecell[l]{
     \scalebox{0.7}{Long to express exactly, but}\\
     \scalebox{0.7}{for $c_1\geq 100$, $c_2\geq 77$ and $n\to\infty$,}\\
     \scalebox{0.7}{$=\frac{n(n+1)}{2}+\big(\frac{300}{101}+\frac{20}{78}\big)n + \rO(1)$}
     }}
 \\ \cline{1-5}
\multirow{6}{*}{\rotatebox[origin=c]{90}{\texttt{\catnonmonot}}}
 & \texttt{\incronename}
 & \texttt{\incronekey}  
 & \gape{\scalebox{0.7}{$f(x)=\begin{cases*}
      1+f(x+1) & if $x < 10$    \\[-2pt]
      1        & if $x \geq 10$ \\
    \end{cases*}$}}
 & $\max(1, 11-x)$
 \\ \cline{2-5}
 & \texttt{\noisystrtonename}
 & \texttt{\noisystrtonekey}  
 & \gape{\scalebox{0.7}{$f(x)=\begin{cases*}
      0        & if $x = 0    \lor  x = 20$    \\[-2pt]
      1+f(x-1) & if $x \neq 0 \land x \neq 20$ \\
    \end{cases*}$}}
 & \gape{\scalebox{0.7}{$\begin{cases*}
      x    & if $x < 20$    \\[-2pt]
      x-20 & if $x \geq 20$ \\
    \end{cases*}$}}
 \\ \cline{2-5}
 & \texttt{\memorythreename}
 & \texttt{\memorythreekey}  
 & \gape{\scalebox{0.7}{$f(i,n,s)=\begin{cases*}
      n+f(0,n,1)   & if $s = 0$                   \\[-2pt]
      1+f(i+1,n,s) & if $s > 0 \wedge i\leq n$    \\[-2pt]
      -2n          & if $s > 0 \wedge i > n$ \\
    \end{cases*}$}}
 &  \gape{\makecell[l]{
     \scalebox{0.7}{Encodes simple loop with memory leak. Sol:}\\%
     \scalebox{0.7}{$\begin{cases*}
      1      & if $s = 0$                   \\[-2pt]
      -i-n+1 & if $s > 0 \wedge i\leq n$    \\[-2pt]
      -2n    & if $s > 0 \wedge i > n$ \\
    \end{cases*}$}}}
 \\ \cline{2-5}
 & \texttt{\multiphaseonename}
 & \texttt{\multiphaseonekey}  
 & \gape{\scalebox{0.7}{$f(i,n,r)=\begin{cases*}
      0            & if $i \geq n$          \\[-2pt]
      1+f(0,n,r-1) & if $i < n \land r > 0$ \\[-2pt]
      1+f(i+1,n,r) & if $i < n \land r = 0$ \\
    \end{cases*}$}}
 & \gape{\scalebox{0.7}{$\begin{cases*}
      0   & if $i \geq n$          \\[-2pt]
      n+r & if $i < n \land r > 0$ \\[-2pt]
      n-i & if $i < n \land r = 0$ \\
    \end{cases*}$}}
 \\ \cline{1-5}
\multirow{4}{*}{\rotatebox[origin=c]{90}{\texttt{\catdivandconq}}}
 & \texttt{\binsearchname}
 & \texttt{\binsearchkey} 
 & \gape{\scalebox{0.7}{$f(n) =
    \begin{cases*}
      f\Big(\Big\lfloor\frac{n}{2}\Big\rfloor\Big) + 1 & if $n > 0$ \\
      0                                                & if $n = 0$ \\
    \end{cases*}$}}
 & \gape{\scalebox{0.7}{$\begin{cases*}
      0                                 & if $n = 0$ \\[-2pt]
      1+\big\lfloor\log_2(n)\big\rfloor & if $n > 0$
    \end{cases*}$}}
 \\ \cline{2-5}
 & \texttt{\qsortonename}
 & \texttt{\qsortonekey}   
 & \gape{\scalebox{0.7}{$f(n) =
    \begin{cases*}
      f\Big(\Big\lfloor\frac{n-1}{2}\Big\rfloor\Big) + f\Big(\Big\lceil\frac{n-1}{2}\Big\rceil\Big) + n & if $n > 1$ \\
      1                                                                                                 & if $n \leq 1$ \\
    \end{cases*}$}}
 & \gape{\makecell[l]{
     \scalebox{0.7}{Hard to express exactly, but}\\
     \scalebox{0.7}{$\sim n\log_2(n)$}
     }}
 \\ \cline{2-5}
 & \texttt{\divandconqmultioneonename}
 & \texttt{\divandconqmultioneonekey} 
 & \gape{\scalebox{0.7}{$f(n) =
    \begin{cases*}
      2 f\Big(\Big\lfloor\frac{n}{3}\Big\rfloor\Big) + 5 f\Big(\Big\lfloor\frac{n}{6}\Big\rfloor\Big)+ n & if $n > 0$ \\
      1                                                & if $n = 0$ \\
    \end{cases*}$}}
 & \gape{\makecell[l]{
     \scalebox{0.7}{Hard to express exactly, but}\\
     \scalebox{0.7}{$n^{1.2764...}$}
     }}
 \\ \cline{1-5}
\multirow{3.3}{*}{\rotatebox[origin=c]{90}{\texttt{\catsuperpoly}}}
 & \texttt{\exponename}
 & \texttt{\exponekey}  
 & \gape{\scalebox{0.7}{$f(n) =
    \begin{cases*}
      2f(n-1) & if $n > 0$ \\[-2pt]
      3       & if $n = 0$ \\
    \end{cases*}$}}
 & $3\times 2^n$
 \\ \cline{2-5}
 & \texttt{\exptwoname}
 & \texttt{\exptwokey} 
 & \gape{\scalebox{0.7}{$f(n) =
    \begin{cases*}
      2f(n-1)+1 & if $n > 0$ \\[-2pt]
      3         & if $n = 0$ \\
    \end{cases*}$}}
 & $4\times 2^n - 1$
 \\ \cline{2-5}
 & \texttt{\factname}
 & \texttt{\factkey} 
 & \gape{\scalebox{0.7}{$f(n) =
    \begin{cases*}
      nf(n-1) & if $n > 0$ \\[-2pt]
      1       & if $n = 0$ \\
    \end{cases*}$}}
 & $n!$
 \\ \cline{1-5}\pagebreak
\multirow{9}{*}{\rotatebox[origin=c]{90}{\texttt{\catmisc}}}
 & \texttt{\basiconename}
 & \texttt{\basiconekey}  
 & \gape{\scalebox{0.7}{$f(n) =
    \begin{cases*}
      f(n-1)+1 & if $n > 0$ \\[-2pt]
      0        & if $n = 0$ \\
    \end{cases*}$}}
 & $n$
 \\ \cline{2-5}
 & \texttt{\sumoscname}
 & \texttt{\sumosckey}  
 & \gape{\scalebox{0.7}{$f(x,y) =
    \begin{cases*}
      f(x-1,y)   + 1 & if $x > 0 \land y > 0$ \\[-2pt]
      f(x+1,y-1) + y & if $x = 0 \land y > 0$ \\[-2pt]
      1              & if $y = 0$ \\
    \end{cases*}$}}
 & \gape{\scalebox{0.7}{$\begin{cases*}
      1                                 & if $y = 0$ \\[-2pt]
      \frac{1}{2}y^2 + \frac{3}{2}y + x & if $y > 0$ \\[2pt]
    \end{cases*}$}}
 \\ \cline{2-5}
 & \texttt{\divname}
 & \texttt{\divkey} 
 & \gape{\scalebox{0.7}{$f(x,y) =
    \begin{cases*}
      f(x-y,y) + 1 & if $x \geq y \land y > 0$ \\[-2pt]
      0            & if $x < y    \land y > 0$   \\
    \end{cases*}$}}
 & $\Big\lfloor{\frac{x}{y}}\Big\rfloor$
 \\ \cline{2-5}
 & \texttt{\nestedname}
 & \texttt{\nestedkey} 
 & \gape{\scalebox{0.7}{$f(n) =
    \begin{cases*}
      f(f(n-1)) + 1 & if $n > 0$ \\[-2pt]
      0             & if $n = 0$ \\
    \end{cases*}$}}
 & $n$
 \\ \cline{2-5}
 & \texttt{\openzipname}
 & \texttt{\openzipkey}  
 & \gape{\scalebox{0.7}{$f(x,y) =
    \begin{cases*}
      f(x-1,y-1) + 1  & if $x > 0 \land y > 0$ \\[-2pt]
      f(x,y-1)   + 1  & if $x = 0 \land y > 0$ \\[-2pt]
      f(x-1,y)   + 1  & if $x > 0 \land y = 0$ \\[-2pt]
      0               & if $x = 0 \land y = 0$ \\
    \end{cases*}$}}
 & $\max(x,y)$
 \\ \cline{2-5}
 & \texttt{\mergename}
 & \texttt{\mergekey} 
 & \gape{\makecell[l]{$f(x,y)=$\\
   \scalebox{0.7}{$\begin{cases*}
      \max(f(x-1,y), f(x,y-1)) + 1 & if $x > 0 \land y > 0$ \\[-2pt]
      0                            & if $x = 0 \lor  y = 0$ \\
    \end{cases*}$}}}
 & \gape{\scalebox{0.7}{$\begin{cases*}
      x+y-1 & if $x>0 \land y>0$ \\[-2pt]
      0     & if $x=0 \lor  y=0$ \\
    \end{cases*}$}}
 \\ \cline{2-5}
 & \gape{\makecell[l]{\texttt{\loopustarjanname}
 }} 
 & \gape{\makecell[l]{\texttt{\loopustarjankey}}}
 & \gape{\makecell[l]{$f(i,j)=$\\
   \scalebox{0.7}{$\begin{cases*}
      0                                & if $i=0$ \\[-2pt]
      1+\max(f(i-1,j+1), 1+f(i-1,j+2)) & if $i>0$
    \end{cases*}$}}}
 & \gape{\makecell[l]{
     Inlined single-function encoding\\
     of tarjan benchmark from~\cite{Sinn17-loopus,mlrec-tplp2024-preprint}\\
     $2i$ (worst-case)}}
 \\ \cline{1-5}
\end{longtable}}
\begin{table}
\centering
\footnotesize
 \setlength{\extrarowheight}{2.5pt}%
\begin{tabular}{| c | c | c |}
\hline
\texttt{\prototypename}
& \texttt{\ramlname}
& \texttt{\cofloconame}
\\ \hline
\gape{\makecell[l]{
    Upper bound:\\
    \scalebox{0.7}{$\begin{cases*}
      \frac{1}{2}n^2+\frac{701}{202}n+\frac{30000}{101}\\
           \qquad \text{ if } n>0 \wedge c\geq 100\\[-2pt]
      \frac{1}{2}n^2+\frac{701}{202}n+\frac{300}{101}c\\
           \qquad \text{ if } n>0 \wedge c <   100\\[-2pt]
      c \qquad\!\!\! \text{ if } n=0\\
    \end{cases*}$}\\
    Lower bound:\\
    \scalebox{0.7}{$\begin{cases*}
      \frac{1}{2}n^2+\frac{701}{202}n+100\\
           \qquad \text{ if } n>0 \wedge c\geq 100\\[-2pt]
      \frac{1}{2}n^2+\frac{701}{202}n+\frac{300}{101}c-\frac{19900}{101}\\
           \qquad \text{ if } n>0 \wedge c <   100\\[-2pt]
      c \qquad\!\!\! \text{ if } n=0\\
    \end{cases*}$}\\
    }}
&
\gape{\makecell[l]{
    Upper bound:\\
    $\frac{1}{2}n^2+\frac{601}{2}n+c$\\
    \phantom{.}\\
    Lower bound:\\
    $\frac{1}{2}n^2+\frac{1}{2}n$\\
    }}
&
\gape{\makecell[l]{
    Upper bound:\\
    \scalebox{0.7}{$\begin{cases*}
      2n^2+300n+100 & if $n\geq 1$\\[-2pt]
      c             & if $n=0$\\
    \end{cases*}$}\\
    \phantom{.}\\
    Lower bound:\\
    \scalebox{0.7}{$\begin{cases*}
      n & if $n\geq 1$\\[-2pt]
      c                     & if $n=0$\\
    \end{cases*}$}\\
    }}
\\\hline
  \texttt{\linregonlyname}
& \texttt{\srname}~\cite{mlrec-tplp2024-preprint}
& \texttt{\mathematicaname}
\\\hline
\gape{\scalebox{0.7}{$\begin{cases*}
      \frac{1}{2}n^2 + \frac{17747}{5166}n + \frac{1874420}{8911}\\
      \qquad \text{ if } n>0 \wedge c\geq 100\\[-2pt]
      \frac{1}{2}n^2 + \frac{12169}{3503}n + \frac{7429}{2570}c -\frac{112277}{1182}\\
      \qquad \text{ if } n>0 \wedge c <   100\\[-2pt]
      c \qquad\!\!\! \text{ if } n=0\\
    \end{cases*}$}}
&
\gape{\scalebox{0.7}{$\begin{cases*}
      \texttt{0.49929*n**2 + 3.4857*n + 222.34}\\
      \qquad \text{ if } n>0 \wedge c\geq 100\\[-2pt]
      \texttt{0.49960*n**2 + 3.6287*n + 2*c    - 66.422}\\
      \qquad \text{ if } n>0 \wedge c <   100\\[-2pt]
      \texttt{c} \qquad\!\!\! \text{ if } n=0\\
    \end{cases*}$}}
&
\texttt{None}
\\\hline
\end{tabular}
\vspace{1em}
\caption{Full results obtained by each tool for the running example.
  Results for the symbolic regression presented in~\cite{mlrec-tplp2024-preprint}
  is included for completeness: it does not manage to discover the
  exact solution either, even if floor and division nodes are allowed.
  \label{table:runnex_results}
}
\end{table}
 \clearpage
\newpage

\section{Domains of $\boundaries$-bounds}
\label{sec:A-bounds}

In this appendix, we introduce an abstract domain framework,
which proposes to obtain abstractions, i.e. safe approximations, of
general functions in $\dd\to\rr$ by ``simple'' functions from a
selected set $\boundaries\subseteq\dd\to\rri$ that we call
\emph{boundary functions}.
This concept can be motivated from at least three points of view.
\begin{enumerate}
  \item Simplification of analysis output, to report it to the user in
    an interpretable format. This is explored in technical
    report~\cite{Bagnara-q344} as part of the \texttt{PURRS}
    project~\cite{BagnaraPZZ05}, but has not been pursued to the best of
    the authors' knowledge.
  \item Simplification of analysis output, to be usable by
    further analyses. In our context, this is particularly relevant to
    decide inequality of functions, e.g. to reduce functions to
    classes supported by CAS. This would provide a unified framework
    to works such as~\cite{Albert2015-cost-func-comp,resource-verification-tplp18-shortest}.
  \item Finally, this can be considered more literally, and used
    directly to abstract lattices of (templates of) candidates to
    improve the search of pre/postfixpoints, e.g. by abstract
    iteration as mentioned in
    Section~\ref{subsubsec:abstract-iteration}.
\end{enumerate}

\noindent To achieve these goals, we may try to use elements of
$\boundaries$ to \emph{under and overapproximate} general functions,
in an intuition similar to \emph{flowpipes} used for differential
equations~\cite{GoubaultHSCC17-short,GoubaultCAV18}.
Intuitively, we would thus obtain abstract elements in
$(\boundaries\times\boundaries,\geq\!\times\!\leq)\simeq\boundaries^{op}\times\boundaries$,
which would abstract abstract ``intervals whose endpoints are
function'' (which appear concretely in cost analysis as described
Section~\ref{sec:context}), with ``concretisation''
{\footnotesize
\begin{equation}\label{eq:concretisation-singleboundary}
\begin{aligned}
  \boundaries\times\boundaries &\to (\dd\to\intervals(\rr))\\
  (f_{lb}^\sharp,\,f_{ub}^\sharp) &\mapsto
     \Big(\vn\mapsto\big[f_{lb}^\sharp(\vn),\,f_{ub}^\sharp(\vn)\big]\Big).
\end{aligned}
\end{equation}}

Unfortunately, in general, this cannot be turned into a Galois
connection, and $\boundaries^{op}\times\boundaries$ is not a complete
lattice. This simply comes from the fact that, in general, for
$b_1,b_2\in\boundaries$, there is no best approximation in
$\boundaries$ of $\max(b_1,b_2)$ and $\min(b_1,b_2)$, as illustrated
in Fig.~\ref{fig:A-bounds-nogalois} for $\boundaries$ the set of
affine functions in $\nn\to\rr$.
We \emph{do} however recover a Galois connection if we apply a power
lifting and instead consider abstractions by \emph{sets} of
(conjunctions of) $\boundaries$-bounds: the concretisation $\gammaB$
below admits a left adjoint $\alphaB$, which extracts the set of
\emph{all} $\boundaries$-bounds of a function.

\begin{proposition}[$\boundaries$-bounds abstraction]
\label{prop:A-bounds-abstr}
   The following is a Galois connection.
   Abstract elements can be interpreted as \emph{conjunctions} of
   $\boundaries$-bounds.
{\footnotesize
\begin{align*}
  (\dd\to\intervals(\rr),\overset{.}{\sqsubseteq}_\intervals)
  &\galois{\alphaB}{\gammaB}
   (\pp(\boundaries),\supseteq)\times(\pp(\boundaries),\supseteq)
  \\
  \left(\vn\mapsto\left[\max_{f_{lb}^\sharp\in F_{lb}^\sharp}f_{lb}^\sharp(\vn),\;
       \min_{f_{ub}^\sharp\in F_{ub}^\sharp}f_{ub}^\sharp(\vn)\right]\right)
  &\longmapsfrom
  (F_{lb}^\sharp,\,F_{ub}^\sharp)
  \\
  (f_{lb},\,f_{ub})
  &\longmapsto
  \left(
    \begin{array}{l}
      \big\{f_{lb}^\sharp\in\boundaries\,\big|\,
              \forall\vn\in\dd,\,f_{lb}^\sharp(\vn)\leq f_{lb}(\vn)\big\},\\[2pt]
      \big\{f_{ub}^\sharp\in\boundaries\,\big|\,
               \forall\vn\in\dd,\,f_{ub}^\sharp(\vn)\geq f_{ub}(\vn)\big\}
    \end{array}
  \right)
\end{align*}}
\end{proposition}

\begin{figure}
  \begin{center}
  \begin{tabular}{p{5cm} p{5cm}}
  {\begin{tikzpicture}[scale=0.5]
    \begin{scope}[transparency group]
        \begin{scope}[blend mode=multiply]
  \draw[very thin,color=gray] (-0.1,-0.1) grid (4.4,4.4);
  \draw[->] (-0.2,0) -- (4.5,0) node[right] {\scriptsize $n$};
  \draw[->] (0,-0.2) -- (0,4.5) node[above] {\scriptsize $f(n)$};
  \foreach \x in {0,...,4} {
        \node [anchor=north] at (\x,-0.5-1ex) {\tiny $\x$};
    }
    \foreach \y in {0,...,4} {
        \node [anchor=east] at (-0.5-0.5em,\y) {\tiny $\y$};
    }
    \draw[-, very thick, blue] (0, 0) -- (4.4,4.4);
    \draw[-, very thick, blue] (0, 1) -- (4.4,1);
    \fill[fill=blue!20] (0,0)--(1,1)--(4.4,1)--(4.4,4.4)--(0,4.4);
        \end{scope}
    \end{scope}
  \end{tikzpicture}}
  &
  {\begin{tikzpicture}[scale=0.5]
    \begin{scope}[transparency group]
        \begin{scope}[blend mode=multiply]
  \draw[very thin,color=gray] (-0.1,-0.1) grid (4.4,4.4);
  \draw[->] (-0.2,0) -- (4.5,0) node[right] {\scriptsize $n$};
  \draw[->] (0,-0.2) -- (0,4.5) node[above] {\scriptsize $f(n)$};
  \foreach \x in {0,...,4} {
        \node [anchor=north] at (\x,-0.5-1ex) {\tiny $\x$};
    }
    \foreach \y in {0,...,4} {
        \node [anchor=east] at (-0.5-0.5em,\y) {\tiny $\y$};
    }
    \draw[-, very thick, black] (0,1) -- (1,1) -- (4.4,4.4);
    \draw[-, very thick, blue] (0,1) -- (3.4,4.4);
    \fill[fill=blue!20] (0,1) -- (3.4,4.4) -- (0,4.4);
        \end{scope}
    \end{scope}
  \end{tikzpicture}}
  \end{tabular}
  \end{center}
  \caption{\footnotesize(Left) $\min(n, 1)$ has no
    best affine upper bound, but is fully represented by a conjunction
    of two bounds.
    (Right) $\max(n,1)$ cannot be represented exactly as a
    conjunction of affine upper bounds.
  \label{fig:A-bounds-nogalois}}
\end{figure}
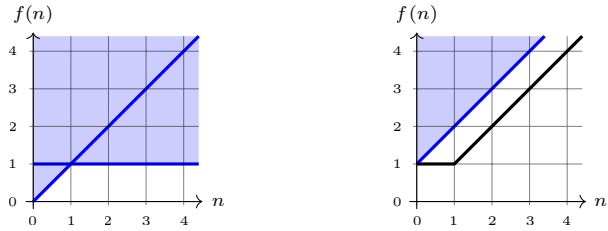

With the concretisation above in mind, we call
$(\pp(\boundaries),\supseteq)\times(\pp(\boundaries),\supseteq)$ the
\emph{domain of $\boundaries$-bounds}. We talk about the \emph{domain
of $\boundaries$-lbs} and \emph{domain of $\boundaries$-ubs} when
extracting the first or second component of the tuple respectively.
In a way, the discussion above shows that the \emph{ideal} solution to
the problem of finding the best $\boundaries$-bounds is to simply
consider \emph{all} $\boundaries$-bounds, or more efficiently an
antichain generating them. In other words, while we cannot optimally
abstract an interval whose endpoints are general functions by
intervals whose endpoints are in $\boundaries$, it can be done for
intervals whose endpoints are \emph{piecewise} in $\boundaries$, or
more precisely maxima and minima of elements in $\boundaries$.

Indeed, it is sufficient to focus our attention on elements in the
image of the closure $\gammaB\circ\alphaB$. All of these elements are
upsets and downsets, and can in many cases be represented by their
extremal elements. In favourable cases, extremal elements are even in
finite number, which allows to finitely represent elements of this
abstract domain.

\begin{remark}
   Upsets and downs are sets closed by taking upper and lower bounds
   respectively.
   We write $\uparrow\;:S\mapsto\cup_{s\in S}\{f\in\boundaries\,|s\leq f\}$
   the $(\pp-)$upper closure operator, and similarly for $\downarrow$.
   For an upset $\uparrow\mathrm{U}\subseteq\pp(\boundaries)$ (or
   symmetrically for downsets), we say that $\mathrm{U}$ is a set of
   \emph{extremal} constraints when it minimal among all $\mathrm{U}'$
   such that $\uparrow\mathrm{U}=\;\uparrow\mathrm{U}'$, which is the
   case if and only if $\mathrm{U}$ is an antichain.
   Extremal constraints are represented in dark blue in
   Fig.~\ref{fig:A-bounds-nogalois}.
   The Galois connection of Proposition~\ref{prop:A-bounds-abstr} may
   equivalently be rephrased as a connection between
   $(\dd\to\intervals(\rr),\overset{.}{\sqsubseteq}_\intervals)$ and
   $(\pp_\downarrow(\boundaries),\supseteq)\times(\pp_\uparrow(\boundaries),\supseteq)$,
   where $\pp_\downarrow(\boundaries)$ and $\pp_\uparrow(\boundaries)$
   respectively refer to the sets of downsets and upsets.
\end{remark}

With this finite representation in mind, we can conveniently compute
\emph{sound} (but non-optimal) transfer functions for operators on
functions, by focusing computing the transfer $\boundaries\to\pp_\uparrow(\boundaries)$
for each extremal element, and passing to the union.
We can introduce normalisation operations to keep small the number of
generators in a representation of each upsets.
To make this practical, we can force representations to stay bounded,
by discarding some generators if necessary (this is a sound
overapproximation, similar to heuristics used for operating on sets of
polyhedra in~\cite{milanese24}).
Choosing to work with only one bound at a time then appears as the
coarsest instance of this (non-optimal) abstraction, using singletons
(i.e. principal upsets $\uparrow\,\{f\}$).

A very simple instantiation of these ideas is to work with
$\boundaries$ being a set of affine bounds. In this case, the domain
is subsumed by polyhedra, and is an abstraction of it, as illustrated
in Fig.~\ref{fig:polyhedra-to-fun}.
However, the approach can be extended beyond affine relationship, e.g.
to piecewise polynomials or arithmetico-geometric relationships.
This leads to interesting questions for computing joins, meets, and
transfer functions in this context (connecting, among others,
with~\cite{Urban-phd}). We note that, in our intuition and experience,
these computations appear easier to perform for such domains of
(functional) bounds than for classical (relational) numerical domains.
We interpret this observation as the consequence of a generic Galois
connection between such domains. We illustrate this connection in the
example of polyhedra and affine upper bounds, whose abstraction
factorises through the following, which may be viewed as an
abstraction from \emph{arbitrary} relations to only relations defined
by (functional) upper bounds.
\begin{proposition}[Relational to functional constraint abstractions -- ub]
  \label{prop:rel2fun-cnstr}
  Let $X$ by any set, $(Y,\leq)$ be a lattice, and $X\to Y$ be equipped with
  the pointwise order.
  The following is a Galois connection.
  {\footnotesize
  \begin{align*}
  (\pp(X\times Y),\subseteq)
  &\galois{\alphafun}{\gammafun}
  (\pp_\uparrow(X\to Y),\supseteq)
  \\
  S &\longmapsto
   \Big\{f\,\Big|\,\forall x \in X, f(x) \geq \sup_{(x,y)\in S}y \Big\}
  \\
  \Big\{(x,y)\,\Big|\,\forall y \in A,\, y\leq f(x)\Big\}
  &\longmapsfrom A
\end{align*}}
\end{proposition}

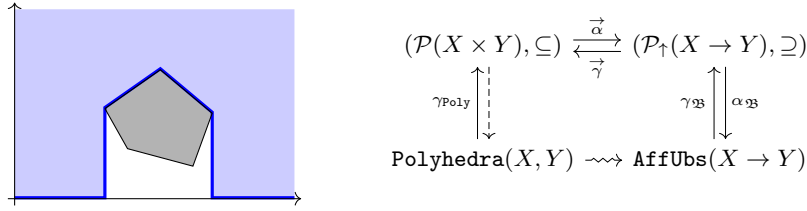
\begin{figure}[t]
  \vspace{-1em}
  \centering
  \begin{minipage}[c]{0.4\linewidth}
    {\begin{tikzpicture}[scale=0.3]
    \begin{scope}[transparency group]
        \begin{scope}[blend mode=multiply]
  \coordinate (zero) at (0,0);
  \coordinate (R) at (12.4,0);
  \coordinate (T) at (0,8.4);
  \coordinate (p1) at (4,4);

  \draw[->] ($(zero)+(-0.3,0)$) -- ($(R)+(0.3,0)$) node[right] {};
  \draw[->] ($(zero)+(0,-0.3)$) -- ($(T)+(0,0.3)$) node[above] {};
  \path[fill=gray!60,draw=black] (p1)
       -- ++(35:3) coordinate(p2)
       -- ++(-40:3) coordinate(p3)
       -- ++(-110:2.5) coordinate(p4)
       -- ++(165:3) coordinate(p5)
       -- (p1);
  \path[fill=blue!20]
    (zero)--(zero-|p1)--(p1)--(p2)--(p3)--(zero-|p3)--(R)--($(R)+(T)$)--(T)--(zero);
  \coordinate (ybias) at (0,1.2pt); %
  \draw[-, very thick, blue]
    ($(zero)+(ybias)$)--($(zero-|p1)+(ybias)$)--($(p1)+(ybias)$)--($(p2)+(ybias)$)--($(p3)+(ybias)$)--($(zero-|p3)+(ybias)$)--($(R)+(ybias)$);
        \end{scope}
    \end{scope}
  \end{tikzpicture}}
  \end{minipage}
  \begin{minipage}[c]{0.55\linewidth}
    {%
      \begin{tikzcd}[row sep=3em,column sep=1.5em]
        (\pp(X\times Y),\subseteq)
        \arrow[d, dashrightarrow, shift left]
        \arrow[r, "\alphafun", shift left]
        &
        (\pp_\uparrow(X\to Y),\supseteq)
        \arrow[d, "\alphaB", shift left]
        \arrow[l, "\gammafun", shift left]
        \\
        \texttt{Polyhedra}(X,Y)
        \arrow[u, "\gamma_{\texttt{Poly}}", shift left]
        \arrow[r, rightsquigarrow]
        &
        \texttt{AffUbs}(X\to Y)
        \arrow[u, "\gammaB", shift left]
      \end{tikzcd}
    }
  \end{minipage}
  \caption{(Left) An illustration of the precision loss when representing a
    polyhedron with online (affine) upper bounds. (Right) Schematic
    representation of the factorisation of this abstraction
    ($\rightsquigarrow$) through that ($\alphafun$) of
    Proposition~\ref{prop:rel2fun-cnstr}.
    \label{fig:polyhedra-to-fun}}
  \vspace{-2em}
\end{figure}

Finally, we can mention that, to realise the first two points
discussed at the the beginning of this section, we can consider
computing explicitly sound abstractions
$(\dd\to\rr)\to\pp_\uparrow(\boundaries)$ themselves, rather than
simply transfer functions, at least for input functions given in a
selected expression syntax.
This is well studied problems for particular $\boundaries$ (i.e.
linearisation of expressions, best abstractions by polynomials, etc.),
but is more difficulty to handle uniformly in general.

We could consider to \emph{compute such abstraction by rewriting},
using rewriting rules of the form used
in~\cite{Albert2015-cost-func-comp,resource-verification-tplp18-shortest}
or more generally in CAS.
We can note that more precise rewriting rules can be designed if we
allow them to require hypotheses to be triggered, e.g.
(anti)monotonicity of a subexpression. For this, an interesting
direction could be to design simple abstract domains operating on
numerical expressions, e.g. to infer (piecewise) (anti)monotonicity.

\end{document}